\pdfoutput=1
\documentclass[12pt]{article}
\usepackage[margin=1in]{geometry}
\usepackage{setspace}
\usepackage{amsmath,amssymb,amsthm,bm,booktabs,array}
\usepackage{graphicx}
\usepackage[authoryear,round]{natbib}
\usepackage{microtype}
\usepackage{xurl}
\usepackage[hidelinks]{hyperref}

\theoremstyle{plain}
\newtheorem{theorem}{Theorem}
\newtheorem{proposition}{Proposition}
\newtheorem{lemma}{Lemma}
\newtheorem{corollary}{Corollary}
\newtheorem{assumption}{Assumption}
\theoremstyle{definition}
\newtheorem{definition}{Definition}
\theoremstyle{remark}
\newtheorem{remark}{Remark}

\newcommand{\IWTheta}{4}
\newcommand{\IWMuA}{1.251}
\newcommand{\IWLamTildeA}{1.089}

\newcommand{\IWMuB}{1.079}

\newcommand{\IWMuC}{1.292}

\newcommand{\IWSMR}{0.573}
\newcommand{\IWPMR}{0.108}
\newcommand{\IWPrivTech}{0.183}
\newcommand{\IWOS}{$-$0.076}
\newcommand{\IWOSPrice}{0.0032}
\newcommand{\IWOSWage}{$-$0.0787}
\newcommand{\IWGap}{0.067}
\newcommand{\IWGapOwn}{0.046}
\newcommand{\IWGapDown}{0.021}
\newcommand{\IWRivals}{0.282}
\newcommand{\IWAE}{0.041}
\newcommand{\IWHarberger}{0.151}
\newcommand{\IWGapsAll}{0.110}
\newcommand{\IWWedge}{0.465}

\newcommand{\IWXRDRatio}{1.238}

\newcommand{\IWXBSProfitTwo}{0.087}
\newcommand{\IWXBSMarkupTwo}{$-$0.048}
\newcommand{\IWXBSTransferTwo}{0.134}
\newcommand{\IWXBSProfitThree}{$-$0.005}
\newcommand{\IWXBSMarkupThree}{$-$0.011}
\newcommand{\IWXBSTransferThree}{0.006}

\newcommand{\IWXDomShareLow}{0.316}
\newcommand{\IWXDomShareHigh}{0.628}
\newcommand{\IWXDomAELow}{0.116}
\newcommand{\IWXDomAEHigh}{$-$0.098}
\newcommand{\IWXDomNestMeanHigh}{0.114}

\newcommand{\IWXDomOSLow}{0.217}
\newcommand{\IWXDomOSHigh}{0.178}
\newcommand{\IWXDomOSShareLow}{0.28}
\newcommand{\IWXDomOSShareHigh}{0.37}

\newcommand{\IWXDomRowOne}{0.316 & 1.828 & 0.463 & 0.453 & 0.116 & 0.217 & 1.460 \\}
\newcommand{\IWXDomRowTwo}{0.465 & 2.338 & 0.870 & 0.572 & 0.077 & 0.214 & 1.678 \\}
\newcommand{\IWXDomRowThree}{0.571 & 2.914 & 1.332 & 0.657 & $-$0.025 & 0.194 & 1.816 \\}
\newcommand{\IWXDomRowFour}{0.628 & 3.363 & 1.691 & 0.703 & $-$0.098 & 0.178 & 1.897 \\}

\newcommand{\DYSMR}{1.725}
\newcommand{\DYPMR}{0.949}
\newcommand{\DYWedge}{0.776}
\newcommand{\DYStatic}{0.655}
\newcommand{\DYOwn}{0.100}
\newcommand{\DYRivSocial}{0.021}
\newcommand{\DYAE}{0.121}
\newcommand{\DYGapOwn}{0.211}
\newcommand{\DYRivKnow}{0.175}
\newcommand{\DYRivTech}{0.140}
\newcommand{\DYMinusOS}{0.008}
\newcommand{\DYRivStake}{0.0003}
\newcommand{\DYRatio}{1.82}
\newcommand{\DYStaticPerPMR}{0.69}
\newcommand{\DYResponsesPerPMR}{0.13}
\newcommand{\DYSubsidy}{45}
\newcommand{\DYGrowth}{0.0152}

\newcommand{\IWAppYear}{2015}
\newcommand{\IWAppThetaBase}{16}
\newcommand{\IWAppThetaAlt}{8}
\newcommand{\IWAppSigmaFinal}{0.9}
\newcommand{\IWAppEpsilonVA}{0.5}
\newcommand{\IWAppThetaInt}{0.2}
\newcommand{\IWAppEtaKL}{0.99}
\newcommand{\IWAppPhiR}{0.042}
\newcommand{\IWAppPhiT}{0.186}

\newcommand{\IWAppDepreciation}{0.15}
\newcommand{\IWAppPreGrowth}{0.05}
\newcommand{\IWAppWindowStart}{2013}
\newcommand{\IWAppWindowEnd}{2017}
\newcommand{\IWAppMinRecords}{10}

\newcommand{\IWAppFullCoverage}{18}
\newcommand{\IWAppPatentingFirms}{1{,}017}

\newcommand{\IWAppCPCGroups}{4{,}640}
\newcommand{\IWAppBFGainUserCost}{13}
\newcommand{\IWAppBFGainAccounting}{11}
\newcommand{\IWAppBFGainProduction}{25}

\newcommand{\IWAppTPExcessUserCost}{14}
\newcommand{\IWAppExcludedUserCost}{11}

\newcommand{\IWAppTPExcessAccounting}{26}
\newcommand{\IWAppExcludedAccounting}{5}
\newcommand{\IWAppTPProduction}{1.79}
\newcommand{\IWAppTPProductionAll}{1.59}
\newcommand{\IWAppDownShareProduction}{43}
\newcommand{\IWAppExcludedProduction}{30}
\newcommand{\IWAppDownShareLow}{41}
\newcommand{\IWAppDownShareHigh}{48}

\newcommand{\IWAppPassOwn}{0.62}

\newcommand{\IWAppBSVSymmetric}{5.4}
\newcommand{\IWAppFactorPremium}{1.20}

\newcommand{\IWAppClipUserCost}{4}
\newcommand{\IWAppClipAccounting}{13}
\newcommand{\IWAppClipBase}{12}
\newcommand{\IWAppClipBaseOutput}{17}
\newcommand{\IWAppClipAlt}{28}
\newcommand{\IWAppClipAltOutput}{30}
\newcommand{\IWAppForeignFirms}{744}
\newcommand{\IWAppForeignSales}{43}

\newcommand{\IWAppSlopeUser}{\ensuremath{-}0.48}

\newcommand{\IWAppSlopeCournot}{\ensuremath{-}0.94}
\newcommand{\IWAppSlopeFirms}{3{,}031}

\newcommand{\IWAppSlopeUserWeighted}{\ensuremath{-}0.04}
\newcommand{\IWAppEtaFlat}{10.2}

\newcommand{\IWAppPassEtaFlat}{0.97}
\newcommand{\IWAppIndAENegUserCost}{88}
\newcommand{\IWAppIndAENegAccounting}{89}
\newcommand{\IWAppIndAENegCournot}{88}

\newcommand{\IWAppFactorPremiumAlt}{1.37}
\newcommand{\IWAppFactorPremiumAccounting}{1.20}

\newcommand{\IWAppBSVCreditBSV}{0.99}

\newcommand{\IWAppFullCoverageSales}{74}
\newcommand{\IWAppFullCoverageRD}{87}
\newcommand{\IWAppFullCoverageRatioMedian}{1.9}

\newcommand{\IWAppSegTheta}{13.8}
\newcommand{\IWAppSegEta}{1.7}
\newcommand{\IWAppSegPremium}{1.20}

\newcommand{\IWAppSegPassOwn}{0.62}
\newcommand{\IWAppSegPassSmall}{0.92}
\newcommand{\IWAppSegPassLarge}{0.54}
\newcommand{\IWAppSegMarkets}{533}
\newcommand{\IWAppSegFringeless}{42}
\newcommand{\IWAppSegFringelessSales}{36}
\newcommand{\IWAppSegClip}{9}
\newcommand{\IWAppSegFirms}{3{,}193}
\newcommand{\IWAppSegNoUSFirms}{276}
\newcommand{\IWAppSegNoUSSales}{11}

\newcommand{\IWAppSegMultiFirms}{962}
\newcommand{\IWAppSegDomesticShare}{49}

\newcommand{\IWAppSegDropped}{14}
\newcommand{\IWAppSegFactorNA}{0.87}
\newcommand{\IWAppSegFactorUSCA}{0.92}
\newcommand{\IWAppSegFactorAM}{0.73}

\newcommand{\IWAppSegRatio}{1.14}
\newcommand{\IWAppSegGap}{0.27}

\newcommand{\IWAppSegMinusIP}{\ensuremath{-}0.17}

\newcommand{\IWAppSegAEPTen}{\ensuremath{-}0.19}
\newcommand{\IWAppSegAEPNinety}{0.28}
\newcommand{\IWAppSegAELarge}{56}

\newcommand{\IWAppSegRivalShareMedian}{0.76}
\newcommand{\IWAppSegSpRatio}{6.24}

\newcommand{\IWAppSegBSVRatio}{1.11}
\newcommand{\IWAppSegBSVRatioBSVCreditPrice}{2.62}
\newcommand{\IWAppSegBSVNoSteal}{5.01}
\newcommand{\IWAppSegBSVCreditModel}{3.50}

\newcommand{\IWAppSegRDNegShare}{80}

\newcommand{\IWAppSegRDGapNet}{0.11}
\newcommand{\IWAppSegRDAE}{\ensuremath{-}0.04}
\newcommand{\IWAppSegRDIPCost}{0.19}
\newcommand{\IWAppSegRDPosShareSp}{99.7}
\newcommand{\IWAppSegRDNegShareSmall}{100}

\newcommand{\IWAppSegRDAESmall}{\ensuremath{-}0.34}

\newcommand{\IWAppSegRDShortfallSmall}{0.31}

\newcommand{\IWAppSegIPMedianSizeA}{0.005}

\newcommand{\IWAppSegSmallFringeAE}{\ensuremath{-}0.11}

\newcommand{\IWAppSegSmallNoFringeAE}{\ensuremath{-}0.57}

\newcommand{\IWAppSegSalesSizeA}{8}
\newcommand{\IWAppSegRDSizeA}{15}

\newcommand{\IWAppSegAESizeA}{\ensuremath{-}0.16}
\newcommand{\IWAppSegIPSizeA}{0.026}
\newcommand{\IWAppSegRatioSizeA}{1.01}

\newcommand{\IWAppSegRatioSizeD}{1.39}

\newcommand{\IWAppSegBertPassOwn}{0.87}

\newcommand{\IWAppSegBertAESizeA}{\ensuremath{-}0.09}
\newcommand{\IWAppSegBertIPSizeA}{\ensuremath{-}0.001}

\newcommand{\IWAppSegBertRDNegShareSmall}{99}
\newcommand{\IWAppSegBertRDShortfallSmall}{0.08}

\newcommand{\IWAppSegLeftOut}{115}
\newcommand{\IWAppSegLeftOutSales}{0.9}

\newcommand{\IWAppRDWRatioSeg}{1.08}

\newcommand{\IWAppRDWRatioSegBert}{1.24}

\newcommand{\IWAppFESmallAESegFringe}{\ensuremath{-}0.08}
\newcommand{\IWAppFESmallIPSegFringe}{0.035}
\newcommand{\IWAppFESmallShortfallSegFringe}{0.09}
\newcommand{\IWAppFESmallRDNegSegFringe}{100}

\newcommand{\IWAppFirmShareSmallSeg}{62}

\newcommand{\IWAppMainMinAELarge}{30}
\newcommand{\IWAppMainMaxAELarge}{68}

\newcommand{\IWAppMainMinRivalShareMedian}{0.72}
\newcommand{\IWAppMainMaxRivalShareMedian}{0.94}

\newcommand{\IWAppFixedMinRDWAEOffset}{8}
\newcommand{\IWAppFixedMaxRDWAEOffset}{35}

\newcommand{\IWAppSegUnitAESizeA}{\ensuremath{-}0.09}

\newcommand{\IWAppRDWAESegUnitVA}{0.02}

\newcommand{\IWAppRDWAESegUnit}{0.05}

\newcommand{\IWAppRDWAESegFree}{0.09}

\newcommand{\IWAppSegSpillFirms}{831}
\newcommand{\IWAppSegSpillRDShare}{59}
\newcommand{\IWAppSegBSVCreditBSVPrice}{0.92}
\newcommand{\IWAppSegBSVEqualMean}{4.8}
\newcommand{\IWAppSegBSVPrivateExcess}{370}
\newcommand{\IWAppSegBSVSteal}{3.66}
\newcommand{\IWAppSegBSVOS}{0.13}
\newcommand{\IWAppSegBSVOwnMarkup}{0.18}
\newcommand{\IWAppSegBSVSocial}{0.84}
\newcommand{\IWAppSegBSVDown}{0.11}
\newcommand{\IWAppSegBSVAE}{0.05}
\newcommand{\IWAppSegBSVWelfareOverOutputPrice}{2.39}
\newcommand{\IWAppSegBSVFactorCreditPrice}{1.92}
\newcommand{\IWAppSegBSVFactorOwnOutput}{1.05}
\newcommand{\IWAppSegBSVFactorSocial}{1.19}
\newcommand{\IWAppSegSpRatioNoIP}{7.17}
\newcommand{\IWAppSegSpIPMultiplePrice}{11.7}
\newcommand{\IWAppSegBSVRatioFirmRivals}{1.55}

\newcommand{\IWAppSegAETailShare}{0.03}
\newcommand{\IWAppSegAETailShareUpper}{0.40}
\newcommand{\IWAppRDWAEFringeSeg}{0.02}

\newcommand{\IWAppRDWAEFringeSegBert}{\ensuremath{-}0.04}

\newcommand{\IWAppRDShareNoFringeSeg}{54}
\newcommand{\IWAppSalesShareNoFringeSeg}{43}
\newcommand{\IWAppSegSpillIPSizeA}{\ensuremath{-}0.0033}
\newcommand{\IWAppSegBertSpillIPSizeA}{\ensuremath{-}0.0037}
\newcommand{\IWAppSegSpillIPSizeC}{\ensuremath{-}0.082}
\newcommand{\IWAppSegBertSpillIPSizeC}{\ensuremath{-}0.119}
\newcommand{\IWAppSegSpillSSizeA}{0.002}
\newcommand{\IWAppSegSpillSSizeC}{0.097}
\newcommand{\IWAppPassSmallSalesCut}{25}

\newcommand{\IWAppFiveSalesSeg}{34}
\newcommand{\IWAppFiveRDSeg}{40}
\newcommand{\IWAppFiveAESeg}{\ensuremath{-}0.11}

\newcommand{\IWAppFiveRDNegSeg}{97}

\newcommand{\IWAppFiveAESegBert}{\ensuremath{-}0.09}

\newcommand{\IWAppFiveRDNegSegBert}{91}

\newcommand{\IWAppSegRDGapNetSmall}{0.05}

\newcommand{\IWAppSegUCRDWAE}{0.01}

\newcommand{\IWAppSegMSSlopeUser}{\ensuremath{-}0.37}
\newcommand{\IWAppSegMSSlopeUserWeighted}{\ensuremath{-}0.17}
\newcommand{\IWAppSegMSSlopeAccounting}{\ensuremath{-}0.40}
\newcommand{\IWAppSegMSSlopeAccountingWeighted}{\ensuremath{-}0.21}

\newcommand{\IWAppSegMSSigmaUserWeighted}{4.1}
\newcommand{\IWAppSegMSSigmaAccounting}{2.1}

\newcommand{\IWAppSegMSCorrUserWeighted}{0.20}
\newcommand{\IWAppSegMSCorrAccounting}{0.09}

\newcommand{\IWAppRobSegCMinSmallAE}{\ensuremath{-}0.19}
\newcommand{\IWAppRobSegCMaxSmallAE}{\ensuremath{-}0.11}
\newcommand{\IWAppRobSegCMinSmallIP}{0.007}
\newcommand{\IWAppRobSegCMaxSmallIP}{0.047}

\newcommand{\SMR}{\operatorname{SMR}}
\newcommand{\PMR}{\operatorname{PMR}}
\newcommand{\SMC}{\operatorname{SMC}}
\newcommand{\AEff}{\mathrm{AE}}
\newcommand{\IP}{\mathrm{IP}}
\newcommand{\TR}{\mathrm{TR}}
\newcommand{\T}{\mathcal{T}}
\newcommand{\Pv}{\mathcal{P}}
\newcommand{\mc}{\mathit{mc}}
\newcommand{\dlog}{d\log}
\newcommand{\Cov}{\operatorname{Cov}}

\newif\ifjournal
\journaltrue
\makeatletter
\newcommand{\onlineappendix}{%
	\setcounter{section}{0}%
	\renewcommand{\thesection}{OA.\arabic{section}}%
	\renewcommand{\theHsection}{OA.\arabic{section}}%
	\@for\oa@c:=equation,table,figure,theorem,proposition,lemma,corollary,assumption,definition,remark\do{%
		\setcounter{\oa@c}{0}%
		\expandafter\edef\csname the\oa@c\endcsname{OA.\noexpand\arabic{\oa@c}}%
		\expandafter\edef\csname theH\oa@c\endcsname{OA.\noexpand\arabic{\oa@c}}}%
}
\makeatother

\title{Allocative Efficiency and the Returns to Innovation}
\author{Koki Okumura\thanks{The Chinese University of Hong Kong. Email: \href{mailto:kokiokumura@cuhk.edu.hk}{kokiokumura@cuhk.edu.hk}.}}
\date{October 2026}

\begin{document}

\maketitle

\begin{abstract}
	\noindent Innovation changes welfare by improving technology and reallocating inputs across producers with different markups. When producers choose markups optimally, an innovator's own output and price adjustments have no private value at the margin, so reallocation affects its profit only through others' responses. For an innovator small in its markets, the change in allocative efficiency is of the same order as its cost saving but does not reach its profit. A process innovator is paid its cost saving, forgoing its own and downstream markups on added output. For U.S. public firms in \IWAppYear{}, these forgone markups average \IWAppTPExcessAccounting{} percent of that saving. For innovators accounting for \IWAppSegAELarge{} percent of sales, allocative efficiency rises or falls by more than a tenth of their private return. With estimated spillovers, patenting firms' R\&D has a social return \IWAppSegSpRatio{} times its private return.
\end{abstract}

\section{Introduction}\label{sec:intro}

An innovation changes welfare in two ways. It raises the output that the economy obtains from a given allocation of goods and factors, and it moves goods and factors between producers. When prices equal marginal costs, the second channel has no first-order effect on welfare \citep{Hulten1978}. With markups it does, because a factor earns its price but produces output worth that price times the markup of the producer that employs it. \citet{BaqaeeFarhi2020} call the two channels the technology effect and the change in allocative efficiency. Markups are large and dispersed across firms \citep{EdmondMidriganXu2023}, and in Baqaee and Farhi's estimates the reallocation of market share toward high-markup firms accounts for about half of aggregate TFP growth in the United States over 1997--2015. They treat productivities and markups as primitives. Technology changes, however, come from innovation that firms choose for profit, and the markups that make reallocation matter come from the same firms' optimization. This paper asks which parts of the social return to an innovation the innovator is paid for when markups arise from producers' own optimization. We focus on process innovations, which raise the innovator's productivity, and Section~\ref{sec:main} also treats improvements in quality.

We study an economy like that of Baqaee and Farhi, with arbitrary demand, production functions, input--output linkages, and factors, but with markups chosen optimally by producers. We show that an innovator's own output and price adjustments have no first-order private value. Reallocation therefore affects its profit only through other agents' responses. The welfare effect of this reallocation is the change in allocative efficiency. For an innovator that is small in its market and perceives how its output affects its price, other agents' responses to its technology change affect its profit only at second order in its size. This holds under both quantity and price competition. In nested-CES economies, the change in allocative efficiency is of first order in the innovator's size except in special cases. When producers compete in quantities, the responses reach the innovator's profit only as price changes, which are transfers, so the innovator receives none of the change in allocative efficiency, whatever its size and whatever the sign of the change. The innovator is paid its cost saving at its current output rather than the technology effect of its innovation, so it also forgoes its own markup and the markups accumulated downstream on the output that the innovation adds at a fixed allocation. The technology changes that its innovation causes at other producers are external as well. Because the change in allocative efficiency can be positive or negative, the externality it creates can raise or lower the social return relative to the private return.

A simple case shows why. Suppose that labor, in fixed supply, is the only input and that a producer that is small in a market of differentiated goods lowers its cost. It passes the cost reduction through to its price, and its expansion draws workers from its rivals. At either producer a worker earns the wage and produces output that sells for the wage times that producer's markup, so moving a worker to the innovator changes allocative efficiency by the wage times the difference between the innovator's markup and the rival's. When markups rise with market share, the small innovator's markup is below its rivals', so each worker that it draws from them lowers allocative efficiency. The innovator gains nothing from hiring the worker at the margin. It has chosen its output so that marginal revenue equals marginal cost, so the markup it earns on the worker's output is offset by the fall in the price it receives on the units it already sells, a fall that benefits its buyers. Its profit also changes as its rivals contract and the wage responds, but because the innovator is small, the value of these responses at its quantities is of second order in its size, whereas the change in allocative efficiency, like the innovator's cost saving, is of first order.

The argument requires only that producers set their markups optimally, as under Cournot oligopoly with any demand system or monopolistic competition. The innovator's output and input mix are at its optimum, so their adjustment to the innovation has no first-order private value, and every other agent's adjustment reaches its profit only through what we call the induced price term, the effect on its profit of the price changes that it does not attribute to its own output or quality. This is the argument of \citet{GreenwaldStiglitz1986}, applied to an innovator that sets its own markup. It fails when markups are exogenous ratios of price to marginal cost, the reduced form of Baqaee and Farhi's analysis, because the innovator's output is then not at its optimum.

Theorem~\ref{thm:main} states the result through two benchmarks. The technology effect, $\T_k$, is society's valuation of producer $k$'s technology change at a fixed allocation. The private technology value, $\Pv_i$, is the innovator's valuation of its own technology change at fixed output and input prices. For a process innovation, this private value is its cost saving. The change in welfare $dW$ and the change in the innovator's profit $d\pi_i$ therefore satisfy
\[
	dW-d\pi_i=(\T_i-\Pv_i)+\sum_{j\neq i}\T_j+\AEff-\IP_i ,
\]
where $\AEff$ is the change in allocative efficiency and $\IP_i$ the induced price term. We call $\T_i-\Pv_i$ the appropriability gap. At the innovator's current output, the innovation saves inputs, which are worth their prices. At a fixed allocation, it raises the innovator's output, which flows to its buyers and raises their output in turn, and consumers value the resulting final output at its prices. Society's valuation is $\mu_i\tilde\lambda_i/\lambda_i$ times the innovator's, its markup times a factor that exceeds one only through the markups downstream of it. The gap is nonnegative and can be measured with markups and input--output tables alone, whereas the change in allocative efficiency, before an innovation occurs, can be predicted only with a model of conduct. Under quantity competition the induced price term consists of transfers between the innovator and its buyers and suppliers. Under price competition it also contains the innovator's margin on the demand that rivals' price responses divert to it \citep{WeylFabinger2013}. For a given set of observed price and quantity changes, the measured induced price term is the same under either form of competition. The result for small innovators holds under both.

The change in allocative efficiency is the only external term without a fixed sign, and Section~\ref{sec:primitives} shows when it is positive, building on Baqaee and Farhi's propagation formulas with Cournot markups as in \citet{AtkesonBurstein2008}. Within a market, if the other producers' markups do not change, an innovation by a producer whose markup exceeds the average raises allocative efficiency by drawing inputs from its rivals, which it does exactly when the elasticity of substitution times the pass-through of its cost reduction exceeds one. Markups that rise with market share lower the pass-through and dampen the gain, and across markets they can reverse it. Along supply chains, a productivity gain at a marked-up supplier lowers allocative efficiency when its buyers' value added and intermediate inputs are complements, because the buyers then shift their spending toward value added, which carries no markup.

Section~\ref{sec:application} measures the terms of the wedge for innovations by U.S. public firms in \IWAppYear{}, in the economy that \citet{BaqaeeFarhi2020} calibrate to input--output tables and Compustat, with markups set by competition among the business segments of these firms within six-digit NAICS markets. With Baqaee and Farhi's user-cost and accounting-profit markups, the technology effect of an innovation exceeds the innovator's private technology value by \IWAppTPExcessUserCost{} and \IWAppTPExcessAccounting{} percent on average, and markups downstream of the innovator account for \IWAppDownShareLow{} to \IWAppDownShareHigh{} percent of the excess. The appropriability gap can be measured using these markups, whereas the other terms use the markups that conduct implies, because the private return rests on producers' first-order conditions. The change in allocative efficiency has either sign across innovations and exceeds a tenth of the private return in absolute value for innovators that account for \IWAppSegAELarge{} percent of sales when firms compete in quantities, and for \IWAppMainMinAELarge{} to \IWAppMainMaxAELarge{} percent depending on conduct and the definition of markets.

Figure~\ref{fig:size} shows the separation of Theorem~\ref{thm:main} by the innovator's share of its market. Relative to the private return, the induced price term vanishes as market share falls, whereas the change in allocative efficiency does not. For innovators with less than 5 percent of their market, which account for \IWAppFiveSalesSeg{} percent of sales and \IWAppFiveRDSeg{} percent of R\&D, it is negative whether firms compete in quantities or in prices.

Theorem~\ref{thm:main} changes how business stealing and the social returns to R\&D should be interpreted. Estimates of the social return treat the output that an innovator takes from its rivals as a transfer \citep{Griliches1979,BloomSchankermanVanReenen2013}. \citet{BloomSchankermanVanReenen2013} attribute half of the change in rivals' market values to lost output, and with the same share \citet{LuckingBloomVanReenen2019} estimate a marginal social return to R\&D of 57.7 percent against a private return of 13.6 percent. In Theorem~\ref{thm:main}, the innovator's gain from taking business is part of its induced price term, and the welfare cost to rivals whose technology does not change is their output loss valued at their markups \citep{MankiwWhinston1986}, which is part of the change in allocative efficiency. Treating business stealing as a transfer therefore omits the change in allocative efficiency. Rivals' profits, and so their market values, move by this welfare cost plus transfers of either sign, so no fixed share of their change recovers it, and they can misstate it even in sign (Proposition~\ref{prop:bs}). In the U.S. economy, rivals' lost output valued at their markups accounts for \IWAppSegRivalShareMedian{} of the fall in their profits at the sales-weighted median innovation, not one half. With the technology spillovers that \citet{BloomSchankermanVanReenen2013} estimate, the social return is \IWAppSegSpRatio{} times the private return, whereas the output-based accounting that they use, applied to the same economy with their estimate of business stealing, gives \IWAppSegBSVRatioBSVCreditPrice{}, close to the ratio of their estimated returns of 55.0 and 20.7 percent. The difference comes mainly from that accounting's private return, which credits the innovator with the output that its rivals lose, valued at their prices.

The same separation holds for the R\&D that an innovation induces. In a Markov game of R\&D (Section~\ref{sec:subsidy}), an innovation changes firms' subsequent R\&D decisions. The innovator's own change in R\&D has no first-order value to it, because it already chooses its R\&D optimally, as it does its output. A rival's change in R\&D raises the innovator's profits through the knowledge that spills over to it and lowers them through the business that the rival takes. These R\&D responses change welfare only through the responding firms' R\&D wedges, the gaps between their social and private marginal returns, as reallocating inputs changes welfare only through markups, so the reallocation of R\&D across firms is external to the innovator, as allocative efficiency is.

For policy, the firm's first-order condition equates its private marginal return to its marginal cost of R\&D, so when that cost is also society's, the R\&D subsidy that closes the wedge at the current equilibrium depends on the two marginal returns only through their ratio. The subsidy's allocative-efficiency component has no fixed sign, so it does not by itself justify subsidizing R\&D. When R\&D draws factors from production in proportion to their employment and this causes no reallocation, as in Section~\ref{sec:application}, a dollar of it costs society the aggregate markup, the mean of the producers' ratios $\mu_k\tilde\lambda_k/\lambda_k$ weighted by their payments to factors, \IWAppSegPremium{} dollars in the U.S. economy. Net of this cost, the appropriability gap favors R\&D exactly for innovators whose ratio exceeds the aggregate markup, and at the mean the two cancel, as uniform production subsidies have no first-order effect in \citet{Liu2019}. Without technology spillovers, the static return to the R\&D of innovators with less than 5 percent of their market, whose change in allocative efficiency is negative, falls short of its social cost for \IWAppFiveRDNegSeg{} percent of it with quantity competition and \IWAppFiveRDNegSegBert{} percent with price competition. At the firm's first-order condition, existing R\&D tax credits lower the private marginal return required to cover a dollar of R\&D spending and thus widen this shortfall. With the technology spillovers of \citet{BloomSchankermanVanReenen2013}, by contrast, the social return exceeds that social cost for \IWAppSegRDPosShareSp{} percent of the R\&D of the firms that generate them, so the case for subsidizing the R\&D of small innovators rests on spillovers.

\subsection{Related literature}\label{sec:literature}

When markups are chosen, the change in allocative efficiency reaches an innovator's profit only through other agents' responses, and Theorem~\ref{thm:main} separates it from the appropriability gaps of all productivity changes. Partial-equilibrium models of strategic R\&D in industrial organization \citep{BranderSpencer1983,Suzumura1992,LeahyNeary1997,FarrellShapiro1990,FarrellShapiroAER1990} combine the two in the Harberger term $\sum_k(p_k-\mc_k)\,dy_k$ (Section~\ref{sec:bridge}), which also appears in the argument of \citet{GreenwaldStiglitz1986} and the identity of \citet{Harberger1971} and \citet{BaqaeeBurstein2026}, and which \citet{BaqaeeFarhi2020} distinguish from allocative efficiency in the decompositions of \citet{PetrinLevinsohn2012} and \citet{BasuFernald2002}. The responses to a small innovator's cost reduction affect its profit only at second order in its size, under quantity or price competition and for any smooth aggregator of differentiated goods (Corollary~\ref{cor:small}, Proposition~\ref{prop:smallgeneral}). In the homogeneous-good Cournot markets of \citet{LahiriOno1988} and \citet{Zhao2001}, a small enough firm's marginal cost cut lowers welfare by shifting output away from rivals with higher margins, the force by which small innovators lower allocative efficiency here, but its induced price term does not vanish. \citet{Katz1986} and \citet{KatzOrdover1990}, and \citet{MukherjeeEtAl2026} for new products, split the wedge by who gains or loses.

Theorem~\ref{thm:main} provides a counterpart for incumbents' innovation to the finding of \citet{BaqaeeFarhiSangani2024} that entrants' profits leave out the welfare gains from reallocation toward high-markup firms. Here, producers choose markups rather than taking them as given \citep{BaqaeeFarhi2020}, and reallocation can raise or lower welfare. The appropriability gap compares cost with the cost-based Domar weight, as do analyses of entry \citep{BaqaeeFarhiEntry}, technology adoption \citep{BueraTrachter2024,AlvarezBueraTrachter2026}, and production subsidies \citep{Liu2019}, so the innovator's share of its technology effect falls with its own and downstream markups, as in \citet{AdhamiBrouilletteRockall2025} for a monopolistically competitive firm. Unlike the new varieties and captured product lines of \citet{Acemoglu2023} and \citet{DeRidder2024}, whose private value rises with the markup, the process innovations studied here include incumbents' improvements of their own products. Such improvements are the main source of U.S. productivity growth \citep{GarciaMaciaHsiehKlenow2019} and are productivity gains when quality multiplies quantity. Incumbents' innovation in \citet{Peters2020} is the zero-pass-through case of Section~\ref{sec:onemarket}, and in \citet{AghionEtAl2025} markups attract R\&D whatever their source.

Theorem~\ref{thm:dynamic} decomposes the gap between the social and private marginal returns to a firm's R\&D in a Markov game, whereas structural studies of dynamic oligopoly compare welfare totals or decompose private incentives \citep{GoettlerGordon2011,Igami2017,BesankoDoraszelskiKryukov2019,CavenaileCelikTian} and studies of knowledge and production networks characterize the allocation of R\&D \citep{LiuMa2021,KonigLiuZenou2019,KoikeOkumura2026}. Each induced R\&D response enters welfare through the responding firm's R\&D wedge and the innovator's return through the value of that R\&D to it, as rivals' strategies enter the costate equations of feedback Nash equilibria \citep{DocknerEtAl2000} and later choices enter generalized Euler equations \citep{KrusellKuruscuSmith2002,KleinKrusellRiosRull2008}.

Section~\ref{sec:environment} sets up the economy, Section~\ref{sec:main} states and proves the main result, and Section~\ref{sec:primitives} characterizes the sign of the change in allocative efficiency. Section~\ref{sec:measured} relates the result to the Harberger term, business stealing, and estimated returns to R\&D, Section~\ref{sec:policy} draws the implications for R\&D policy, and Section~\ref{sec:application} measures the terms of the wedge for U.S. public firms. Section~\ref{sec:conclusion} concludes. The Appendix proves the main results, and the Online Appendix contains the R\&D game, the other proofs, extensions, example economies, and details of the application.

\section{Environment}\label{sec:environment}

The economy follows \citet{BaqaeeFarhi2020}, with arbitrary production functions, input--output linkages, and factors, except that markups arise from producers' optimization rather than from exogenous wedges and that final demand and factor supplies are more general.

\paragraph{Production and markets.} There are $N$ goods and $F$ primary factors. Good $k\in\{1,\dots,N\}$ is produced by producer $k$ according to
\[
	y_k=A_k\,G_k(\bm x_k,\bm l_k),
\]
where $\bm x_k=(x_{km})_m$ are its intermediate inputs, $\bm l_k=(l_{kf})_f$ are its factor inputs, $G_k$ is differentiable, and $A_k$ is Hicks-neutral productivity. Some producers may be competitive aggregators of other goods, such as the CES nests of Section~\ref{sec:primitives}. Markets clear,
\[
	y_m=c_m+\sum_kx_{km},
	\qquad
	L_f=\sum_kl_{kf},
\]
where $c_m$ is final consumption of good $m$ and $L_f$ is the supply of factor $f$. Good $m$ sells at price $p_m$, and factor $f$ at price $w_f$.

\paragraph{Household.} Prices and welfare are measured in a numeraire whose marginal utility to the household is one, and producers maximize profit in this numeraire, a normalization that is part of their conduct \citep{GabszewiczVial1972}. Welfare is
\[
	W=U(\bm c;\bm b)-\sum_fV_f(L_f),
\]
where $U$ is differentiable and increasing in $\bm c$, $b_m$ is the quality of good $m$, which enters only $U$ and only for goods sold exclusively to final consumers, and $V_f$ is the disutility of supplying factor $f$. The household takes prices as given, so $p_m=U_{c_m}$ for every consumed good and $w_f=V_f'(L_f)$ for every factor in elastic supply, and profits are rebated lump sum. A factor in inelastic supply has $dL_f=0$.

\begin{assumption}[Price-taking input demand]
	\label{ass:inputs}
	Each producer takes input prices $\bm p$ and factor prices $\bm w$ as given and minimizes cost. Its cost function is $C_k(y_k;\bm p,\bm w,A_k)$, its marginal cost is $\mc_k=\partial C_k/\partial y_k$, and
	\[
		A_k\frac{\partial G_k}{\partial x_{km}}=\frac{p_m}{\mc_k},
		\qquad
		A_k\frac{\partial G_k}{\partial l_{kf}}=\frac{w_f}{\mc_k}.
	\]
\end{assumption}

\begin{assumption}[Optimizing producers]
	\label{ass:conduct}
	Producer $k$ chooses $y_k$ to maximize $\pi_k=p_ky_k-C_k$, taking all input and factor prices as given. It perceives the inverse demand $p_k=\varphi_k(y_k,b_k)$, the price at which its buyers, including competitive aggregators and downstream producers, purchase $y_k$ when the outputs and qualities of the other producers selling to them are held fixed. Perceived inverse demand is weakly decreasing, $\varphi_{k,y}=\partial\varphi_k/\partial y_k\le0$, with $\varphi_{k,y}=0$ for competitive producers, which take their price as given. The first-order condition is
	\[
		\varphi_k+y_k\,\varphi_{k,y}=\mc_k,
	\]
	so the markup $\mu_k\equiv p_k/\mc_k$ is at least one.
\end{assumption}

Assumption~\ref{ass:conduct} covers Cournot oligopoly with any demand system, monopolistic competition, and competitive producers. Remark~\ref{rem:bertrand} in Section~\ref{sec:main} treats price competition, and \ifjournal the working-paper version of this paper, available from the authors, treats\fi{} multi-product producers and monopsony.

\paragraph{Shares and Domar weights.} Final expenditure is $E=\sum_mp_mc_m$, with shares $\beta_m=p_mc_m/E$. The cost-based and revenue-based input shares of producer $k$ are
\[
	\widetilde\Omega_{km}=\frac{p_mx_{km}}{\mc_k\,y_k},
	\qquad
	\Omega_{km}=\frac{p_mx_{km}}{p_ky_k}=\frac{\widetilde\Omega_{km}}{\mu_k}.
\]
Factors are treated as additional nodes that use no inputs and are not consumed, so their rows of $\widetilde\Omega$ and $\Omega$ and their entries of $\beta$ are zero. In producer $k$'s row, the shares of factor $f$ are $\widetilde\Omega_{kf}=w_fl_{kf}/(\mc_ky_k)$ and $\Omega_{kf}=w_fl_{kf}/(p_ky_k)$. By Assumption~\ref{ass:inputs}, each cost-based share of producer $k$ is the elasticity of its output with respect to the corresponding input. With constant returns, $\mc_ky_k$ is producer $k$'s cost, and the cost-based shares are the shares of its inputs in that cost.

The Leontief inverses are
\[
	\widetilde\Psi=(I-\widetilde\Omega)^{-1}=I+\widetilde\Omega+\widetilde\Omega^2+\cdots,
	\qquad
	\Psi=(I-\Omega)^{-1}=I+\Omega+\Omega^2+\cdots,
\]
where the series are assumed to converge, as they do under constant returns when the supply chain of every good uses some factor. The cost-based and revenue-based Domar weights are $\tilde\lambda^\top=\beta^\top\widetilde\Psi$ and $\lambda^\top=\beta^\top\Psi$. Market clearing gives $p_ky_k=E\lambda_k$, so $\lambda_k$ is producer $k$'s sales relative to final expenditure. The cost-based weight $\tilde\lambda_k$ instead sums, over every chain of purchases that leads from a final good back to producer $k$, the expenditure share of that final good times the product of the cost-based input shares along the chain, which run through the producers downstream of $k$.

The factor columns $\widetilde\Psi_{(f)}$ and $\Psi_{(f)}$ of the Leontief inverses give the share of factor $f$ in a producer's cost and in its revenue along its whole supply chain, with $\widetilde\Psi_{ff}=\Psi_{ff}=1$. The income share of factor $f$ is $\Lambda_f=\beta^\top\Psi_{(f)}=w_fL_f/E$, and its cost-based share is $\widetilde\Lambda_f=\beta^\top\widetilde\Psi_{(f)}=\sum_k\tilde\lambda_k\widetilde\Omega_{kf}$.

\paragraph{Innovation.} The set of producers and goods is fixed, so innovation changes the technologies of existing producers of existing goods, and Remark~\ref{rem:entry} in Online Appendix~\ref{oa:extensions} treats free entry. Let $\tau$ collect the technologies $(A_k,b_k)$ of all producers. An innovation by producer $i$, the innovator, is a differential change $d\tau=(\dlog A_k,db_k)_k$, which also includes spillovers to other producers' technologies and the changes in technique that the innovation induces (Remark~\ref{rem:technique}). A change in $A_i$ with $db_i=0$ is a process innovation, and a change in $b_i$ with $\dlog A_i=0$ a quality innovation. Proposition~\ref{prop:general} in Online Appendix~\ref{oa:extensions} lets technology enter production functions and utility in any differentiable way. We assume that the equilibrium is locally a continuously differentiable function of $\tau$, as it is when the equilibrium conditions are continuously differentiable and their Jacobian is nonsingular, and $d$ denotes the resulting first-order change in an equilibrium object, along any such selection. The social return to the innovation is $dW$, the private return of innovator $i$ is $d\pi_i$, and we call their difference $dW-d\pi_i$ the wedge. The cost of the innovation is omitted from both returns until Section~\ref{sec:subsidy}.

\paragraph{Benchmarks.} Section~\ref{sec:main} decomposes the wedge by valuing technology changes at two benchmarks. Society's benchmark is a fixed allocation, at which every user, the household included, keeps its share of every good and factor, and factor supplies are held fixed. A producer's own benchmark holds its output and its input prices fixed at their current values.

\begin{definition}[Technology values]
	\label{def:values}
	For each producer $k$, the technology effect $\T_k$, the value of its technology change at a fixed allocation, and the private technology value $\Pv_k$, the value of the same change to producer $k$ at its current output and input prices, are
	\begin{align*}
		\T_k&\equiv E\tilde\lambda_k\,\dlog A_k+U_{b_k}\,db_k,\\
		\Pv_k&\equiv\mc_k\,y_k\,\dlog A_k+y_k\,\varphi_{k,b}\,db_k,
	\end{align*}
	where $U_{b_k}=\partial U/\partial b_k$ and $\varphi_{k,b}=\partial\varphi_k/\partial b_k$.
\end{definition}

Lemma~\ref{lem:allocation} shows that $\T_k$ is the technology effect of \citet{BaqaeeFarhi2020}. For a productivity gain, $\Pv_k$ is the cost saving $\mc_ky_k\,\dlog A_k$, which values at their prices the inputs that the producer saves at its current output, whereas $\T_k$ values at consumer prices the final consumption that its additional output yields at a fixed allocation, as the output flows to its buyers and raises their output in turn. The two differ unless neither producer $k$ nor any producer downstream of it charges a markup, and we call $\T_k-\Pv_k$ the appropriability gap of producer $k$'s technology change.
\section{Private and Social Returns}\label{sec:main}

This section compares each return with its value at the corresponding benchmark (Sections~\ref{sec:lemmas} and~\ref{sec:theorem}), characterizes the appropriability gap and the induced price term (Sections~\ref{sec:gap} and~\ref{sec:os}), and shows that for a small innovator the induced price term is of second order under quantity or price competition (Section~\ref{sec:small}).

\subsection{Two benchmarks}\label{sec:lemmas}

\begin{lemma}[Welfare at a fixed allocation]\label{lem:allocation}
	Let the allocation $\mathcal X$ collect the shares in which each good and factor is divided among its users, $\mathcal X_{km}=x_{km}/y_m$, $\mathcal X_{kf}=l_{kf}/L_f$, and $\mathcal X_{cm}=c_m/y_m$, and let $W(\tau,\mathcal X,\bm L)$ be the welfare attained with technologies $\tau$, allocation $\mathcal X$, and factor supplies $\bm L$. Under Assumption~\ref{ass:inputs}, at the equilibrium,
	\[
		\frac{\partial W}{\partial\log A_k}\Big|_{\mathcal X,\bm L}\dlog A_k+\frac{\partial W}{\partial b_k}\Big|_{\mathcal X,\bm L}db_k=\T_k
	\]
	for every producer $k$.
\end{lemma}

\begin{proof}
	At a fixed $\mathcal X$, $x_{km}=\mathcal X_{km}y_m$ and $l_{kf}=\mathcal X_{kf}L_f$ with $L_f$ fixed, so
	\[
		\dlog y_k=\dlog A_k+\sum_m\frac{\partial\log G_k}{\partial\log x_{km}}\,\dlog y_m .
	\]
	By Assumption~\ref{ass:inputs}, $\partial\log G_k/\partial\log x_{km}=A_k(\partial G_k/\partial x_{km})\,x_{km}/y_k=p_mx_{km}/(\mc_ky_k)=\widetilde\Omega_{km}$. Hence $\dlog\bm y=\widetilde\Psi\,\dlog\bm A$. Since $c_m=\mathcal X_{cm}y_m$, $\dlog c_m=\dlog y_m$, and with $p_m=U_{c_m}$ and fixed factor supplies,
	\[
		dW=\sum_mp_mc_m\,\dlog y_m+\sum_kU_{b_k}db_k=E\,\beta^\top\widetilde\Psi\,\dlog\bm A+\sum_kU_{b_k}db_k .
	\]
	Because $\beta^\top\widetilde\Psi=\tilde\lambda^\top$, the $k$th terms are $E\tilde\lambda_k\,\dlog A_k+U_{b_k}\,db_k=\T_k$.
\end{proof}

\begin{definition}[Change in allocative efficiency]\label{def:ae}
	For an innovation $d\tau$, the change in allocative efficiency is the welfare effect of the changes in the allocation and in factor supplies that the innovation causes,
	\[
		\AEff\equiv\frac{\partial W}{\partial\mathcal X}\,d\mathcal X+\sum_f\frac{\partial W}{\partial L_f}\,dL_f ,
	\]
	with the derivatives evaluated at the equilibrium. The second term is zero when factor supplies are fixed (Remark~\ref{rem:factorsupply} in Online Appendix~\ref{oa:extensions}).
\end{definition}

This is the change in allocative efficiency of \citet{BaqaeeFarhi2020}. For a process innovation in an economy in which $U$ is homothetic, technologies have constant returns, and factor supplies are fixed, their Theorem~1 expresses it through the changes in markups and factor income shares, $\AEff=-E\,(\tilde\lambda^\top\dlog\bm\mu+\widetilde\Lambda^\top\dlog\bm\Lambda)$. In an efficient economy $\AEff=0$ to first order, and when every markup is one, $\tilde\lambda=\lambda$, so the technology effect weights productivity changes by sales \citep{Hulten1978}.

\begin{corollary}[Allocative efficiency at markup wedges]\label{cor:aeformula}
	Under Assumption~\ref{ass:inputs},
	\begin{equation}
		\AEff=\sum_k(\mu_k-1)\,dI_k-E\sum_k(\tilde\lambda_k-\lambda_k)\,\dlog A_k ,
		\label{eq:ae}
	\end{equation}
	where $dI_k\equiv\sum_mp_m\,dx_{km}+\sum_fw_f\,dl_{kf}$ is the change in producer $k$'s input use valued at current prices.
\end{corollary}

The proof is in Appendix~\ref{app:proofs}. The first term of~\eqref{eq:ae} values each producer's change in input use at its markup wedge $\mu_k-1$, the excess of the value of the output that a marginal unit of spending on $k$'s inputs yields over its cost. The second term removes the contribution of input changes that would occur even at a fixed allocation. The change in allocative efficiency therefore reflects only departures from those benchmark input changes, valued at markup wedges.

\begin{lemma}[Private return]\label{lem:private}
	Under Assumptions~\ref{ass:inputs} and~\ref{ass:conduct},
	\[
		d\pi_i=\Pv_i+\IP_i,
		\qquad
		\IP_i\equiv y_i\,(dp_i)^{-i}-\sum_mx_{im}\,dp_m-\sum_fl_{if}\,dw_f,
	\]
	where $\IP_i$ is the induced price term and $(dp_i)^{-i}\equiv dp_i-\varphi_{i,y}\,dy_i-\varphi_{i,b}\,db_i$ is the part of the change in $p_i$ that producer $i$ does not attribute to its own output or quality.
\end{lemma}

\begin{proof}
	By Shephard's lemma, $dC_i=\mc_i\,dy_i-\mc_iy_i\,\dlog A_i+\sum_mx_{im}\,dp_m+\sum_fl_{if}\,dw_f$. Hence
	\[
		d\pi_i=y_i\,dp_i+(p_i-\mc_i)\,dy_i+\mc_iy_i\,\dlog A_i-\sum_mx_{im}\,dp_m-\sum_fl_{if}\,dw_f .
	\]
	Substituting $dp_i=\varphi_{i,y}\,dy_i+\varphi_{i,b}\,db_i+(dp_i)^{-i}$, the terms in $dy_i$ are $(p_i-\mc_i+y_i\varphi_{i,y})\,dy_i=0$ by Assumption~\ref{ass:conduct}.
\end{proof}

Lemma~\ref{lem:private} is an envelope result. Cost minimization gives the change in the innovator's input mix no first-order effect on its cost, and profit maximization along its perceived inverse demand makes the markup that it earns on an output change, $(p_i-\mc_i)\,dy_i$, exactly offset by the fall in revenue on its current output that it attributes to that change, $y_i\varphi_{i,y}\,dy_i$. The innovator's own adjustments, including the change in its markup that it attributes to its own output, therefore have no first-order private value, even though they can move resources between producers with different markups. Besides $\Pv_i$, the innovator receives only $\IP_i$, the changes in its input and factor prices and the part of the change in its own price that it does not attribute to its own output or quality, valued at its current quantities.

\subsection{The decomposition}\label{sec:theorem}

\begin{theorem}[Social and private returns to innovation]\label{thm:main}
	Under Assumptions~\ref{ass:inputs} and~\ref{ass:conduct}, for any innovation $d\tau$ with innovator $i$, the social return is
	\[
		dW=\sum_k\T_k+\AEff .
	\]
	The private return is
	\[
		d\pi_i=\Pv_i+\IP_i ,
	\]
	so the innovator's own changes in output and input mix have no first-order effect on its profit, and the rest of its profit change consists of price changes valued at its current quantities. Hence
	\begin{equation}
		dW-d\pi_i=(\T_i-\Pv_i)+\sum_{j\neq i}\T_j+\AEff-\IP_i ,
		\label{eq:main}
	\end{equation}
	and the change in allocative efficiency, of either sign, enters the wedge one for one.
\end{theorem}

\begin{proof}
	By the chain rule, $dW$ is the sum of $(\partial W/\partial\tau)\,d\tau$ and the two terms of Definition~\ref{def:ae}, and Lemma~\ref{lem:allocation} gives $(\partial W/\partial\tau)\,d\tau=\sum_k\T_k$. The private return is Lemma~\ref{lem:private}. In $\IP_i$, the change $(dp_i)^{-i}$ excludes the change in $p_i$ along the perceived inverse demand, through which alone a change in $y_i$ would enter, and by Shephard's lemma the input mix does not enter the cost.
\end{proof}

The two benchmarks separate the welfare effect of reallocation from what the innovator is paid, which other exact splits of the private return do not. By the first display in the proof of Lemma~\ref{lem:private}, without Assumption~\ref{ass:conduct},
\[
	d\pi_i=p_iy_i\,\dlog A_i+(\mu_i-1)\,dI_i+y_i\,dp_i-\sum_mx_{im}\,dp_m-\sum_fl_{if}\,dw_f ,
\]
where $(\mu_i-1)\,dI_i$ is the innovator's own term in~\eqref{eq:ae}, so the innovator appears to receive part of the change in allocative efficiency. Under Assumption~\ref{ass:conduct}, however, the change in its price includes the fall $\varphi_{i,y}\,dy_i$ that it attributes to its own output, and $(\mu_i-1)\,dI_i+y_i\varphi_{i,y}\,dy_i=-(p_i-\mc_i)\,y_i\,\dlog A_i$, which does not depend on how its input use changes. The innovator's own term in $\AEff$ is passed to its buyers through its price.

The first three terms of~\eqref{eq:main} are external. The innovator values its own technology change holding its output and its input prices fixed, and society values it at a fixed allocation, which gives the appropriability gap (Corollary~\ref{cor:gap}). Other producers' technology changes are valued socially at their technology effects, which the private return does not contain. Every first-order change in allocative efficiency caused by an innovation, of either sign, is an externality, because the innovator's own adjustments have no first-order private value and every other agent's adjustment reaches its profit only through $\IP_i$. The welfare effect of the rise in a large innovator's own markup after its productivity gain \citep{BaqaeeFarhi2020}, for example, is part of $\AEff$, whereas the part of that rise that other producers' responses cause raises the innovator's profit through $\IP_i$. The fourth term differs in kind. Under quantity competition, $\IP_i$ consists of price changes valued at the innovator's current quantities, transfers between the innovator and its buyers and suppliers that net out with one household \citep{DavilaKorinek2018}. In an open economy, the part paid by foreign buyers or received by foreign suppliers is a transfer between countries, which national welfare counts as a gain or loss. For a price-taking producer with constant returns, $d\pi_i=0$ and $\IP_i=-\Pv_i$, so the entire social return is external, the appropriability problem of \citet{Arrow1962} in its extreme form.

The argument is that of \citet{GreenwaldStiglitz1986}, in which the price changes that a small change causes are transfers and its welfare effect beyond its direct effect is the sum of existing distortions times the induced changes in the distorted quantities. With markups as the distortions, this sum is the Harberger term of Lemma~\ref{lem:accounting}. Theorem~\ref{thm:main} applies the argument to an innovator that perceives its own price effect and separates the Harberger term into the change in allocative efficiency and the appropriability gaps (Section~\ref{sec:bridge}). The envelope step fails when markups are exogenous wedges that producers keep as profit, the reduced form of \citet{BaqaeeFarhi2020}, because a producer's output is then not at its optimum (Remark~\ref{rem:wedges} in Online Appendix~\ref{oa:extensions}).

Although they are measured from different benchmarks, the terms of~\eqref{eq:main} can move together, and in the economy of Corollary~\ref{cor:fringe} the induced price term is proportional to the change in allocative efficiency. They differ, however, in sign and in what measuring them requires. For a productivity gain, the appropriability gap is nonnegative and requires only markups and the input--output network, and with constant returns the induced price term requires only the innovator's sales and markup, their changes, and its private technology value, whereas the change in allocative efficiency has no fixed sign and, for an innovation that has not occurred, depends on substitution and conduct.

\begin{remark}[Price competition]\label{rem:bertrand}
	Suppose instead that producer $k$ chooses its price. It takes as given all input and factor prices and the prices of the other producers selling to its buyers, and it perceives its buyers' demand $y_k=D_k(p_k,b_k)$, with partial derivatives $D_{k,p}<0$ and $D_{k,b}$. Its first-order condition is $y_k+(p_k-\mc_k)D_{k,p}=0$, so $p_k-\mc_k=-y_k/D_{k,p}>0$ and the markup exceeds one. The steps of Lemma~\ref{lem:private} give $d\pi_i=\Pv_i^B+\IP_i^B$, with
	\begin{align*}
		\Pv_i^B&=\mc_iy_i\,\dlog A_i+(p_i-\mc_i)\,D_{i,b}\,db_i,\\
		\IP_i^B&=(p_i-\mc_i)\,(dy_i)^{-i}-\sum_mx_{im}\,dp_m-\sum_fl_{if}\,dw_f,
	\end{align*}
	where $(dy_i)^{-i}\equiv dy_i-D_{i,p}\,dp_i-D_{i,b}\,db_i$ is the part of the change in $y_i$ that producer $i$ does not attribute to its own price or quality. Because Lemma~\ref{lem:allocation} does not use Assumption~\ref{ass:conduct}, Theorem~\ref{thm:main} holds with these terms (Online Appendix~\ref{oa:bertrand}). Cost minimization and the first-order condition for the price again give the innovator's own adjustments no first-order private value, but the induced price term now contains the margin on the output that rivals' price responses divert to the innovator (\citealp{Shapiro2012}; \citealp{WeylFabinger2013}). This margin is an output change weighted by the innovator's markup, not a price change, so the reading of $\IP_i$ as a pure transfer is specific to quantity competition. Under price competition this margin is part of the innovator's own term in~\eqref{eq:ae}, and the change in allocative efficiency is external only net of it. For an innovator that is small in its markets, however, $\IP_i$ and $\IP_i^B$ are both of second order in its size (Corollary~\ref{cor:small}).
\end{remark}

\subsection{The appropriability gap}\label{sec:gap}

\begin{corollary}[Appropriability gap]\label{cor:gap}
	Under Assumptions~\ref{ass:inputs} and~\ref{ass:conduct}, the appropriability gap is the sum of a production term and a Spence term,
	\[
		\T_i-\Pv_i=(E\tilde\lambda_i-\mc_iy_i)\,\dlog A_i+\nu_i\,db_i,
	\]
	where $\nu_i\equiv U_{b_i}-y_i\varphi_{i,b}$ is the quality distortion of \citet{Spence1975}. The coefficient of $\dlog A_i$ in the production term satisfies
	\[
		E\tilde\lambda_i-\mc_iy_i=(p_i-\mc_i)\,y_i+E\,(\tilde\lambda_i-\lambda_i)\ \ge\ 0,
	\]
	with both summands nonnegative. The second summand is zero when no producer downstream of $i$ charges a markup.
\end{corollary}

\begin{proof}
	The first expression follows from Definition~\ref{def:values}. Since $p_iy_i=E\lambda_i$, $E\tilde\lambda_i-\mc_iy_i=(p_i-\mc_i)\,y_i+E\,(\tilde\lambda_i-\lambda_i)$, and $p_i\ge\mc_i$. Because $\mu_k\ge1$, $0\le\Omega\le\widetilde\Omega$ elementwise, so $\Psi=I+\Omega+\Omega^2+\cdots\le I+\widetilde\Omega+\widetilde\Omega^2+\cdots=\widetilde\Psi$ and $\lambda\le\tilde\lambda$. A path from producer $i$ to final demand contributes to $\tilde\lambda_i-\lambda_i$ only if one of the buyers along it has $\mu>1$.
\end{proof}

In the production term, $(p_i-\mc_i)\,y_i\,\dlog A_i$ is the innovator's own markup on the output that the innovation adds at fixed inputs, and $E\,(\tilde\lambda_i-\lambda_i)\,\dlog A_i$ consists of the markups that the producers downstream of the innovator earn on the further output that they produce at a fixed allocation. Because $\mc_iy_i=p_iy_i/\mu_i=E\lambda_i/\mu_i$, the ratio of the technology effect of a process innovation to the innovator's private technology value is

\begin{equation}
	\frac{E\tilde\lambda_i}{\mc_iy_i}=\mu_i\,\frac{\tilde\lambda_i}{\lambda_i}.
	\label{eq:liu}
\end{equation}

The first factor, the innovator's own markup, keeps the gap positive whenever the innovator charges a markup, and the second exceeds one only through the markups downstream of it, so the gap is largest for producers far upstream of marked-up producers. In Cobb--Douglas economies, \citet{BigioLaO2020} show that $\tilde\lambda$ gives the elasticity of aggregate TFP with respect to each producer's productivity whatever the distortions. With one factor, and with markups represented as wedges on purchases, the distortion centrality of \citet{Liu2019} is the ratio in~\eqref{eq:liu} times labor's income share. His welfare results assume that the payments on wedges are dissipated, which removes the change in allocative efficiency.

\paragraph{A vertical chain.}\label{sec:vertical}
The vertical economy of \citet{BaqaeeFarhi2020} isolates the appropriability gap from the change in allocative efficiency. Producers have constant returns. Labor, in fixed supply, is used only by producer~1, whose good is the only input of producer~2, and so on up to producer~$N$, whose good is the only one consumed. The allocation is unique, so $\AEff=0$ for every innovation, and each $\tilde\lambda_k=1$. Producer $k$'s cost is $\mc_ky_k=E\prod_{j\ge k}\mu_j^{-1}$, because the sales of each producer are the cost of the producer it supplies, so the appropriability gap of a process innovation by $k$ is $\T_k-\Pv_k=E\bigl(1-\prod_{j\ge k}\mu_j^{-1}\bigr)\dlog A_k$, which rises with $k$'s own markup and with every markup downstream of it.

\paragraph{Quality innovation.} For a quality innovation the gap is $\nu_i\,db_i$, the value of the improvement to consumers at a fixed allocation less the revenue that the innovator gains from the price increase at its current output \citep{Spence1975}. It can have either sign, and it is zero when quality shifts inverse demand in parallel.

\subsection{The induced price term}\label{sec:os}

The induced price term collects every effect of the innovation on the innovator's profit other than its private technology value. The next result expresses it through the effect of all price changes on the innovator's profit and, under constant returns, through the innovator's sales and markup.

\begin{corollary}[Induced price term]\label{cor:os}
	For every producer $k$, let $\TR_k\equiv y_k\,dp_k-\sum_mx_{km}\,dp_m-\sum_fl_{kf}\,dw_f$ be the effect of price changes on its profit at its current quantities. Under Assumptions~\ref{ass:inputs} and~\ref{ass:conduct},
	\[
		\IP_i=\TR_i+(p_i-\mc_i)\,dy_i-y_i\,\varphi_{i,b}\,db_i .
	\]
	If producer $i$ has constant returns, then, whatever its conduct,
	\[
		d\pi_i=\Bigl(1-\frac{1}{\mu_i}\Bigr)d(p_iy_i)+\frac{p_iy_i}{\mu_i}\,\dlog\mu_i ,
	\]
	and $\IP_i=d\pi_i-\Pv_i$, which for a process innovation is $\IP_i=p_iy_i\bigl[(1-1/\mu_i)\,\dlog(p_iy_i)-(\dlog A_i-\dlog\mu_i)/\mu_i\bigr]$.
\end{corollary}

The proof is in Online Appendix~\ref{oa:otherproofs}. Under constant returns, the private return, and with it $\IP_i=d\pi_i-\Pv_i$, is measured by the changes in the innovator's sales and markup, without the first-order condition or any assumption on conduct. Conduct matters only for reading $d\pi_i-\Pv_i$ as the price changes that the innovator does not attribute to itself.

\subsection{Small innovators}\label{sec:small}

The induced price term is of first order for a price taker and for a dominant producer that takes business from its rivals (Table~\ref{tab:dominant}). The next result shows that for an innovator that is small in its market the term is of second order in its size under quantity or price competition, and that the change in allocative efficiency per unit of the innovator's sales has a finite limit, which Section~\ref{sec:primitives} shows to be nonzero in nested-CES economies except in special cases. It requires the innovator to sell to a CES aggregate of differentiated goods, such as an industry, and leaves the rest of the economy general.

\begin{corollary}[Small innovators]\label{cor:small}
	Under Assumption~\ref{ass:inputs}, and Assumption~\ref{ass:conduct} or Remark~\ref{rem:bertrand} for the producers other than $i$, let producer $i$ have constant returns and sell its whole output to a competitive producer $n$ with a constant-returns CES technology of elasticity $\theta_n$, in which $i$'s cost share is $s_i$, and consider a process innovation $\dlog A_i$ with no other technology change. Whatever $i$'s conduct,
	\begin{equation}
		\begin{split}
			\frac{d\pi_i-\Pv_i}{p_iy_i}={}&\Bigl[\theta_n\Bigl(1-\frac{1}{\mu_i}\Bigr)-1\Bigr](\dlog A_i-\dlog\mu_i)\\
			&+\Bigl(1-\frac{1}{\mu_i}\Bigr)\Bigl[(\theta_n-1)\Bigl(\dlog p_n-\sum_m\widetilde\Omega_{im}\,\dlog p_m\Bigr)+\dlog(p_ny_n)\Bigr],
		\end{split}
		\label{eq:small}
	\end{equation}
	where the sum runs over $i$'s inputs, factors included. Suppose further that $\theta_n>1$ and that $i$ competes with the other suppliers of $n$ in quantities, as in Assumption~\ref{ass:conduct}, or in prices, as in Remark~\ref{rem:bertrand}, perceiving the demand for $n$'s output to have an elasticity $\eta_n>0$, so that $\mu_i=\varepsilon_i/(\varepsilon_i-1)$ with $1/\varepsilon_i=(1-s_i)/\theta_n+s_i/\eta_n$ under quantity competition and $\varepsilon_i=\theta_n(1-s_i)+\eta_ns_i$ under price competition. Then $d\pi_i-\Pv_i$ is $\IP_i$ under quantity competition and $\IP_i^B$ under price competition, so for this process innovation the two take the same value given the observed changes in prices and quantities. Write $n$'s technology with the weight $\delta_i$ on good $i$, so that $s_i=\delta_i(p_i/p_n)^{1-\theta_n}$, as in Online Appendix~\ref{oa:dominant}, and let $\delta_i$ go to zero, with the other primitives fixed, along equilibria that converge to an equilibrium of the economy with $\delta_i=0$ at which the Jacobian of the equilibrium conditions, taken to be twice continuously differentiable, is nonsingular (in variables that remain regular at $\delta_i=0$, where this amounts to regularity of the economy without good $i$). Then $s_i\to0$, $\mu_i\to\theta_n/(\theta_n-1)$, and
	\begin{enumerate}
		\item $(d\pi_i-\Pv_i)/(p_iy_i)$ and $\dlog\mu_i$ are of order $s_i$, whereas $\Pv_i/(p_iy_i)=\dlog A_i/\mu_i$ tends to $(1-1/\theta_n)\,\dlog A_i$;
		\item $\AEff/(p_iy_i)$ converges to a finite limit, which is the same when $i$'s markup is held fixed.
	\end{enumerate}
\end{corollary}

The proof is in Appendix~\ref{app:proofs}. Equation~\eqref{eq:small} uses only constant returns and CES demand within the aggregate, so it holds for any conduct and whatever the rest of the economy. Its first term is the innovator's margin on the expansion that its own price change causes within the aggregate, net of the revenue that this price change costs it on its current output. The bracket is zero at the markup $\theta_n/(\theta_n-1)$, which a producer that perceives its own-price effect charges when it is small in its market, and it is $s_i(\theta_n/\eta_n-1)$ under quantity competition and $s_i(\theta_n-\eta_n)/\varepsilon_i$ under price competition. The second term values the changes in the aggregate's price relative to the innovator's input costs and in the aggregate's sales, which the innovation moves only through the innovator's share.

The corollary makes the separation of Theorem~\ref{thm:main} independent of conduct for innovators that are small in their markets. To first order in its size, such an innovator's profit changes by its cost saving, whether it competes in quantities or in prices, so none of the change in allocative efficiency reaches it. That change is in general of first order, because the innovator's expansion and the fall in the aggregate's price move resources between producers with different markups, and in the nested-CES economies of Section~\ref{sec:primitives} its limit is zero only in special cases, such as the economy of Corollary~\ref{cor:fringe}. In a homogeneous-good Cournot market, a marginal cost reduction by a producer with a small enough share lowers welfare because it shifts output away from rivals with higher margins \citep{LahiriOno1988,Zhao2001}, and Section~\ref{sec:onemarket} carries this force over to differentiated goods under quantity or price competition. Corollary~\ref{cor:small} adds that the responses to a small innovator's cost reduction affect its profit only at second order in its size, which requires differentiated goods. With a homogeneous good, linear demand, and $J$ Cournot producers, a producer's induced price term is $(J-1)/(J+1)$ of its cost saving whatever its share.{} The result on the induced price term holds for any smooth constant-returns aggregate and with spillovers that shrink in proportion to the innovator's size, as pool spillovers do when its knowledge stock is proportional to its sales, whereas spillovers that do not shrink, such as imitation of a given improvement, move rivals' prices at first order (Proposition~\ref{prop:smallgeneral}).{}
\section{When Does Innovation Raise Allocative Efficiency?}\label{sec:primitives}

By Theorem~\ref{thm:main}, the change in allocative efficiency is no part of the innovator's private return, and for a process innovation that raises the innovator's productivity it is the only external term without a fixed sign. This section adds to the propagation formulas of \citet{BaqaeeFarhi2020} the conduct that determines how markups respond. Allocative efficiency is unchanged when markups are equal, there are no intermediate inputs, and factor supplies are fixed (Section~\ref{sec:uniform}), and in nested-CES economies when technologies are Cobb--Douglas and markups are fixed or set by Cournot competition (Section~\ref{sec:ces}). Otherwise three forces determine its sign, within a market (Section~\ref{sec:onemarket}), across markets (Section~\ref{sec:nests}), and along supply chains (Section~\ref{sec:chains}).{} Proposition~\ref{prop:ae} and equation~\eqref{eq:onefixed} restate results of \citet{BaqaeeFarhi2020}, and Online Appendix~\ref{oa:allocative} derives them and gives the proofs.

\subsection{Economies without intermediate inputs}\label{sec:uniform}

Equation~\eqref{eq:ae} holds for any demand system. Without intermediate inputs, $\tilde\lambda=\lambda$, so its second term vanishes. In Corollary~\ref{cor:uniform}, $dI_k$ is the change in producer $k$'s input use valued at current prices (Corollary~\ref{cor:aeformula}), and with one factor $l_k$ is producer $k$'s employment.

\begin{corollary}[Equal markups]\label{cor:uniform}
	Suppose there are no intermediate inputs and factor supplies are fixed. Then for every innovation and every constant $c$,
	\[
		\AEff=\sum_k(\mu_k-c)\,dI_k .
	\]
	With one factor and constant returns, and with $\bar\mu=(\sum_k\lambda_k/\mu_k)^{-1}$ the sales-weighted harmonic mean markup,
	\[
		\frac{\AEff}{E}=\sum_k\lambda_k\Bigl(1-\frac{\bar\mu}{\mu_k}\Bigr)\dlog l_k .
	\]
	If all markups are equal, $\AEff=0$ for every innovation, whatever the demand system and however markups respond.
\end{corollary}

With one factor the weights $\lambda_k(1-\bar\mu/\mu_k)$ sum to zero, so $\AEff/E$ is the sales-weighted covariance between the growth of employment and the markup gap $1-\bar\mu/\mu_k$, and an innovation raises allocative efficiency when it moves labor toward producers whose markups exceed those of the producers that release it. The neutrality requires both conditions, because a common markup still distorts the use of intermediate inputs relative to factors and, with elastic factor supplies, factor supply. 

\subsection{Nested CES economies}\label{sec:ces}

\paragraph{Setting.} Sections~\ref{sec:ces}--\ref{sec:nests} assume that the economy is in the nested-CES standard form of \citet{BaqaeeFarhi2020}. Final demand buys only the output of the final-demand aggregator, indexed by~$0$, so $\beta=\bm e_0$. Every producer $j$, including each CES nest and the aggregator, has a constant-returns CES technology with elasticity $\theta_j$ over its inputs, and nests and the aggregator are competitive. There is one primary factor, labor, in fixed supply, with wage $w$. Because labor is the only primary factor and every producer's cost shares sum to one, $\widetilde\Psi_{jL}=1$ for every producer and $\widetilde\Lambda_L=1$. Let $\mathcal K$ be the producers other than nests and the aggregator. The innovation changes their productivities by $\dlog\bm A$ and their markups by $\dlog\bm\mu$, two vectors indexed by $\mathcal K$. The equilibrium of Section~\ref{sec:environment} is then the equilibrium of the \citet{BaqaeeFarhi2020} economy with wedges equal to the markups, so their results apply with the markup changes $\dlog\bm\mu$ in the role of changes in wedges.

\paragraph{Notation.} Proposition~\ref{prop:ae} compares two vectors across the inputs of each producer. The first records how the innovation changes prices. Cost minimization and $\widetilde\Psi_{jL}=1$ give the price changes
\[
	\dlog p_j=\sum_k\widetilde\Psi_{jk}\,(\dlog\mu_k-\dlog A_k)+\dlog w
\]
for every producer $j$, nests and the aggregator included. We call $\widetilde\Psi\,\dlog(\bm A/\bm\mu)$ the exposure. It uses the columns of $\widetilde\Psi$ for $\mathcal K$, its entry for each producer is the fall in that producer's price relative to the wage, and its entry for labor is zero. The second vector, $\Psi_{(L)}/\Lambda_L$, is the share of labor in each producer's revenue along its whole supply chain relative to the economy's labor share, with the entry $1/\Lambda_L$ for labor. As \citet{BaqaeeFarhi2020} note, it is low for producers whose supply chains are heavily marked up. When the producers in $\mathcal K$ use only labor and employ all of it, its entry for producer $k$ is $\bar\mu/\mu_k$, with $\bar\mu$ the sales-weighted harmonic mean markup of these producers, as in Corollary~\ref{cor:uniform}, so $1-\Psi_{kL}/\Lambda_L$ generalizes the markup gap $1-\bar\mu/\mu_k$. For vectors $\bm u$ and $\bm v$ indexed by producers and labor, the input--output covariance of \citet{BaqaeeFarhi2020} for producer $j$ is
\[
	\Cov_{\widetilde\Omega(j)}(\bm u,\bm v)=\sum_m\widetilde\Omega_{jm}\,u_mv_m-\Bigl(\sum_m\widetilde\Omega_{jm}\,u_m\Bigr)\Bigl(\sum_m\widetilde\Omega_{jm}\,v_m\Bigr),
\]
the covariance across all of $j$'s inputs, labor included, with weights equal to its cost shares.

The propagation formulas of \citet{BaqaeeFarhi2020} then give the change in allocative efficiency for given changes in productivities and markups (Online Appendix~\ref{oa:allocative}).

\begin{proposition}[Allocative efficiency given markup changes; \citealp{BaqaeeFarhi2020}]\label{prop:ae}
	In the setting above, for productivity changes $\dlog\bm A$ and markup changes $\dlog\bm\mu$,
	\begin{equation}
		\frac{\AEff}{E}=-\sum_j\frac{\lambda_j}{\mu_j}(\theta_j-1)\Cov_{\widetilde\Omega(j)}\Bigl(\widetilde\Psi\,\dlog(\bm A/\bm\mu),\frac{\Psi_{(L)}}{\Lambda_L}\Bigr)+\sum_k\Bigl(\lambda_k\frac{\Psi_{kL}}{\Lambda_L}-\tilde\lambda_k\Bigr)\dlog\mu_k .
		\label{eq:ae-prim}
	\end{equation}
\end{proposition}

The first term of \eqref{eq:ae-prim} is substitution. A producer $j$ with $\theta_j>1$ shifts its spending toward the inputs whose prices fall most, and allocative efficiency rises when these inputs have more heavily marked-up supply chains than $j$'s other inputs, that is, lower relative labor shares $\Psi_{mL}/\Lambda_L$. The second term is the direct effect of markup changes. A higher markup at $k$ acts like a lower productivity at $k$, which lowers $\AEff/E$ by $\tilde\lambda_k$ per unit of $\dlog\mu_k$, and releases labor to other producers, which raises it by $\lambda_k\Psi_{kL}/\Lambda_L$. Conduct determines the markup changes, and the next result derives them when producers compete in quantities within CES nests.

\begin{proposition}[Markup responses]\label{prop:markups}
	In the setting above, suppose each producer $k\in\mathcal K$ has a constant markup or is a strategic producer, one that sells its whole output to one CES nest $n(k)$, in which its cost share is $s_k=\widetilde\Omega_{n(k)k}$, competes in quantities with the other suppliers of that nest, and perceives the demand for the nest's output to have the constant elasticity $\eta_{n(k)}$. The markup of a strategic producer $k$ is then $\mu_k=(1-1/\varepsilon_k)^{-1}$ with
	\[
		\frac{1}{\varepsilon_k}=\frac{1-s_k}{\theta_{n(k)}}+\frac{s_k}{\eta_{n(k)}}.
	\]
	Let $\widetilde\Psi_{\cdot\mathcal K}$ be the columns of $\widetilde\Psi$ for $\mathcal K$. Let $Z=\operatorname{diag}(\zeta_k)$, with $\zeta_k=\mu_ks_k(1/\eta_{n(k)}-1/\theta_{n(k)})$ for strategic producers and $\zeta_k=0$ otherwise, and let $\mathcal D$ map a vector $\bm u$ indexed by all nodes and labor to
	\[
		(\mathcal D\bm u)_k=(1-\theta_{n(k)})\Bigl(u_k-\sum_m\widetilde\Omega_{n(k)m}\,u_m\Bigr)
	\]
	for strategic producers and to zero otherwise, where the sum runs over all inputs of the nest, other nests and labor included. Then, when the inverse exists,
	\[
		\dlog\bm\mu=B\,\dlog\bm A,
		\qquad
		B=-\bigl(I-Z\mathcal D\widetilde\Psi_{\cdot\mathcal K}\bigr)^{-1}Z\mathcal D\widetilde\Psi_{\cdot\mathcal K}.
	\]
\end{proposition}

The coefficient $\zeta_k$ is the elasticity of producer $k$'s markup with respect to its within-nest share, so markups rise with the share when $\eta_n<\theta_n$. The proposition also holds when strategic producers compete in prices, perceiving the elasticity $\varepsilon_k=\theta_{n(k)}(1-s_k)+\eta_{n(k)}s_k$ of their own demand and charging $\mu_k=\varepsilon_k/(\varepsilon_k-1)$, with $\zeta_k=\mu_ks_k(\theta_{n(k)}-\eta_{n(k)})/\varepsilon_k^2$ (Online Appendix~\ref{oa:allocative}). The matrix $B$ is the solution of a feedback loop. The vector $\widetilde\Psi_{\cdot\mathcal K}(\dlog\bm\mu-\dlog\bm A)$ gives every price change relative to the wage, $\mathcal D$ turns these price changes into changes in the strategic producers' within-nest shares, and $Z$ turns the share changes into markup changes. Because markup changes feed back into prices, the markup changes $\dlog\bm\mu=B\,\dlog\bm A$ are the fixed point of this loop.

The markup rule is that of the Atkeson--Burstein oligopoly, which \citet{Grassi2018} embeds in an input--output network, and $\zeta_k$ is the elasticity with which \citet{BursteinCarvalhoGrassi2025} characterize the response of markups within a sector. When no nest buys labor or the output of other nests directly and the suppliers of each nest have the same input shares, as in both papers and in Section~\ref{sec:application}, $B$ reduces to their within-sector solution (Online Appendix~\ref{oa:allocative}).

With a continuum of nests combined in Cobb--Douglas fashion, $\eta_n=1$, $1/\mu_k=(1-s_k)(1-1/\theta_n)$, and $\zeta_k=s_k/(1-s_k)$, the case of \citet{BaqaeeFarhi2020}, based on \citet{AtkesonBurstein2008}. When final demand is a CES with elasticity $\eta$ over nests, a supplier perceives the elasticity $\eta/(1-\lambda_n)$, so a constant $\eta_n=\eta$ is exact only for negligible nests, and the results below treat $\eta$ as a parameter of conduct. Theorem~\ref{thm:main} holds for any $\eta$, and when $\eta$ differs from the elasticity $\eta^a$ of the demand that the nest faces, the induced price term also contains the part of the change in the producer's price caused by its own output that it does not perceive, $p_ky_ks_k(1/\eta-1/\eta^a)\,\dlog y_k$.

The two propositions also extend Proposition~4 of \citet{BaqaeeFarhi2020} from fixed to Cournot markups.

\begin{corollary}[Cobb--Douglas]\label{cor:cd}
	If $\theta_j=1$ for every producer, then $\AEff=0$ for every innovation, with fixed markups or with the markups of Proposition~\ref{prop:markups}, which are finite when every $\eta_n>1$.
\end{corollary}

With Cobb--Douglas technologies every buyer keeps its cost shares, and Cournot markups, which depend only on these shares, do not change, so the allocation does not respond. In nested-CES economies with fixed or Cournot markups, allocative efficiency therefore moves only through substitution ($\theta_j\neq1$) and the markup responses that substitution triggers.

\subsection{One market}\label{sec:onemarket}

The forces in Proposition~\ref{prop:ae} are clearest in one market, the horizontal economy of \citet{BaqaeeFarhi2020}. Final demand is a single CES nest with elasticity $\theta$ over producers that use only labor, so $\lambda_k=\tilde\lambda_k=s_k$, and $\bar\mu=(\sum_ks_k/\mu_k)^{-1}$ is the sales-weighted harmonic mean markup. For a productivity change $\dlog A_i$ accompanied by a change $\dlog\mu_i$ in the innovator's markup, with all other markups fixed, their equations~(5) and~(6) give
\begin{equation}
	\frac{\AEff}{E}=s_i\Bigl(1-\frac{\bar\mu}{\mu_i}\Bigr)\bigl[(\theta-1)\,\dlog A_i-\theta\,\dlog\mu_i\bigr]
	\label{eq:onefixed}
\end{equation}
(Online Appendix~\ref{oa:allocative}). For a small innovator that perceives its own-price effect, under quantity or price competition, the markup response is of order $s_i$ (Corollary~\ref{cor:small}), so as its weight goes to zero $\AEff/(s_iE)$ tends to $(\theta-1)(1-\bar\mu/\mu_i)\,\dlog A_i$ with $\mu_i=\theta/(\theta-1)$. This limit is zero only when $\theta=1$ or $\bar\mu=\theta/(\theta-1)$, and it is negative when $\theta>1$ and the other producers' mean markup exceeds $\theta/(\theta-1)$.

The innovator's price relative to the wage changes by $\dlog\mu_i-\dlog A_i$, so the bracket equals $[\theta(1-\dlog\mu_i/\dlog A_i)-1]\,\dlog A_i$, where $1-\dlog\mu_i/\dlog A_i$ is the pass-through of its cost reduction. For an innovator whose markup exceeds the harmonic mean, a productivity gain therefore raises allocative efficiency exactly when $\theta$ times the pass-through exceeds one. With fixed markups the pass-through is one and the condition is $\theta>1$, the case of \citet{BaqaeeFarhi2020}, in which producers with above-mean markups are too small. With zero pass-through, as in \citet{Peters2020}, where an incumbent's own innovation raises its markup and leaves its price unchanged, the innovation lowers allocative efficiency whatever the elasticity.

Under Cournot conduct, markups that rise with the share (Proposition~\ref{prop:markups}) lower the pass-through below one and so dampen the gain, as \citet{BaqaeeFarhi2020} show for $\eta=1$ and rivals that form a fringe with constant markups. Corollary~\ref{cor:onecournot} in Online Appendix~\ref{oa:onecournot} gives the change in allocative efficiency for any perceived elasticity $\eta$, and Section~\ref{sec:app-ae} measures the dampening for U.S. public firms.

\subsection{Several markets}\label{sec:nests}

When final demand combines several markets with elasticity $\eta$ and the innovator's rivals are strategic as well, the change in allocative efficiency also depends on the flow of labor between markets (Corollary~\ref{cor:nests} in Online Appendix~\ref{oa:nests}). A change in the markups of the innovator's market changes the labor that the market employs at a given share of spending, so labor moves between markets even when $\eta=1$. A dominant producer's innovation raises its market's mean markup and so releases labor to other, less marked-up markets, and at large shares this loss outweighs the gain within its market (Table~\ref{tab:dominant}). As the innovator's weight goes to zero, Corollary~\ref{cor:nests} evaluated at $s_i=0$ gives the limit of $\AEff/(\lambda_iE)$ in Corollary~\ref{cor:small}.

\subsection{Supply chains}\label{sec:chains}

With intermediate inputs, an innovation also changes how its buyers divide their costs between its good and their own value added. Suppose producer $j$, with a fixed markup, combines labor, its value added, with the good of producer $i$ according to a CES technology with elasticity $\theta_j$, and final demand buys only $j$'s good. Producer $i$ uses only labor and has a fixed markup $\mu_i$, and $\omega=\widetilde\Omega_{ji}$ is the share of $i$'s good in $j$'s cost.

\begin{corollary}[A supplier and its buyer]\label{cor:chain}
	In this economy, a productivity change $\dlog A_i$ gives
	\[
		\frac{\AEff}{E}=(\theta_j-1)\,\omega(1-\omega)\,\frac{1-1/\mu_i}{1-\omega(1-1/\mu_i)}\,\dlog A_i .
	\]
\end{corollary}

When $\mu_i>1$, the change in allocative efficiency has the sign of $\theta_j-1$, whatever $j$'s markup. Labor is worth $\mu_i$ times more in $i$'s supply chain than in $j$'s value added, and when $\theta_j<1$ the fall in $i$'s price shifts $j$'s spending, and labor, from $i$ to its own value added. In the economy of Section~\ref{sec:application}, every industry combines value added and intermediate inputs with the elasticity \IWAppEpsilonVA{} of \citet{BaqaeeFarhi2020}, and this force lowers the average change in allocative efficiency, while the rise of Cournot markups with market share makes it vary with the innovator's size (Section~\ref{sec:app-ae}).

\subsection{Allocative efficiency and the private return}\label{sec:ae-private}

With fixed markups, equation~\eqref{eq:onefixed} shows how reallocation changes the relation between the social return and the innovator's markup. In one market, $\T_i=Es_i\,\dlog A_i$, $\Pv_i=Es_i\,\dlog A_i/\mu_i$, and $dW=\T_i+\AEff$, so
\begin{equation}
	\frac{dW}{\Pv_i}=\mu_i+(\theta-1)(\mu_i-\bar\mu).
	\label{eq:ratio}
\end{equation}

The first term is $\T_i/\Pv_i$, the ratio in equation~\eqref{eq:liu} without intermediate inputs, and the second is $\AEff/\Pv_i$. Across innovators in the same market, the ratio rises with $\mu_i$ at the rate $\theta$, while its first term rises at the rate one, so reallocation makes $dW/\Pv_i$ rise with the innovator's markup more steeply than $\T_i/\Pv_i$ when goods are substitutes. With zero pass-through, as in \citet{Peters2020}, $\AEff=-(\mu_i-\bar\mu)\Pv_i$ and $dW/\Pv_i=\bar\mu$ for every innovator, whatever $\theta$. The fall in allocative efficiency then exactly offsets the innovator's own markup in the appropriability gap, and the ratio equals $1/\Lambda_L$, the premium at which labor drawn into R\&D costs society more than its wage (Section~\ref{sec:factors}). For a small innovator that perceives its own-price effect, under quantity or price competition, $d\pi_i/(p_iy_i)$ tends to $\dlog A_i/\mu_i$ (Corollary~\ref{cor:small}) and $\AEff/(p_iy_i)$ to the limit given after equation~\eqref{eq:onefixed}, so in that market $dW/d\pi_i$ tends to the value of \eqref{eq:ratio} at $\mu_i=\theta/(\theta-1)$.

Under Cournot conduct the private return also contains the induced price term, and the comparison depends on the innovator's share. Both returns have closed forms in the economy of \citet{BaqaeeFarhi2020} in which one Cournot producer faces a fringe of atomistic producers in each of a continuum of small industries (Corollary~\ref{cor:fringe} in Online Appendix~\ref{oa:fringe}). There the induced price term is proportional to the change in allocative efficiency, $\IP_i=(\theta-1)(1+s_i)\,\AEff/\theta$, although it consists of transfers rather than of any part of that change, and $dW/d\pi_i$ rises with the innovator's share when the fringe's mass is fixed. When the fringe enters freely, $dW/d\pi_i$ is the welfare value of the labor that the innovation saves, per unit of its wage, whatever the share (Remark~\ref{rem:entry}). With strategic rivals in heterogeneous industries, the induced price term need not move with the change in allocative efficiency at all, and it is no measure of that change (Online Appendix~\ref{oa:nests}).
\section{Measured Returns}\label{sec:measured}\label{sec:harberger}

\subsection{The Harberger term}\label{sec:bridge}

The welfare change can also be accounted for at fixed quantities.

\begin{lemma}[Welfare at fixed quantities]\label{lem:accounting}
	Under Assumption~\ref{ass:inputs},
	\[
		dW=\sum_k\bigl(\mc_k\,y_k\,\dlog A_k+U_{b_k}\,db_k\bigr)+\sum_k(p_k-\mc_k)\,dy_k .
	\]
\end{lemma}

The proof is in Online Appendix~\ref{oa:otherproofs}, where $dI_k\equiv\sum_mp_m\,dx_{km}+\sum_fw_f\,dl_{kf}=\mc_k(dy_k-y_k\,\dlog A_k)$. The first sum values each technology change at fixed quantities, a productivity gain at the cost of the inputs it saves at current output. We call the second sum, $\sum_k(p_k-\mc_k)\,dy_k$, the Harberger term. The lemma does not use Assumption~\ref{ass:conduct}, so it holds for any conduct in output markets. It is the identity of \citet{Harberger1971}, which \citet{BaqaeeBurstein2026} extend to economies with heterogeneous agents.

\begin{corollary}[Allocative efficiency and the Harberger term]\label{cor:bridge}
	Under Assumption~\ref{ass:inputs},
	\begin{equation}
		\sum_k(p_k-\mc_k)\,dy_k=\AEff+\sum_k\bigl(E\tilde\lambda_k-\mc_k\,y_k\bigr)\,\dlog A_k .
		\label{eq:bridge}
	\end{equation}
	Under Assumption~\ref{ass:conduct} as well, the wedge can be written as
	\begin{equation}
		dW-d\pi_i=\sum_k(p_k-\mc_k)\,dy_k+\sum_{j\neq i}\bigl(\mc_j\,y_j\,\dlog A_j+U_{b_j}\,db_j\bigr)+\nu_i\,db_i-\IP_i ,
		\label{eq:harberger}
	\end{equation}
	where $\nu_i$ is the quality distortion of Corollary~\ref{cor:gap}.
\end{corollary}

The proof is in Online Appendix~\ref{oa:otherproofs}. Equation~\eqref{eq:harberger}, which we call the Harberger classification, and the decomposition~\eqref{eq:main} are two exact classifications of the same wedge that differ only in where they place the markups on the output that the productivity changes add at a fixed allocation, $\sum_k(E\tilde\lambda_k-\mc_ky_k)\,\dlog A_k$. The decomposition~\eqref{eq:main} assigns them to the appropriability gap and the technology effects $\T_j$, whereas the Harberger classification places them in the Harberger term, where they join the change in allocative efficiency in a single external term that has no fixed sign and requires the change in every producer's output.

Because it contains these markups, the Harberger term is not the change in allocative efficiency. In the vertical chain of Section~\ref{sec:vertical}, no reallocation is possible and $\AEff=0$, yet the Harberger term of a process innovation equals its appropriability gap, as \citet{BaqaeeFarhi2020} point out for the productivity decompositions of \citet{PetrinLevinsohn2012} and \citet{BasuFernald2002}. The reallocation term of \citet{PetrinLevinsohn2012}, the first term of~\eqref{eq:ae}, equals the change in allocative efficiency plus the downstream-markup components of the appropriability gaps of all productivity changes, $E\sum_k(\tilde\lambda_k-\lambda_k)\,\dlog A_k$.

The strategic R\&D literature writes the wedge in the Harberger classification \citep{BranderSpencer1983,Suzumura1992,LeahyNeary1997}.{} By Corollary~\ref{cor:bridge}, the Harberger term in these expressions combines the change in allocative efficiency with the appropriability gaps, which Theorem~\ref{thm:main} separates.

\subsection{Business stealing and output-based returns}\label{sec:bs}\label{sec:measurement}

An innovation that expands the innovator's output takes business from rivals that sell substitutes. The welfare effect of this business stealing \citep{MankiwWhinston1986} is $\sum_{j\neq i}(p_j-\mc_j)\,dy_j$, the rivals' output changes weighted by their margins $p_j-\mc_j$, and is a cost when their outputs fall.

\begin{proposition}[Business stealing]\label{prop:bs}
	Under Assumption~\ref{ass:inputs}, let producer $j\neq i$ experience no technology change, $\dlog A_j=0$ and $db_j=0$. Then
	\[
		d\pi_j=(p_j-\mc_j)\,dy_j+\TR_j ,
	\]
	where $\TR_j$, defined in Corollary~\ref{cor:os}, is the effect of price changes on $j$'s profit at its current quantities. Under Assumption~\ref{ass:conduct} as well, $d\pi_j=\IP_j$. The first term is $j$'s contribution to the Harberger term in~\eqref{eq:harberger} and equals $(\mu_j-1)\,dI_j$, its contribution to~\eqref{eq:ae}. The second consists of price changes and is a transfer.
\end{proposition}

The proof is in Appendix~\ref{app:proofs}. The welfare cost of business stealing from a rival without technology change is thus part of $\AEff$, whereas the rival's profit, and so its market value, changes by this cost plus the transfer $\TR_j$, which can have either sign. When markups fall with market share, a rival's profit can even rise as its output falls (Online Appendix~\ref{oa:roundabout}). The innovator's own gain from taking business is part of $\IP_i$.

Leading estimates of the social return to R\&D treat business stealing as a transfer \citep{Griliches1979}. \citet{BloomSchankermanVanReenen2013} define the marginal social return to a firm's R\&D as the increase in aggregate output that it causes, including spillovers and induced R\&D, and the marginal private return as the increase in the firm's own output, including the output it takes from its rivals, which they measure as half of the change in rivals' market values. The one-half share is exact in a homogeneous-good Cournot model with linear demand. The corresponding share exceeds one half under concave inverse demand. Theorem~\ref{thm:main} instead assigns the innovator's gain from taking business to $\IP_i$ and the welfare cost to rivals to $\AEff$. Three implications follow.

First, no fixed share of the change in rivals' market values recovers the welfare cost of business stealing. With constant returns, the transfer is $\TR_j=y_j\,d(p_j-\mc_j)$, so the share of a rival's profit change due to its lost output, $(p_j-\mc_j)\,dy_j/d\pi_j$, varies across rivals and need not lie between zero and one.

Second, an output-based social return values each firm's output gain at its price and treats business stealing as a transfer, so it omits the markups downstream of the firms whose technology changes and $\AEff$ (Online Appendix~\ref{oa:application}). It is exact with equal markups, no intermediate inputs, and fixed factor supplies, and in an efficient economy it reduces to the theorem of \citet{Hulten1978}. Because the innovation also moves labor between markets, comparing the innovator with its product-market rivals does not sign the omitted term. A social return measured from aggregate productivity, as in \citet{JonesSummers2022}, contains both omitted terms \citep{BaqaeeFarhi2020}.

Third, an output-based private return exceeds the innovator's profit change by its markup on its own output gain and by the credit for rivals' lost output, less its induced price term (Online Appendix~\ref{oa:application}). It equals the profit change only without markups. At the firm's first-order condition, the marginal cost of R\&D equals the change in its profit, so inferring the socially optimal level of R\&D from the ratio of output-based returns compares the social return with a private return other than the one that governs the firm's R\&D. The decomposition does not correct these estimates, which are measured from output and market values. It identifies the terms that their accounting values differently and when each is exact.

\section{R\&D Policy}\label{sec:policy}

\subsection{Local R\&D subsidies}\label{sec:subsidy}

Innovations come from R\&D, and R\&D responds to the innovations of others. Online Appendix~\ref{oa:dynamic} sets up a Markov game in which firm $k$ chooses R\&D effort $r_k$ at cost $\Gamma_k(r_k)$, effort raises knowledge stocks with spillovers, and the static economy of Section~\ref{sec:environment} is in equilibrium at each state. The social and private marginal returns to firm $k$'s effort, $\SMR_k$ and $\PMR_k$, are the effects of one more unit of it on welfare and on firm $k$'s value, each net of the R\&D costs that it induces. Firm $k$'s R\&D wedge is $\Delta_k\equiv\SMR_k-\PMR_k$, and firm $i$'s stake in firm $j$'s effort, $\xi_{ij}$, is the effect of one more unit of $j$'s effort on $i$'s value. One more unit of firm $i$'s effort adds knowledge that spreads over time. To trace the technology changes generated by this extra knowledge, hold each firm's effort on its equilibrium path. Let $dr^0_{k,t}$ denote the effort response prescribed by firm $k$'s strategy to the resulting change in the date-$t$ knowledge state.

\begin{theorem}[Social and private returns to R\&D]
	\label{thm:dynamic}
	Under Assumptions~\ref{ass:mpe} and~\ref{ass:regular} of Online Appendix~\ref{oa:dynamic}, let $\T_{k,t}$, $\Pv_{i,t}$, $\AEff_t$, and $\IP_{i,t}$ be the terms of Theorem~\ref{thm:main}, with firm $i$ as the innovator, for the change in technologies that the extra knowledge induces at date $t$ with every firm's effort held on its path, and evaluate $\Delta_k$ and $\xi_{ij}$ at the date-$t$ state. Then, with discount rate $\rho$,
	\begin{align*}
		\SMR_i&=\int_0^\infty e^{-\rho t}\Bigl[\sum_k\T_{k,t}+\AEff_t+\sum_k\Delta_k\,dr^0_{k,t}\Bigr]dt,\\
		\PMR_i&=\int_0^\infty e^{-\rho t}\Bigl[\Pv_{i,t}+\IP_{i,t}+\sum_{j\neq i}\xi_{ij}\,dr^0_{j,t}\Bigr]dt,
	\end{align*}
	and therefore
	\begin{equation}
		\begin{split}
			\SMR_i-\PMR_i=\int_0^\infty e^{-\rho t}\Bigl[&(\T_{i,t}-\Pv_{i,t})+\sum_{j\neq i}\T_{j,t}+\AEff_t-\IP_{i,t}\\
			&+\Delta_i\,dr^0_{i,t}+\sum_{j\neq i}(\Delta_j-\xi_{ij})\,dr^0_{j,t}\Bigr]dt .
		\end{split}
		\label{eq:dynamic}
	\end{equation}
\end{theorem}

The proof is in Online Appendix~\ref{oa:dynamic-game}. The first line of the integrand in~\eqref{eq:dynamic} is the wedge of Theorem~\ref{thm:main} at each date, and the second values the R\&D responses. The innovator's own R\&D response has no private value at the margin, and rivals' responses enter its private return only through its stakes $\xi_{ij}$ in their R\&D, which extends the strategic effect of \citet{BranderSpencer1983} to R\&D. A firm's R\&D response changes welfare only through its R\&D wedge, by $\Delta_k\,dr^0_{k,t}$, which at the first-order condition $\PMR_k=\Gamma_k'$ is the change in its R\&D spending times $\SMR_k/\PMR_k-1$, as equation~\eqref{eq:ae} values each producer's change in input use at its markup less one. The reallocation of R\&D across firms is therefore external to the innovator, as allocative efficiency is.

Suppose firm $k$ pays $(1-\varsigma_k)\,\Gamma_k(r_k)$ for R\&D effort, where $\varsigma_k$ is a proportional subsidy financed lump sum. Its first-order condition becomes $\PMR_k=(1-\varsigma_k)\,\Gamma_k'$. Effort is efficient at the margin when $\SMR_k=\SMC_k$, where $\SMC_k$ is the social marginal cost of effort. If the ratios $\PMR_k/\SMR_k$ and $\SMC_k/\Gamma_k'$ stay at their values at the observed equilibrium, the first-order condition coincides with this efficiency condition exactly when the subsidy rate is
\begin{equation}
	\varsigma_k^{\mathrm{loc}}=1-\frac{\PMR_k}{\SMR_k}\,\frac{\SMC_k}{\Gamma_k'} .
	\label{eq:subsidy}
\end{equation}

At the observed rate, the first-order condition gives $\Gamma_k'=\PMR_k/(1-\varsigma_k)$, so the rate needs no separate estimate of the cost of effort, only the two marginal returns and $\SMC_k/\Gamma_k'$, the social cost of a dollar of R\&D spending. When R\&D is paid in the numeraire, $\SMC_k=\Gamma_k'$ and $\varsigma_k^{\mathrm{loc}}=(\SMR_k-\PMR_k)/\SMR_k$, which depends on the two marginal returns only through their ratio, and Theorem~\ref{thm:dynamic} decomposes its numerator, the R\&D wedge. The rate is local because it holds the two ratios at their observed values. Determining the optimal subsidy requires accounting for how these ratios change with the equilibrium induced by the subsidy \citep[computed in the linear--quadratic oligopoly by][]{HopenhaynOkumura2026}.

The decomposition has three implications for R\&D policy. First, the allocative-efficiency component of the wedge has no fixed sign (Section~\ref{sec:primitives}), so unlike knowledge spillovers it does not by itself justify subsidizing R\&D. Frameworks with a common markup across producers \citep{JonesWilliams2000,AtkesonBurstein2019} abstract from the reallocation that leaves this component without a fixed sign, and in \citet{AkcigitHanleyStantcheva2022}, whose aggregate output is additively separable across firms, its counterpart is positive. Second, for a process innovation that raises productivity, the appropriability gap is nonnegative and raises the rate for innovators upstream of marked-up producers, as upstream position does for production subsidies in \citet{Liu2019} and for entry subsidies in \citet{BaqaeeFarhiEntry}. For R\&D that draws factors from production, Section~\ref{sec:factors} compares the gap with the premium on those factors. Third, a firm's local rate depends on its rivals' R\&D wedges and on its stakes in their R\&D.{} These implications take the social return at the distorted equilibrium. At the efficient allocation that a planner with production-side instruments would choose, $\AEff=0$ to first order, so when such instruments are available, $\AEff$ is a reason to use them rather than to subsidize R\&D \citep[see also][]{JonesWilliams1998,AlvarezBueraTrachter2026}.

\subsection{R\&D that uses factors}\label{sec:factors}

R\&D often uses factors that would otherwise be employed in production, and the social cost of such R\&D then differs from its price.

\begin{proposition}[Innovation that uses factors]
	\label{prop:factors}
	Let $L_f$ now denote the employment of factor $f$ in production and $L^R_f$ its employment in innovation, so that $\Lambda_f=w_fL_f/E$, and let the supply $L_f+L^R_f$ of each factor be fixed. Suppose the innovation also withdraws $dL^R_f\ge0$ units of factor $f$ from production, which the innovator buys at $w_f$, and the cost of the innovation is included in both returns. Let $\widetilde\Lambda_f=\sum_k\tilde\lambda_k\widetilde\Omega_{kf}$ be the cost-based share of factor $f$. Then
	\[
		dW=\sum_k\T_k-E\sum_f\widetilde\Lambda_f\frac{dL^R_f}{L_f}+\AEff,
		\qquad
		d\pi_i=\Pv_i+\IP_i-\sum_fw_f\,dL^R_f,
	\]
	where $\AEff$ also contains the reallocation that the withdrawal causes. The wedge $dW-d\pi_i$ therefore contains, in addition to the terms of Theorem~\ref{thm:main}, the term
	\[
		-E\sum_f\bigl(\widetilde\Lambda_f-\Lambda_f\bigr)\frac{dL^R_f}{L_f}\le0 .
	\]
\end{proposition}

The proof is in Appendix~\ref{app:proofs}. At a fixed allocation, the factors used for innovation would have produced output worth $\widetilde\Lambda_f/\Lambda_f$ times their cost, because the markups along the supply chains that use them are forgone. The additional term is the network form of the monopoly-distortion effect of \citet{AghionHowitt1992}, in which research labor costs society more than its wage because it would otherwise work for a monopolist. When R\&D uses factors, $\SMC_k$ in equation~\eqref{eq:subsidy} therefore values the factors that a unit of effort uses at $\widetilde\Lambda_f/\Lambda_f$ times their cost and subtracts the change in allocative efficiency that their withdrawal causes.

The premium on a factor averages the ratios of equation~\eqref{eq:liu} over the producers that employ the factor. Because $w_fl_{kf}=E\lambda_k\widetilde\Omega_{kf}/\mu_k$,
\[
	\frac{\widetilde\Lambda_f}{\Lambda_f}=\sum_k\frac{l_{kf}}{L_f}\,\frac{\mu_k\tilde\lambda_k}{\lambda_k},
\]

and with constant returns, for which $\sum_f\widetilde\Lambda_f=1$, the mean of these premiums weighted by factor incomes is $1/\sum_f\Lambda_f$, the ratio of final expenditure to factor income.{} The next result compares this premium with the appropriability gap.

\begin{corollary}[R\&D paid in factors]\label{cor:factors}
	Under the assumptions of Proposition~\ref{prop:factors}, in which the supply of each factor is fixed, consider a process innovation by producer $i$, and let $R=\sum_fw_f\,dL^R_f$ be its cost.
	\begin{enumerate}
		\item Suppose that no producer uses intermediate inputs other than competitive aggregators that sell, directly or through other such aggregators, only to final demand, and that every other producer charges the same markup $\mu$. Then $\AEff=0$ and
		\[
			dW=\mu\Bigl(\sum_k\mc_ky_k\,\dlog A_k-R\Bigr)
		\]
		for any demand system, conduct, and returns to scale, whatever factors the innovation uses.
		\item Suppose that technologies have constant returns, that final demand is homothetic, that every markup depends only on expenditure shares, as in Proposition~\ref{prop:markups}, and that the innovation withdraws factors in proportion to their employment in production, with $dL^R_f/L_f$ the same for every $f$. Then the withdrawal causes no reallocation, and $\AEff$ is the change in allocative efficiency caused by the technology change alone. Without spillovers or induced technology changes,
		\[
			dW=\frac{\mu_i\tilde\lambda_i}{\lambda_i}\,\Pv_i+\AEff-\frac{R}{\sum_f\Lambda_f},
		\]
		where $1/\sum_f\Lambda_f$ is the mean of the ratios $\mu_k\tilde\lambda_k/\lambda_k$ weighted by producers' payments to factors. If $R$ equals the private return of the technology change, $\Pv_i+\IP_i$, as at the first-order condition for R\&D, then
		\[
			dW=\Bigl(\frac{\mu_i\tilde\lambda_i}{\lambda_i}-\frac{1}{\sum_f\Lambda_f}\Bigr)\Pv_i+\AEff-\frac{\IP_i}{\sum_f\Lambda_f}.
		\]
	\end{enumerate}
\end{corollary}

The proof is in Appendix~\ref{app:proofs}. In Part~1, society values each dollar of cost that a producer saves, and each dollar of factors that the innovation withdraws, at $\mu$ dollars, so per dollar the appropriability gap and the premium cancel. Without spillovers or induced technology changes, at $R=\Pv_i+\IP_i$ the social return net of the cost of R\&D is then $-\mu\,\IP_i$. In Part~2, a withdrawal in proportion to employment scales the economy down without changing its allocation, whereas R\&D that uses only some of the factors, returns to scale that differ across producers, or decreasing returns with intermediate inputs make producers contract unevenly, and the reallocation enters $\AEff$. With a proportional withdrawal, the coefficient of $\Pv_i$ in the last display is the appropriability gap net of the premium, positive exactly when the innovator's ratio exceeds the mean. At the mean the two cancel, the counterpart for R\&D of the neutrality of uniform production subsidies in \citet{Liu2019}, and the innovator adds to welfare, net of the cost of its R\&D, only its change in allocative efficiency less the social cost of the R\&D that its induced price term, a transfer, draws.{} Section~\ref{sec:app-cost} measures these terms for U.S. public firms.
\section{The Wedge in the U.S. Economy}\label{sec:application}

This section measures the terms of Theorem~\ref{thm:main} for process innovations by U.S. public firms in \IWAppYear{}, in a baseline with quantity competition within six-digit NAICS markets and in two alternatives (Table~\ref{tab:application7}). The findings hold in every column of Table~\ref{tab:application7} unless conduct is named.

\subsection{The economy}\label{sec:app-economy}

The economy has the 66 industries of the BEA summary use table, with the nested-CES production and final demand of \citet{BaqaeeFarhi2020} and labor and capital in fixed supply. Each industry's output is a CES aggregate, with elasticity $\sigma$, of its six-digit NAICS markets, and each market's output is a CES aggregate, with elasticity $\theta$, of the goods of its producers and of a fringe of negligible producers. The producers are the business segments of the \IWAppSegFirms{} Compustat firms of the sample of \citet{BaqaeeFarhi2020} that sell in the United States, and a producer's share $s$ is its domestic sales relative to its market's size from Economic Census receipts. A fringe of fixed mass produces the rest, except in the \IWAppSegFringeless{} of \IWAppSegMarkets{} markets where the producers' sales exceed that size, which hold \IWAppSegFringelessSales{} percent of the producers' domestic sales (Online Appendix~\ref{oa:application}).

Producers perceive the elasticity $\sigma$ of their market's demand, the true elasticity for a market small in its industry, so by Proposition~\ref{prop:markups} a producer's markup is $(1-1/\varepsilon)^{-1}$ with $1/\varepsilon=(1-s)/\theta+s/\sigma$. The elasticities $\theta=\IWAppSegTheta$ and $\sigma=\IWAppSegEta$ match the aggregate markup of \IWAppSegPremium{} that the accounting-profit markups of \citet{BaqaeeFarhi2020} give and an average pass-through of \IWAppSegPassOwn{} of an innovator's cost reduction into its price when only its own markup responds, close to the estimate of about 0.6 for Belgian manufacturers in \citet{AmitiItskhokiKonings2019}. As in their estimates, it falls with size, from \IWAppSegPassSmall{} for the innovators with the smallest market shares, which account for \IWAppPassSmallSalesCut{} percent of sales, to \IWAppSegPassLarge{} for the others, against nearly one and about 0.5 for small and large Belgian firms.

For each innovation, we report first-order effects per unit increase in the log productivity of the firm's largest producer, allowing markups to respond only within that producer's market. One alternative replaces six-digit markets with the BEA industries of the economy that \citet{BaqaeeFarhi2020} calibrate, with worldwide sales and a unit perceived elasticity of the industry's demand. The other replaces quantity with price competition \citep{AtkesonBurstein2008}, with the same aggregate markup and a pass-through of \IWAppSegBertPassOwn{}, above the estimates of \citet{AmitiItskhokiKonings2019}.

\subsection{The terms of the wedge}\label{sec:app-ae}

By equation~\eqref{eq:liu}, the technology effect of a process innovation relative to the private technology value is $\mu_i\tilde\lambda_i/\lambda_i$, which depends only on markups and the network. With the user-cost and accounting-profit markups of \citet{BaqaeeFarhi2020}, in their economy, its sales-weighted mean exceeds one by \IWAppTPExcessUserCost{} and \IWAppTPExcessAccounting{} percent, and the markups downstream of the innovator account for \IWAppDownShareLow{} to \IWAppDownShareHigh{} percent of the excess.{} The other terms rest on producers' first-order conditions, which measured markups need not satisfy (Remark~\ref{rem:wedges}), so they use the markups that conduct implies, with which the gap is \IWAppSegGap{} per unit of the private return in the baseline.

\begin{table}[!t]
	\centering
	\footnotesize
	\setlength{\tabcolsep}{4pt}
	\caption{Social and private returns to process innovation by U.S. public firms, 2015}
	\label{tab:application7}
	\begin{tabular}{@{}lcccc@{}}
		\toprule
		& \multicolumn{2}{c}{Quantity competition} & \multicolumn{2}{c}{Price competition} \\ \cmidrule(lr){2-3}\cmidrule(lr){4-5}
		& Six-digit & BEA & Six-digit & BEA \\
		& markets & industries & markets & industries \\
		\midrule
		\multicolumn{5}{@{}l}{\emph{A. Returns without technology spillovers}} \\
		Own-cost pass-through & 0.62 & 0.62 & 0.87 & 0.94 \\
		Social over private return, $dW/d\pi_i$ & 1.14 & 0.99 & 1.26 & 1.19 \\
		\quad Appropriability gap & 0.27 & 0.23 & 0.26 & 0.24 \\
		\quad Change in allocative efficiency & 0.04 & \ensuremath{-}0.03 & \ensuremath{-}0.04 & \ensuremath{-}0.07 \\
		\quad Minus induced price term & \ensuremath{-}0.17 & \ensuremath{-}0.21 & 0.05 & 0.02 \\
		Sales with $|\AEff/d\pi_i|>0.1$ (\%) & 56 & 68 & 31 & 30 \\
		R\&D-weighted $dW/d\pi_i$ & 1.08 & 0.98 & 1.24 & 1.20 \\
		R\&D-weighted allocative efficiency & \ensuremath{-}0.04 & \ensuremath{-}0.09 & \ensuremath{-}0.07 & \ensuremath{-}0.08 \\
		R\&D with $dW/d\pi_i$ below the premium (\%) & 80 & 90 & 46 & 70 \\
		Rivals' lost output / profit loss, median & 0.76 & 0.72 & 0.93 & 0.94 \\
		\addlinespace \multicolumn{5}{@{}l}{\emph{B. Innovators below 1 percent of their market or industry}} \\
		Share of innovating firms (\%) & 62 & 84 & 62 & 84 \\
		Sales / R\&D (\%) & 8 / 15 & 19 / 23 & 8 / 15 & 19 / 23 \\
		Change in allocative efficiency & \ensuremath{-}0.16 & \ensuremath{-}0.31 & \ensuremath{-}0.09 & \ensuremath{-}0.10 \\
		Induced price term & 0.026 & 0.054 & \ensuremath{-}0.001 & \ensuremath{-}0.001 \\
		R\&D below the factor premium (\%) & 100 & 99 & 99 & 99 \\
		Shortfall from the premium, R\&D-weighted & 0.31 & 0.40 & 0.08 & 0.06 \\
		\addlinespace \multicolumn{5}{@{}l}{\emph{C. With technology spillovers (innovators with patents and R\&D)}} \\
		Social over private return, $dW/d\pi_i$ & 6.24 & 5.39 & 6.92 & 6.39 \\
		Output-based ratio, credit of BSV & 2.62 & 2.61 & 2.64 & 2.65 \\
		Output-based ratio, no credit & 5.01 & 5.02 & 5.01 & 5.02 \\
		\bottomrule
	\end{tabular}

	\vspace{0.3em}
	\parbox{\linewidth}{\footnotesize Notes: Per unit of the private return; ratios pool the innovations of all innovators. For each innovation, we report first-order effects per unit increase in the log productivity of the firm's largest producer, allowing markups to respond only within that producer's market; the own-cost pass-through has only the innovator's markup responding. Six-digit markets: competition within six-digit NAICS markets, each Compustat business segment a producer with its domestic sales; $\theta=13.8$ and $\eta=\sigma=1.7$ with quantity competition, $\theta=11.2$ and $\eta=\sigma=1.7$ with price competition. BEA industries: competition within BEA industries; $\theta=16$ with quantity competition and $\theta=11.0$ with price competition, $\eta=1$. Every economy has the aggregate markup 1.20, which is also the factor premium, the social cost of a dollar of R\&D paid in labor and capital. R\&D-weighted statistics weight innovators by their R\&D spending. Panel B: innovators with less than 1 percent of their six-digit market (six-digit columns) or BEA industry (BEA columns). The output-based ratios credit the innovator with its rivals' lost output as \citet{BloomSchankermanVanReenen2013} estimate it, in units of output at prices, or with none. Online Appendix Table~\ref{tab:oa-robustness} reports further specifications, and Online Appendix~\ref{oa:application} describes the construction.}
\end{table}

Without technology spillovers (panel~A of Table~\ref{tab:application7}), the social return in the baseline is \IWAppSegRatio{} times the private return. With quantity competition, the induced price term, a transfer, enters the wedge with \IWAppSegMinusIP{} per unit of the private return and offsets most of the gap. With price competition the term is close to zero.

The change in allocative efficiency is small on average but not for individual innovations. In the baseline it ranges from \IWAppSegAEPTen{} to \IWAppSegAEPNinety{} per unit of the private return between the 10th and 90th percentiles of innovators weighted by sales, and it exceeds a tenth in absolute value for innovators that account for \IWAppSegAELarge{} percent of sales. With a fixed fringe and the network elasticities of \citet{BaqaeeFarhi2020}, the R\&D-weighted change in allocative efficiency is negative in every specification except quantity competition at the user-cost markup level. In the specifications with a negative change, it offsets \IWAppFixedMinRDWAEOffset{} to \IWAppFixedMaxRDWAEOffset{} percent of the R\&D-weighted appropriability gap. At the user-cost markup level under quantity competition, it is close to zero, at \IWAppSegUCRDWAE{} (Online Appendix Table~\ref{tab:oa-robustness}). In the baseline with free entry by the fringe in the innovator's market, it is \IWAppRDWAESegFree{} (Section~\ref{sec:app-small}). In the baseline with quantity competition, the negative sign comes from the innovators whose largest producer's market has no fringe, \IWAppSalesShareNoFringeSeg{} percent of the innovators' sales and \IWAppRDShareNoFringeSeg{} percent of their R\&D. For the others, the R\&D-weighted change is \IWAppRDWAEFringeSeg{} with quantity competition and \IWAppRDWAEFringeSegBert{} with price competition. This is the change that one innovation causes, not the realized change from all sources that accounts for about half of aggregate TFP growth in \citet{BaqaeeFarhi2020}.

Two forces shape the change that an innovation causes. Within markets, markups rise with share, so most producers charge less than their market's mean markup, and their innovations move resources away from more marked-up rivals (Corollaries~\ref{cor:onecournot} and~\ref{cor:nests}). Along supply chains, value added and intermediate inputs are complements, so an innovation that lowers its industry's price shifts its buyers' spending from intermediate inputs, which carry markups, to their own labor and capital (Corollary~\ref{cor:chain}). The first force makes the change vary with the innovator's size, and the second, on which the negative sign of the R\&D-weighted mean also rests, lowers its average. In the baseline, that mean is \IWAppRDWAESegUnitVA{} with the elasticity between value added and intermediate inputs set to one and \IWAppRDWAESegUnit{} with all elasticities outside markets set to one, in which case the change for innovators below 1 percent of their market is still \IWAppSegUnitAESizeA{}.

\subsection{Small innovators}\label{sec:app-small}

\begin{figure}[!t]
	\centering
	\includegraphics[width=0.8\linewidth]{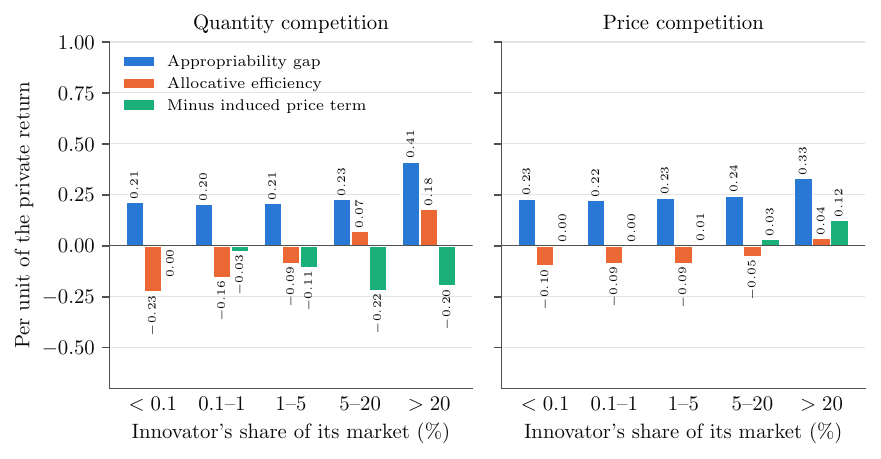}
	\caption{Terms of the wedge by the innovator's share of its market, 2015}
	\label{fig:size}

	\vspace{0.3em}
	\parbox{\linewidth}{\footnotesize Notes: The appropriability gap $\T_i-\Pv_i$, the change in allocative efficiency $\AEff$, and minus the induced price term $-\IP_i$, each summed over the innovators in a group and divided by the sum of their private returns $d\pi_i$, in the baseline with quantity competition (left) and with price competition (right). Groups are by the share of the innovator's largest producer in its market. There are no technology spillovers.}
\end{figure}

Innovators with less than 1 percent of their market are \IWAppFirmShareSmallSeg{} percent of the innovating firms, with \IWAppSegSalesSizeA{} percent of sales and \IWAppSegRDSizeA{} percent of R\&D. As Corollary~\ref{cor:small} predicts, their induced price term nearly vanishes under either conduct, whereas their change in allocative efficiency does not. With quantity competition, their induced price term is \IWAppSegIPSizeA{} per unit of the private return, with a median of \IWAppSegIPMedianSizeA{}, and their change in allocative efficiency is \IWAppSegAESizeA{}. With price competition the two are \IWAppSegBertIPSizeA{} and \IWAppSegBertAESizeA{} (Figure~\ref{fig:size}).

With quantity competition, the change in allocative efficiency of small innovators is \IWAppSegSmallNoFringeAE{} per unit of the private return in markets without a fringe and \IWAppSegSmallFringeAE{} in markets with one. If that fringe enters freely, so that the innovator takes business only from it (Remark~\ref{rem:entry} in Online Appendix~\ref{oa:extensions}), the change is \IWAppFESmallAESegFringe{} and the induced price term \IWAppFESmallIPSegFringe{}.

Without technology spillovers, the static social return falls short of the social cost (Section~\ref{sec:app-cost}) for innovators accounting for \IWAppSegRDNegShareSmall{} percent of small innovators' R\&D spending under quantity competition. The corresponding shares are \IWAppSegBertRDNegShareSmall{} percent under price competition and \IWAppFESmallRDNegSegFringe{} percent with free entry. The shortfall per dollar of private return is \IWAppSegRDShortfallSmall{}, \IWAppSegBertRDShortfallSmall{}, and \IWAppFESmallShortfallSegFringe{}. With quantity competition, their appropriability gap net of the factor premium is \IWAppSegRDGapNetSmall{} per unit of the private return and their change in allocative efficiency \IWAppSegRDAESmall{}, so the shortfall comes from the change in allocative efficiency, although their own markups are below the premium (Corollary~\ref{cor:factors}). Innovators with less than 5 percent of their market account for \IWAppFiveSalesSeg{} percent of sales and \IWAppFiveRDSeg{} percent of R\&D. Their change in allocative efficiency is \IWAppFiveAESeg{} with quantity competition and \IWAppFiveAESegBert{} with price competition, and the static return falls short of the social cost for \IWAppFiveRDNegSeg{} and \IWAppFiveRDNegSegBert{} percent of their R\&D. Measured markups also rise with share within markets, and across the quantity-competition rows of panels A and C of Online Appendix Table~\ref{tab:oa-robustness} other than unit elasticities and free entry, including those with the $\sigma$ that their slopes imply, small innovators' change in allocative efficiency is \IWAppRobSegCMinSmallAE{} to \IWAppRobSegCMaxSmallAE{} and their induced price term \IWAppRobSegCMinSmallIP{} to \IWAppRobSegCMaxSmallIP{}.

\subsection{Business stealing and estimated returns to R\&D}\label{sec:app-bs}

By Proposition~\ref{prop:bs}, the welfare cost of the business that rivals lose is their lost output weighted by their margins $p_j-\mc_j$, and the rest of their profit change is a transfer. At the sales-weighted median innovation, this output part is \IWAppSegRivalShareMedian{} of the rivals' profit loss in the baseline and \IWAppMainMinRivalShareMedian{} to \IWAppMainMaxRivalShareMedian{} across Table~\ref{tab:application7}. The half that \citet{BloomSchankermanVanReenen2013} attribute to lost output, a share they regard as conservative, therefore understates the welfare cost of business stealing and, in their own accounting, the private return.

With the technology spillovers that \citet{BloomSchankermanVanReenen2013} estimate from the proximity of firms' patent portfolios \citep{Jaffe1986}, the social return of the \IWAppSegSpillFirms{} innovators with patents and R\&D is \IWAppSegSpRatio{} times the private return in the baseline (panel~C). The output-based accounting of \citet{BloomSchankermanVanReenen2013} credits the innovator's private return with its rivals' lost output at their prices, half of each rival's percentage fall in market value times its sales. With their estimate of this credit, the output-based ratio is \IWAppSegBSVRatioBSVCreditPrice{}, close to the ratio of their estimated returns, 55.0 and 20.7 percent, and without it \IWAppSegBSVNoSteal{}, close to the \IWAppBSVSymmetric{} of their symmetric case. The difference from the welfare-based ratio comes mainly from the credit, which is large because a rival's sales are $\mu_j/(\mu_j-1)$ times its profit (Online Appendix~\ref{oa:application}).

\subsection{The cost of R\&D}\label{sec:app-cost}

At the first-order condition without R\&D subsidies, a marginal dollar of R\&D buys a private return of one dollar and, in the stationary specification of \citet{BloomSchankermanVanReenen2013} and leaving aside R\&D responses, a social return of $dW/d\pi_i$ dollars (Corollary~\ref{cor:stationary} in Online Appendix~\ref{oa:dynamic}). Paid in labor and capital in proportion to their employment, it costs society \IWAppSegPremium{} dollars (Section~\ref{sec:factors}), and where it draws scientists and engineers, whose supply is inelastic, from other firms' R\&D \citep{Goolsbee1998}, it costs the displaced R\&D valued at those firms' social return per dollar.

Without technology spillovers, $dW/d\pi_i$ weighted by R\&D spending is \IWAppRDWRatioSeg{} in the baseline and falls short of the factor premium for innovators that account for \IWAppSegRDNegShare{} percent of R\&D, for which the local rate of equation~\eqref{eq:subsidy} is a tax. That rate is the total subsidy rate, here $1-\IWAppSegPremium\,d\pi_i/dW$, so the conclusion does not depend on existing R\&D tax credits. A credit at the rate $\varsigma$ lowers the private return on a marginal dollar to $1-\varsigma$ dollars, so at current credits the static net return is negative for at least these innovators. With quantity competition, the shortfall comes mainly from the R\&D that firms undertake for the induced price term, a transfer. Per unit of the private return and weighted by R\&D, that R\&D costs society \IWAppSegRDIPCost{}, whereas the appropriability gap net of the premium adds \IWAppSegRDGapNet{} and the change in allocative efficiency \IWAppSegRDAE{} (Corollary~\ref{cor:factors}). With price competition, where that term is close to zero, the R\&D-weighted ratio is \IWAppRDWRatioSegBert{}, about the premium. With technology spillovers, the social return in the baseline exceeds the cost of R\&D for innovators that account for \IWAppSegRDPosShareSp{} percent of the R\&D of the firms that generate spillovers, \IWAppSegSpillRDShare{} percent of the total.

\section{Conclusion}\label{sec:conclusion}

We asked which parts of the social return to an innovation the innovator is paid for when markups arise from producers' own optimization. The innovator's own output and price adjustments have no first-order value to it, so the change in allocative efficiency that the innovation causes reaches its profit only through other agents' responses. Under quantity competition, these responses affect the innovator's profit through transfers. For an innovator that is small in its market the responses to its own technology change affect its profit at second order in its size under quantity or price competition (Corollary~\ref{cor:small} and Proposition~\ref{prop:smallgeneral}), while in the nested-CES economies of Section~\ref{sec:primitives} the change in allocative efficiency is of first order except in special cases. For a process innovation, the innovator is paid its cost saving and so forgoes its own markup and the markups downstream of it on the output that the innovation adds at a fixed allocation. The Harberger term, on which models of strategic R\&D and the argument of \citet{GreenwaldStiglitz1986} rest, combines the change in allocative efficiency with these appropriability gaps.

The change in allocative efficiency has no fixed sign. It is positive when an innovation moves resources toward more marked-up supply chains. Markups that rise with market share dampen this gain and across markets can reverse it, and along supply chains a productivity gain at a marked-up supplier lowers allocative efficiency when its buyers' value added and intermediate inputs are complements (Section~\ref{sec:primitives}). In a Markov game of R\&D, the reallocation of R\&D across firms is external to the innovator in the same way, because the ratio of the social to the private marginal return plays the role in R\&D that the markup plays in production (Theorem~\ref{thm:dynamic}).

These results have three consequences. First, the profits and market values of rivals whose technology does not change move by the welfare cost of business stealing from them plus transfers of either sign, so no fixed share of their change recovers that cost, and they can misstate it even in sign (Proposition~\ref{prop:bs}). Second, a social return that treats business stealing as a transfer omits the change in allocative efficiency. Third, unlike knowledge spillovers, the allocative-efficiency component of the R\&D subsidy that closes the wedge has no fixed sign, so it does not by itself justify subsidizing R\&D. When R\&D draws factors from production in proportion to employment, the appropriability gap exceeds the factor premium exactly for innovators whose ratio $\mu_i\tilde\lambda_i/\lambda_i$ exceeds its factor-payment-weighted mean (Section~\ref{sec:factors}).

For U.S. public firms in \IWAppYear{}, Baqaee and Farhi's markups give an appropriability gap of \IWAppTPExcessUserCost{} to \IWAppTPExcessAccounting{} percent of the private technology value on average, of which downstream markups account for \IWAppDownShareLow{} to \IWAppDownShareHigh{} percent. Under quantity or price competition, the change in allocative efficiency has either sign across innovations and, for innovators small in their markets, far exceeds the induced price term in absolute value. Rivals' lost output accounts for more than half of the fall in their profits. With estimated technology spillovers, output-based accounting puts the ratio of social to private returns to R\&D at \IWAppSegBSVRatioBSVCreditPrice{}, against \IWAppSegSpRatio{} in welfare terms, mainly because it credits innovators with their rivals' lost output at their prices (Section~\ref{sec:application}).

\begingroup
\singlespacing
\bibliography{references}
\endgroup

\appendix

\section{Proofs}\label{app:proofs}

This appendix proves Corollaries~\ref{cor:aeformula}, \ref{cor:small}, and~\ref{cor:factors} and Propositions~\ref{prop:bs} and~\ref{prop:factors}. Lemmas~\ref{lem:allocation} and~\ref{lem:private} and Theorem~\ref{thm:main} are proved in the text, Online Appendix~\ref{oa:extensions} derives the results for price competition stated in Remark~\ref{rem:bertrand}, and the Online Appendix proves the other results.

Several steps use the differential of a producer's profit $\pi_k=p_ky_k-C_k$. Under Assumption~\ref{ass:inputs}, Shephard's lemma gives $\partial C_k/\partial p_m=x_{km}$ and $\partial C_k/\partial w_f=l_{kf}$. Because productivity is Hicks-neutral, cost depends on $y_k$ and $A_k$ only through $y_k/A_k$, so $\partial C_k/\partial\log A_k=-\mc_k\,y_k$. Hence, for every producer $k$,
\begin{align*}
	d\pi_k&=y_k\,dp_k+(p_k-\mc_k)\,dy_k+\mc_k\,y_k\,\dlog A_k-\sum_mx_{km}\,dp_m-\sum_fl_{kf}\,dw_f\\
	&=\TR_k+(p_k-\mc_k)\,dy_k+\mc_k\,y_k\,\dlog A_k ,
\end{align*}
with $\TR_k$ as defined in Corollary~\ref{cor:os}. This expression does not use Assumption~\ref{ass:conduct} or any other assumption on how $k$ sets its output or its price.

\begin{proof}[Proof of Corollary~\ref{cor:aeformula}]
	Using $p_m=U_{c_m}$, $w_f\,dL_f=V_f'\,dL_f$ for every factor (for a factor in inelastic supply because $dL_f=0$), and market clearing,
	\begin{align*}
		dW&=\sum_mp_m\Bigl(dy_m-\sum_kdx_{km}\Bigr)+\sum_kU_{b_k}\,db_k-\sum_fw_f\sum_kdl_{kf}\\
		&=\sum_kp_k\,dy_k-\sum_kdI_k+\sum_kU_{b_k}\,db_k .
	\end{align*}
	By Assumption~\ref{ass:inputs}, $dI_k=\mc_kA_k\,dG_k=\mc_k\bigl(dy_k-y_k\,\dlog A_k\bigr)$, so $p_k\,dy_k-dI_k=(\mu_k-1)\,dI_k+p_ky_k\,\dlog A_k$ and
	\[
		dW=\sum_k(\mu_k-1)\,dI_k+\sum_k\bigl(p_ky_k\,\dlog A_k+U_{b_k}\,db_k\bigr).
	\]
	By Lemma~\ref{lem:allocation}, the chain rule, and Definition~\ref{def:ae}, $dW=\sum_k\T_k+\AEff$ with $\T_k=E\tilde\lambda_k\,\dlog A_k+U_{b_k}\,db_k$. Subtracting $\sum_k\T_k$ and using $p_ky_k=E\lambda_k$ gives \eqref{eq:ae}.
\end{proof}

\begin{proof}[Proof of Corollary~\ref{cor:small}]
	Producer $i$ has constant returns, so by Corollary~\ref{cor:os}, whatever its conduct,
	\[
		\frac{d\pi_i-\Pv_i}{p_iy_i}=\Bigl(1-\frac{1}{\mu_i}\Bigr)\dlog(p_iy_i)-\frac{\dlog A_i-\dlog\mu_i}{\mu_i}.
	\]
	Producer $n$ is competitive with a constant-returns CES technology, so $p_iy_i=s_ip_ny_n$ and $\dlog s_i=(1-\theta_n)(\dlog p_i-\dlog p_n)$, where $\dlog p_i=\dlog\mu_i-\dlog A_i+\sum_m\widetilde\Omega_{im}\,\dlog p_m$ by cost minimization. Substituting $\dlog(p_iy_i)=\dlog s_i+\dlog(p_ny_n)$ and using $(\theta_n-1)(1-1/\mu_i)-1/\mu_i=\theta_n(1-1/\mu_i)-1$ gives equation~\eqref{eq:small}.

	Under quantity competition, $d\pi_i-\Pv_i=\IP_i$ by Lemma~\ref{lem:private}, and $\theta_n(1-1/\mu_i)-1=\theta_n/\varepsilon_i-1=s_i(\theta_n/\eta_n-1)$. Under price competition, $d\pi_i=\Pv_i^B+\IP_i^B$ by Remark~\ref{rem:bertrand}, with $\Pv_i^B=\Pv_i$ for a process innovation, and $\theta_n(1-1/\mu_i)-1=s_i(\theta_n-\eta_n)/\varepsilon_i$. The elasticity $\eta_n$ may vary with the equilibrium. As long as it stays positive, $\varepsilon_i\to\theta_n$ as $s_i\to0$, so $\mu_i\to\theta_n/(\theta_n-1)$, and the bracket in~\eqref{eq:small} is of order $s_i$.

	Take as variables $i$'s output in levels and $n$'s technology in dual form, in which the equilibrium conditions remain twice continuously differentiable at $\delta_i=0$. There, $i$'s own conditions form a block with unit diagonal, and the other conditions depend on $i$'s variables only through $i$'s input demands, which vanish with its output, so the Jacobian is nonsingular exactly when that of the economy without good $i$ is. The implicit function theorem then makes the equilibrium a twice continuously differentiable function of $(\delta_i,\log A_i)$ near $\delta_i=0$, so the derivatives of equilibrium objects with respect to $\log A_i$ are continuously differentiable in $\delta_i$. Under Assumption~\ref{ass:conduct} or Remark~\ref{rem:bertrand}, the other producers' conditions involve good $i$ only through $n$'s demand for it and $n$'s unit cost, both of which carry the weight $\delta_i$. When $\delta_i=0$, producer $n$ does not use good $i$, so no equilibrium object other than $p_i$ depends on $A_i$, and $s_i=0$, so $\mu_i=\theta_n/(\theta_n-1)$ whatever $A_i$. The derivatives with respect to $\log A_i$ of every price other than $p_i$, factor prices included, of every markup, of the sales of every producer other than $i$, and of welfare are therefore zero at $\delta_i=0$. They are of order $\delta_i$, and hence of order $s_i$, since $s_i/\delta_i$ tends to $(p_i/p_n)^{1-\theta_n}>0$. The second term of~\eqref{eq:small} consists of such derivatives, and so does $\dlog\mu_i$, which proves the first part.

	With no other technology change, Theorem~\ref{thm:main} gives $\AEff/E=dW/E-\tilde\lambda_i\,\dlog A_i$. Both terms are continuously differentiable in $\delta_i$ and zero at $\delta_i=0$, the second because $\tilde\lambda_i=0$ when $n$ does not use good $i$. Since $p_iy_i=\lambda_nEs_i$ with $\lambda_n>0$, $\AEff/(p_iy_i)$ is the ratio of two continuously differentiable functions of $\delta_i$ that vanish at $\delta_i=0$, the second with a positive derivative there, so it converges.

	Finally, consider the economy in which $i$'s markup is fixed at a value $m$ that $i$ keeps as profit, with everything else unchanged. At $m$ equal to $i$'s markup in the original equilibrium, the two economies have the same equilibrium, and at $\delta_i=0$ the same Jacobian, because there $s_i=0$ for all prices and $i$'s markup does not respond. The original equilibrium is that of the economy with the fixed markup at $m=\mu_i$, so its welfare change is that of the economy with the fixed markup plus $(\partial W/\partial\log m)\,\dlog\mu_i$, and the technology effects are the same. At $\delta_i=0$ no equilibrium object other than $p_i$ depends on $m$, so $\partial W/\partial\log m$ is of order $s_i$, as is $\dlog\mu_i$. Holding $i$'s markup fixed therefore changes $\AEff/E$ by a term of order $s_i^2$ and $\AEff/(p_iy_i)$ by a term of order $s_i$.
\end{proof}

\begin{proof}[Proof of Proposition~\ref{prop:bs}]
	Producer $j$ has $\dlog A_j=0$ and $db_j=0$. The profit differential at the start of this appendix, which does not use the first-order condition, then gives $d\pi_j=(p_j-\mc_j)\,dy_j+\TR_j$. Under Assumption~\ref{ass:conduct}, Lemma~\ref{lem:private} applies to $j$ with $\Pv_j=0$, so $d\pi_j=\IP_j$. With no technology change, $dI_j=\mc_j\,dy_j$, so $(p_j-\mc_j)\,dy_j=(\mu_j-1)\,dI_j$. This is $j$'s term in \eqref{eq:ae}, and $(p_j-\mc_j)\,dy_j$ is by definition its term in the Harberger term of \eqref{eq:harberger}. Finally, $\TR_j$ values changes in prices at $j$'s quantities. By market clearing, $y_j\,dp_j=c_j\,dp_j+\sum_kx_{kj}\,dp_j$ is paid by $j$'s buyers, while $x_{jm}\,dp_m$ and $l_{jf}\,dw_f$ are received by its suppliers of goods and factors, so $\TR_j$ is a transfer.

	Under price competition $d\pi_j=(p_j-\mc_j)\,dy_j+\TR_j$ holds unchanged, and $d\pi_j=\IP_j^B$ because $\Pv_j^B=0$.
\end{proof}

\begin{proof}[Proof of Proposition~\ref{prop:factors}]
	With the supply $\sum_kl_{kf}+L^R_f$ of each factor fixed, $\sum_kdl_{kf}=-dL^R_f$, and the proof of Lemma~\ref{lem:accounting} gives $dW$ the additional term $-\sum_fw_f\,dL^R_f$. At a fixed allocation, production employment of factor $f$ falls by $dL^R_f$ and each producer's use of it falls in the same proportion, so
	\[
		\dlog y_k=\dlog A_k+\sum_m\widetilde\Omega_{km}\,\dlog y_m-\sum_f\widetilde\Omega_{kf}\,\frac{dL^R_f}{L_f},
	\]
	and the proof of Lemma~\ref{lem:allocation} gives the welfare change at a fixed allocation the additional term $-E\sum_k\tilde\lambda_k\sum_f\widetilde\Omega_{kf}\,dL^R_f/L_f=-E\sum_f\widetilde\Lambda_f\,dL^R_f/L_f$. Welfare also depends on the withdrawn factors, and the change in welfare at a fixed allocation is the term above, so the chain rule and Definition~\ref{def:ae} give the expression for $dW$. Lemma~\ref{lem:private} does not involve factor supplies, so the innovator's profit from production changes by $\Pv_i+\IP_i$, and the innovator pays $w_f\,dL^R_f=E\Lambda_f\,dL^R_f/L_f$ for the withdrawn factors. Subtracting, $dW-d\pi_i$ consists of the terms of Theorem~\ref{thm:main} and $-E\sum_f(\widetilde\Lambda_f-\Lambda_f)\,dL^R_f/L_f$. Finally, $\Lambda_f=\sum_k\lambda_k\Omega_{kf}$ with $\Omega_{kf}=\widetilde\Omega_{kf}/\mu_k\le\widetilde\Omega_{kf}$, and $\lambda_k\le\tilde\lambda_k$ by the argument in the proof of Corollary~\ref{cor:gap}, so $\Lambda_f\le\widetilde\Lambda_f$.
\end{proof}

Because the welfare accounting at fixed quantities and the private return acquire the same term $-\sum_fw_f\,dL^R_f$, equation~\eqref{eq:harberger} holds unchanged.

\begin{proof}[Proof of Corollary~\ref{cor:factors}]
	For Part~1, competitive aggregators that sell, directly or through other such aggregators, only to final demand, such as CES nests, can be treated as part of the demand system, so producers use only factors. With the supply of each factor fixed, the proof of Proposition~\ref{prop:factors} gives the welfare accounting of Lemma~\ref{lem:accounting} the additional term $-R$,
	\[
		dW=\sum_k\mc_ky_k\,\dlog A_k+\sum_k(p_k-\mc_k)\,dy_k-R .
	\]
	The proof of Lemma~\ref{lem:accounting} gives $\mc_k\,dy_k=dI_k+\mc_ky_k\,\dlog A_k$, using cost minimization but not constant returns, so with $p_k-\mc_k=(\mu-1)\mc_k$,
	\[
		\sum_k(p_k-\mc_k)\,dy_k=(\mu-1)\Bigl(\sum_kdI_k+\sum_k\mc_ky_k\,\dlog A_k\Bigr).
	\]
	Without intermediate inputs, $\sum_kdI_k=\sum_fw_f\sum_kdl_{kf}=-R$, which gives the expression for $dW$. In Proposition~\ref{prop:factors}, $\tilde\lambda=\lambda$ and $\widetilde\Omega_{kf}=\mu\Omega_{kf}$, so $\T_k=E\lambda_k\,\dlog A_k=\mu\mc_ky_k\,\dlog A_k$ and $E\widetilde\Lambda_f=\mu E\Lambda_f=\mu w_fL_f$. The first two terms of its expression for $dW$ therefore sum to the same value, and $\AEff=0$. None of these steps uses the demand system, Assumption~\ref{ass:conduct}, or constant returns.

	For Part~2, consider the economy in which the production employment of every factor is $(1-x)L_f$ and technologies are unchanged. Scaling every quantity by $1-x$, and every price by the common factor that keeps the marginal utility of the numeraire at one, satisfies the equilibrium conditions. With constant returns, each producer's cost-minimizing input shares depend only on relative input prices, and with homothetic final demand, the household's expenditure shares depend only on relative prices. Markups that depend only on expenditure shares are therefore unchanged, and so are prices relative to marginal costs. Markets clear because every quantity is scaled by the same factor, and the common scaling of prices scales profits without changing producers' choices. This scaled equilibrium is the selection of Section~\ref{sec:environment}, and it is the only one nearby when the Jacobian of the equilibrium conditions is nonsingular. The withdrawal therefore leaves the allocation $\mathcal X$ unchanged and contributes nothing to $\AEff$, which in Proposition~\ref{prop:factors} is the welfare effect of the change in the allocation. Because first-order changes add, $\AEff$ is that of the technology change alone. With constant returns, $\sum_f\widetilde\Lambda_f=1$. If $dL^R_f/L_f=x$ for every $f$, then $R=xE\sum_f\Lambda_f$, so the second term of $dW$ in Proposition~\ref{prop:factors} is $-Ex=-R/\sum_f\Lambda_f$. Without spillovers or induced technology changes, $\sum_k\T_k=\T_i=(\mu_i\tilde\lambda_i/\lambda_i)\Pv_i$ by equation~\eqref{eq:liu}, which gives the first expression for $dW$, and substituting $R=\Pv_i+\IP_i$ gives the second. Next, $w_fl_{kf}=E\lambda_k\widetilde\Omega_{kf}/\mu_k$ gives
	\[
		\sum_k\Bigl(\sum_fw_fl_{kf}\Bigr)\frac{\mu_k\tilde\lambda_k}{\lambda_k}=E\sum_k\tilde\lambda_k\sum_f\widetilde\Omega_{kf}=E\sum_f\widetilde\Lambda_f=E,
	\]
	so the mean of the ratios weighted by factor payments is $E/\sum_fw_fL_f=1/\sum_f\Lambda_f$.
\end{proof}

\section*{Online Appendix}
\onlineappendix
This Online Appendix contains the R\&D game behind Theorem~\ref{thm:dynamic} and its proofs (Section~\ref{oa:dynamic}), further results and proofs on allocative efficiency (Section~\ref{oa:allocative}), extensions of the main results and the remaining proofs (Section~\ref{oa:extensions}), example economies (Section~\ref{oa:examples}), and details of the 2015 economy (Section~\ref{oa:application}). Numbers of sections, equations, and results without the prefix OA refer to the main paper.\ifjournal{} Further examples and results are in the working-paper version of this paper, available from the authors.\fi

\section{R\&D Dynamics}\label{oa:dynamic}

This section sets up the R\&D game of Section~\ref{sec:subsidy} and proves Theorem~\ref{thm:dynamic}, which separates the spread of the knowledge created by an extra unit of one firm's R\&D from the R\&D responses that it triggers and values each response at the responding firm's R\&D wedge.

\subsection{The game and the decomposition}\label{oa:dynamic-game}

\paragraph{Environment.} Knowledge $\bm z$ lies in an open subset $\mathcal Z$ of a Euclidean space and follows $\dot{\bm z}=\bm h(\bm z,\bm r)$, where firm $k$ chooses R\&D effort $r_k$ at cost $\Gamma_k(r_k)$ in units of the numeraire. At each state the static economy of Section~\ref{sec:environment} is in an equilibrium that satisfies Assumptions~\ref{ass:inputs} and~\ref{ass:conduct}, and each firm is a producer in it. Knowledge enters technologies directly, $\tau_k(\bm z)$, which allows spillovers, or through each firm's choice of technique given its own knowledge $z_k$ (Remark~\ref{rem:technique}), in which case $\Pv_k=\omega_k\,dz_k$. Let $W(\bm z)$ and $\pi_k(\bm z)$ be static welfare and firm $k$'s static profit. Firms use Markov strategies $\bm r=\bm\chi(\bm z)$, so knowledge follows $\dot{\bm z}=\bm f(\bm z)\equiv\bm h(\bm z,\bm\chi(\bm z))$ along the base path $\bm z_t$ from $\bm z_0$. With discount rate $\rho>0$, household welfare and firm $k$'s value, both net of R\&D costs, are
\[
	V(\bm z_0)=\int_0^\infty e^{-\rho t}\Bigl[W(\bm z_t)-\sum_k\Gamma_k\bigl(\chi_k(\bm z_t)\bigr)\Bigr]dt,
	\qquad
	V^k(\bm z_0)=\int_0^\infty e^{-\rho t}\bigl[\pi_k(\bm z_t)-\Gamma_k\bigl(\chi_k(\bm z_t)\bigr)\bigr]dt .
\]

\paragraph{Marginal returns to R\&D.} At any state $\bm z$, with $\bm h_{r_k}$ evaluated at $(\bm z,\bm\chi(\bm z))$, the social and private marginal returns to firm $k$'s effort, firm $k$'s R\&D wedge, and firm $i$'s stake in firm $j$'s effort, the marginal value to $i$ of one more unit of $j$'s effort, are
\[
	\SMR_k\equiv\bm h_{r_k}^\top V_{\bm z},
	\qquad
	\PMR_k\equiv\bm h_{r_k}^\top V^k_{\bm z},
	\qquad
	\Delta_k\equiv\SMR_k-\PMR_k,
	\qquad
	\xi_{ij}\equiv\bm h_{r_j}^\top V^i_{\bm z}.
\]
They exclude the cost of the extra unit of effort and include the R\&D costs that it induces, those of every firm in $\SMR_k$ and firm $k$'s own in $\PMR_k$. For the innovator $i$, one more unit of effort moves the state in the direction $\bm v_i\equiv\bm h_{r_i}(\bm z_0,\bm\chi(\bm z_0))$, so $\SMR_i=V_{\bm z}(\bm z_0)^\top\bm v_i$ and $\PMR_i=V^i_{\bm z}(\bm z_0)^\top\bm v_i$ at $\bm z_0$.

\paragraph{Fixed-effort response.} Along the base path, the equilibrium and fixed-effort Jacobians of the knowledge dynamics are $\mathcal J(t)=\bm h_{\bm z}+\sum_k\bm h_{r_k}\chi_{k,\bm z}$ and $\mathcal J_0(t)=\bm h_{\bm z}$, evaluated at $(\bm z_t,\bm\chi(\bm z_t))$. Let $\Pi(t,s)$ and $\Pi_0(t,s)$ be their state-transition matrices, with $\partial_t\Pi(t,s)=\mathcal J(t)\Pi(t,s)$ and $\Pi(s,s)=I$, and likewise for $\Pi_0$. In the fixed-effort response $d\bm z^0_t=\Pi_0(t,0)\bm v_i$, the extra knowledge spreads with every firm's effort held at its base-path value; in the equilibrium response $d\bm z_t=\Pi(t,0)\bm v_i$, it also changes every firm's effort. Firm $k$'s R\&D response,
\[
	dr^0_{k,t}\equiv\chi_{k,\bm z}(\bm z_t)\,d\bm z^0_t ,
\]
is the change in effort that its strategy prescribes for $d\bm z^0_t$. Because $d\bm z^0_t$ excludes the knowledge that the R\&D responses themselves create, $dr^0_{k,t}$ differs from the change in effort along the equilibrium response, and Theorem~\ref{thm:dynamic} values that knowledge, and the further responses it triggers, through the weights $\Delta_k$ and $\xi_{ij}$.

\begin{assumption}[Markov perfect equilibrium]
	\label{ass:mpe}
	Given the other firms' strategies, $\chi_k$ maximizes firm $k$'s value, and the maximum in firm $k$'s Hamilton--Jacobi--Bellman equation is interior along the base path,
	\[
		\Gamma_k'\bigl(\chi_k(\bm z_t)\bigr)=\bm h_{r_k}\bigl(\bm z_t,\bm\chi(\bm z_t)\bigr)^\top V^k_{\bm z}(\bm z_t)\qquad\text{for all }t\ge0,
	\]
	where $V^k_{\bm z}$ exists by Lemma~\ref{lem:closedloop}.
\end{assumption}

\begin{assumption}[Dynamic regularity]
	\label{ass:regular}
	$\bm h$, $\bm\chi$, $W$, $\pi_k$, and $\Gamma_k$ are continuously differentiable, and the paths of $\dot{\bm z}=\bm f(\bm z)$ from a neighborhood of $\bm z_0$ stay in $\mathcal Z$. There are constants $K$, $\gamma_\Pi$, and $\gamma_F$ with $\gamma_\Pi+\gamma_F<\rho$ such that, uniformly over initial states in that neighborhood, $\|\Pi(t,s)\|$ and $\|\Pi_0(t,s)\|$ are at most $Ke^{\gamma_\Pi(t-s)}$, the gradients of $W$, $\pi_k$, and $\Gamma_k\circ\chi_k$ along the paths are at most $Ke^{\gamma_Ft}$, and $\|\bm h_{r_k}\|\,\|\chi_{k,\bm z}\|$ is at most $K$.
\end{assumption}

Assumption~\ref{ass:mpe} enters only through the first-order conditions $\PMR_k=\Gamma_k'$, which make $\Delta_k$ the social marginal value of firm $k$'s effort less its marginal cost.

Under these assumptions, Theorem~\ref{thm:dynamic} in Section~\ref{sec:subsidy} decomposes $\SMR_i$ and $\PMR_i$, with $\T_{k,t}$, $\Pv_{i,t}$, $\AEff_t$, and $\IP_{i,t}$ the terms of Theorem~\ref{thm:main} for the change in technologies that the fixed-effort response $d\bm z^0_t$ induces at date $t$.

\begin{proof}[Proof of Theorem~\ref{thm:dynamic}]
	Lemma~\ref{lem:openloop} below writes $\SMR_i$ and $\PMR_i$ as discounted integrals of the static change in $W$ or $\pi_i$ induced by $d\bm z^0_t$ and of the responses $dr^0_{k,t}$, weighted by $\bm h_{r_k}^\top V_{\bm z}-\Gamma_k'$ and by $\bm h_{r_k}^\top V^i_{\bm z}-\mathbf 1\{k=i\}\,\Gamma_i'$. By Assumption~\ref{ass:mpe}, $\Gamma_k'=\bm h_{r_k}^\top V^k_{\bm z}$, so these weights are $\Delta_k$, zero for $k=i$, and $\xi_{ij}$ for $k=j\neq i$. Theorem~\ref{thm:main} at each date writes the static changes in $W$ and $\pi_i$ as $\sum_k\T_{k,t}+\AEff_t$ and $\Pv_{i,t}+\IP_{i,t}$.
\end{proof}

Section~\ref{sec:subsidy} interprets equation~\eqref{eq:dynamic}. By the envelope theorem the innovator's own response adds $\Delta_i\,dr^0_{i,t}$ to the social return and nothing to the private return, and rivals' responses enter the private return only through the stake $\xi_{ij}$, which extends to R\&D the strategic effect of rivals' output \citep{BranderSpencer1983,JunVives2004}, whereas society values them at $\Delta_j$. This is the logic of feedback Nash costate equations \citep{DocknerEtAl2000} and generalized Euler equations \citep{KrusellKuruscuSmith2002,KleinKrusellRiosRull2008}.

With proportional R\&D subsidies $\varsigma_j$, the decomposition holds with each response $dr^0_{j,t}$, the innovator's own included, valued socially at $\Delta_j-\varsigma_j\Gamma_j'$, the social marginal value of firm $j$'s effort less its marginal cost.

\begin{corollary}[R\&D stocks in a stationary economy]\label{cor:stationary}
	Under Assumptions~\ref{ass:mpe} and~\ref{ass:regular}, suppose that $\bm h(\bm z,\bm r)=\kappa\bm r-\delta\bm z$, with $\kappa>0$ and $\rho+\delta>0$, that technologies are functions $\tau(\bm z)$ of knowledge, and that the base path is a steady state, $\bm z_t=\bm z_0$. Let $dW=\sum_k\T_k+\AEff$ and $d\pi_i=\Pv_i+\IP_i$ be the static returns of Theorem~\ref{thm:main} for the change in technologies $\tau_{\bm z}(\bm z_0)\,\bm v_i$, and let $dr_k\equiv\chi_{k,\bm z}(\bm z_0)\,\bm v_i$. Then
	\[
		\SMR_i=\frac{dW+\sum_k\Delta_k\,dr_k}{\rho+\delta},
		\qquad
		\PMR_i=\frac{d\pi_i+\sum_{j\neq i}\xi_{ij}\,dr_j}{\rho+\delta},
	\]
	with $\Delta_k$ and $\xi_{ij}$ evaluated at $\bm z_0$.
\end{corollary}

\begin{proof}
	Here $\mathcal J_0=-\delta I$, so $d\bm z^0_t=e^{-\delta t}\bm v_i$. At a steady state the static economy, $\chi_{k,\bm z}$, $\Delta_k$, and $\xi_{ij}$ are constant, and the terms of Theorem~\ref{thm:main} are linear in the change in technologies, so each integrand in Theorem~\ref{thm:dynamic} is $e^{-\delta t}$ times its value at $t=0$, with $dr^0_{k,t}=e^{-\delta t}dr_k$, and $\int_0^\infty e^{-(\rho+\delta)t}\,dt=1/(\rho+\delta)$.
\end{proof}

\subsection{Proofs}\label{oa:dynamic-proofs}

\begin{lemma}[Marginal values along the equilibrium path]\label{lem:closedloop}
	Let Assumption~\ref{ass:regular} hold, let $F$ be $W$, $\pi_k$, or $\Gamma_k\circ\chi_k$ for some firm $k$, and suppose that $\mathcal V_F(\bm z_0)\equiv\int_0^\infty e^{-\rho t}F(\bm z_t)\,dt$, its discounted value along the path of $\dot{\bm z}=\bm f(\bm z)$ from $\bm z_0$, is finite. Then, at $\bm z_0$ and at every state $\bm z_s$ on the base path,
	\[
		\nabla\mathcal V_F(\bm z_0)=\int_0^\infty e^{-\rho t}\,\Pi(t,0)^\top\nabla F(\bm z_t)\,dt,
		\qquad
		\nabla\mathcal V_F(\bm z_s)=\int_s^\infty e^{-\rho(t-s)}\,\Pi(t,s)^\top\nabla F(\bm z_t)\,dt .
	\]
	The same formulas hold for $V$, with $F=W-\sum_j\Gamma_j\circ\chi_j$, and for $V^k$, with $F=\pi_k-\Gamma_k\circ\chi_k$.
\end{lemma}

\begin{proof}
	The path is continuously differentiable in its initial state, with derivative solving the variational equation with coefficient matrix $\bm f_{\bm z}=\mathcal J(t)$, so $\partial\bm z_t/\partial\bm z_0=\Pi(t,0)$. For a unit vector $\bm e$, the mean value theorem writes the difference quotient of $e^{-\rho t}F(\bm z_t)$ in the direction $\bm e$ as $e^{-\rho t}\nabla F^\top\Pi(t,0)\,\bm e$ along the path from an intermediate initial state. By Assumption~\ref{ass:regular} it is bounded by $K^2e^{(\gamma_\Pi+\gamma_F-\rho)t}$, which is integrable, and dominated convergence gives the first formula. Because $\dot{\bm z}=\bm f(\bm z)$ is autonomous, the path from $\bm z_s$ is the base path from date $s$ on, so $\mathcal V_F(\bm z_s)=\int_s^\infty e^{-\rho(t-s)}F(\bm z_t)\,dt$ and $\partial\bm z_t/\partial\bm z_s=\Pi(t,s)$. The flow over $[0,s]$ maps a neighborhood of $\bm z_0$ diffeomorphically onto a neighborhood of $\bm z_s$, so the bounds of Assumption~\ref{ass:regular} hold along paths from states near $\bm z_s$, and the same argument with the dominating function $K^2e^{\gamma_Fs}e^{(\gamma_\Pi+\gamma_F-\rho)(t-s)}$ gives the second formula. The formulas for $V$ and $V^k$ follow by linearity.
\end{proof}

\begin{lemma}[Decomposition along the fixed-effort response]\label{lem:openloop}
	Under Assumption~\ref{ass:regular}, with $\nabla W$, $\nabla\pi_i$, $V_{\bm z}$, and $V^i_{\bm z}$ evaluated at $\bm z_t$, $\bm h_{r_k}$ at $(\bm z_t,\bm\chi(\bm z_t))$, and $\Gamma_k'$ at $\chi_k(\bm z_t)$,
	\begin{align*}
		V_{\bm z}(\bm z_0)^\top\bm v_i&=\int_0^\infty e^{-\rho t}\Bigl[\nabla W^\top d\bm z^0_t+\sum_k\bigl(\bm h_{r_k}^\top V_{\bm z}-\Gamma_k'\bigr)\,dr^0_{k,t}\Bigr]dt,\\
		V^i_{\bm z}(\bm z_0)^\top\bm v_i&=\int_0^\infty e^{-\rho t}\Bigl[\nabla\pi_i^\top d\bm z^0_t+\bigl(\bm h_{r_i}^\top V^i_{\bm z}-\Gamma_i'\bigr)\,dr^0_{i,t}+\sum_{j\neq i}\bm h_{r_j}^\top V^i_{\bm z}\,dr^0_{j,t}\Bigr]dt .
	\end{align*}
\end{lemma}

\begin{proof}
	Write $\mathcal J(t)=\mathcal J_0(t)+\mathcal J_\chi(t)$, with $\mathcal J_\chi(t)\equiv\sum_k\bm h_{r_k}\chi_{k,\bm z}$ at $(\bm z_t,\bm\chi(\bm z_t))$. The difference $\Pi(t,0)-\Pi_0(t,0)$ is zero at $t=0$ and solves $\dot Y=\mathcal J(t)Y+\mathcal J_\chi(t)\Pi_0(t,0)$, so by Duhamel's formula
	\[
		\Pi(t,0)=\Pi_0(t,0)+\int_0^t\Pi(t,s)\,\mathcal J_\chi(s)\,\Pi_0(s,0)\,ds .
	\]
	Let $F=W-\sum_k\Gamma_k\circ\chi_k$. Applying this to $\bm v_i$ in the first formula of Lemma~\ref{lem:closedloop} for $V$ gives
	\[
		V_{\bm z}(\bm z_0)^\top\bm v_i=\int_0^\infty e^{-\rho t}\,\nabla F(\bm z_t)^\top d\bm z^0_t\,dt+\int_0^\infty\!\!\int_0^te^{-\rho t}\,\nabla F(\bm z_t)^\top\Pi(t,s)\,\mathcal J_\chi(s)\,d\bm z^0_s\,ds\,dt .
	\]
	By Assumption~\ref{ass:regular}, $\|\mathcal J_\chi(s)\|$ is bounded, and the integrand of the double integral is bounded by a constant times $e^{(\gamma_\Pi+\gamma_F-\rho)t}$ for $s\le t$, so its integral over $s$ is bounded by a constant times $te^{(\gamma_\Pi+\gamma_F-\rho)t}$, which is integrable, and the order of integration can be exchanged. By the second formula of Lemma~\ref{lem:closedloop}, $\int_s^\infty e^{-\rho t}\,\nabla F(\bm z_t)^\top\Pi(t,s)\,dt=e^{-\rho s}V_{\bm z}(\bm z_s)^\top$, so the double integral equals $\int_0^\infty e^{-\rho s}\,V_{\bm z}(\bm z_s)^\top\mathcal J_\chi(s)\,d\bm z^0_s\,ds=\int_0^\infty e^{-\rho s}\sum_k\bm h_{r_k}^\top V_{\bm z}\,dr^0_{k,s}\,ds$. In the first integral, $\nabla F^\top d\bm z^0_t=\nabla W^\top d\bm z^0_t-\sum_k\Gamma_k'\,dr^0_{k,t}$, because the gradient of $\Gamma_k\circ\chi_k$ is $\Gamma_k'\,\chi_{k,\bm z}^\top$. Together these give the first formula. The argument for $V^i$ is the same with $F=\pi_i-\Gamma_i\circ\chi_i$, for which $\nabla F^\top d\bm z^0_t=\nabla\pi_i^\top d\bm z^0_t-\Gamma_i'\,dr^0_{i,t}$.
\end{proof}

\section{Allocative Efficiency: Further Results and Proofs}\label{oa:allocative}

\subsection{One market under Cournot conduct}\label{oa:onecournot}

Equation~\eqref{eq:onefixed} takes the innovator's markup response as given. Under Cournot conduct, Proposition~\ref{prop:markups} determines it. For $\eta=1$ and rivals that form a fringe with constant markups, \citet{BaqaeeFarhi2020} show that the rise in the innovator's markup dampens the gain from its productivity change. The next corollary gives the change in allocative efficiency for any perceived elasticity $\eta$.

\begin{corollary}[One market, Cournot innovator]\label{cor:onecournot}
	Suppose producer $i$ has the markup of Proposition~\ref{prop:markups}, with perceived elasticity $\eta$, and all other producers keep constant markups. Then
	\[
		\frac{\AEff}{E}=(\theta-1)\,s_i\Bigl(1-\frac{\bar\mu}{\mu_i}\Bigr)\cdot\frac{1-\zeta_i(1-s_i)}{1+\zeta_i(\theta-1)(1-s_i)}\,\dlog A_i ,
	\]
	and with $\eta=1$ the last factor is $(1-s_i)/\bigl(1+s_i(\theta-1)\bigr)$.
\end{corollary}

Cournot conduct sets the pass-through to $1/\bigl(1+\zeta_i(\theta-1)(1-s_i)\bigr)$. When $\theta>1$ and the innovator's markup exceeds the mean, a markup that rises with the share works against allocative efficiency twice, because less labor moves to the innovator and the higher markup at a producer already above the mean lowers allocative efficiency directly. The factor is negative only when $\zeta_i(1-s_i)>1$, which requires $\eta<1$.

\subsection{Several markets}\label{oa:nests}

This subsection adds further markets, between which labor moves, and lets the innovator's rivals be strategic producers as well, so that their markups also respond. Final demand is a CES with elasticity $\eta$ over nests. Each nest $n$ is a CES with elasticity $\theta>1$ over producers that use only labor, and its Domar weight $\lambda_n$ is its share of final expenditure. The economy's and nest $n$'s harmonic mean markups are $\bar\mu=(\sum_k\lambda_k/\mu_k)^{-1}$ and $\bar\mu_n=(\sum_{k\in n}s_k/\mu_k)^{-1}$.

\begin{corollary}[Cournot producers in several nests]\label{cor:nests}
	Suppose producers have the markups of Proposition~\ref{prop:markups}, every strategic producer competing in quantities within its nest and perceiving the elasticity $\eta$. For producers $k$ in the innovator's nest $n$ let
	\[
		\gamma_k=\Bigl(1-\frac{\bar\mu}{\mu_k}\Bigr)(1-\zeta_k),
		\qquad
		S_n=\sum_{k\in n}\frac{s_k}{1+(\theta-1)\zeta_k},
		\qquad
		\bar\gamma_n=\frac{1}{S_n}\sum_{k\in n}\frac{s_k\gamma_k}{1+(\theta-1)\zeta_k} .
	\]
	Then a productivity change $\dlog A_i$ gives
	\[
		\frac{\AEff}{\lambda_iE}=\frac{1}{1+(\theta-1)\zeta_i}\Bigl[(\theta-1)(\gamma_i-\bar\gamma_n)+\frac{\eta-1}{S_n}\Bigl(1-\frac{\bar\mu}{\bar\mu_n}\Bigr)\Bigr]\dlog A_i .
	\]
\end{corollary}

With one nest and only the innovator strategic, the corollary reduces to Corollary~\ref{cor:onecournot}. The prefactor $1/\bigl(1+(\theta-1)\zeta_i\bigr)$ is the pass-through of the innovator's cost reduction into its price relative to the wage at a given price of its nest, and it scales both the change in the innovator's share and the fall in its nest's price. The bracket has two terms. The first collects the changes in shares and markups within nest $n$ at a given share of final expenditure on the nest, and the second values the change in that share.

In the first term, $\gamma_k$ combines two forces for each producer $k$ in the nest. The markup gap $1-\bar\mu/\mu_k$ measures how valuable it is to move labor to $k$. The factor $1-\zeta_k$ is the elasticity of $k$'s wage bill per unit of its nest's spending, $s_k/\mu_k$, with respect to its share, since its markup moves with the share at the elasticity $\zeta_k$. When $\zeta_i>1$, which with $\eta=1$ means $s_i>1/2$, a higher share lowers $s_i/\mu_i$, and $\gamma_i<0$ for an innovator whose markup exceeds the mean. The innovation raises the innovator's share, which $\gamma_i$ values, and lowers the nest's price relative to the wage, which lowers every producer's share in proportion to $(\theta-1)/\bigl(1+(\theta-1)\zeta_k\bigr)$. The average $\bar\gamma_n$ values these declines.

The second term vanishes when $\eta=1$. When $\eta>1$, the fall in the nest's price raises its share of final expenditure, which raises $\AEff$ if the nest's mean markup exceeds the economy's. Labor nonetheless moves between nests even when $\eta=1$, because a change in the nest's markups changes the labor it employs at a given share of spending. The reference markup in $\gamma_k$ is therefore the economy's mean markup, not the nest's. Unlike in Corollary~\ref{cor:onecournot}, where only the innovator is a strategic producer in a single market, the sign of $\AEff$ can reverse with $\eta=1$.

\begin{table}[!t]
	\centering
	\caption{Innovation by a dominant producer}
	\label{tab:dominant}
	\small
	\begin{tabular}{@{}rrrrrrr@{}}
		\toprule
		Share $s_i$ & Markup $\mu_i$ & $\zeta_i$ & $(\T_i-\Pv_i)/\lambda_iE$ & $\AEff/\lambda_iE$ & $\IP_i/\lambda_iE$ & $dW/d\pi_i$ \\
		\midrule
		\IWXDomRowOne
		\IWXDomRowTwo
		\IWXDomRowThree
		\IWXDomRowFour
		\bottomrule
	\end{tabular}

	\vspace{0.3em}
	\parbox{0.92\linewidth}{\footnotesize Notes: The large producer and four small ones form a nest with a negligible share of final expenditure, in a Cobb--Douglas continuum of nests whose other nests have five equal producers. All producers compete in quantities within their nests; $\theta=5$ and $\eta=1$. The large producer's productivity changes by $\dlog A_i=1$. The returns are $dW=\lambda_iE+\AEff$ and $d\pi_i=\Pv_i+\IP_i$, with $\Pv_i=\lambda_iE/\mu_i$. The large producer's share varies with its demand weight (Section~\ref{oa:dominant}).}
\end{table}

Table~\ref{tab:dominant} shows the reversal for a dominant producer. As its share rises from \IWXDomShareLow{} to \IWXDomShareHigh, $\AEff$ per unit of its sales falls from \IWXDomAELow{} to \IWXDomAEHigh. Its innovation raises its share and markup and therefore its nest's mean markup, so at every share it releases labor from its nest to the other, less marked-up nests. Both this release and the gap between the mean markup of its nest and that of the other nests grow with its share, and at the two largest shares the loss from moving labor between the nests outweighs the gain within the nest. Replacing $\bar\mu$ by $\bar\mu_n$ in the corollary values the net flow of labor between nests at zero and gives \IWXDomNestMeanHigh{} per unit of sales at the largest share.

With strategic rivals in heterogeneous industries, the induced price term need not move with the change in allocative efficiency. In Table~\ref{tab:dominant}, as the dominant producer's share rises from \IWXDomShareLow{} to \IWXDomShareHigh, $\AEff$ per unit of its sales turns negative (Section~\ref{sec:nests}), while the induced price term stays positive and moves only from \IWXDomOSLow{} to \IWXDomOSHigh. Its share of the private return rises from \IWXDomOSShareLow{} to \IWXDomOSShareHigh. The induced price term is thus no measure of the change in allocative efficiency, even for a dominant producer.

The proof of Corollary~\ref{cor:nests} also gives $\IP_i$ in closed form (equation~\eqref{eq:os-closed} in Section~\ref{oa:allocative-proofs}), so $dW/d\pi_i$, which still rises with the share in Table~\ref{tab:dominant}, depends only on shares, markups, and elasticities.

\subsection{A Cournot producer facing a fringe}\label{oa:fringe}

In the economy of \citet{BaqaeeFarhi2020}, labor is the only factor and final demand is a Cobb--Douglas aggregate of a continuum of identical industries, each with a negligible share of final expenditure. In each industry, goods are combined with elasticity $\theta$, and one producer competes in quantities, perceiving $\eta=1$, against a fringe of atomistic producers priced at $\theta/(\theta-1)$, so its markup is $1/\mu_i=(1-s_i)(1-1/\theta)$.

\begin{corollary}[Cournot producer facing a fringe]\label{cor:fringe}
	In this economy, a productivity change $\dlog A_i$ of the Cournot producer of one industry gives, per unit of its sales,
	\[
		\frac{dW}{\lambda_iE}=\frac{1+(2\theta-1)s_i}{(1+s_i)\bigl(1+(\theta-1)s_i\bigr)}\,\dlog A_i,
		\qquad
		\frac{d\pi_i}{\lambda_iE}=\frac{(\theta-1)(1-s_i)\bigl(1+2(\theta-1)s_i\bigr)}{\theta\bigl(1+(\theta-1)s_i\bigr)}\,\dlog A_i ,
	\]
	and $\IP_i=(\theta-1)(1+s_i)\,\AEff/\theta$. The ratio
	\[
		\frac{dW}{d\pi_i}=\frac{\theta\bigl(1+(2\theta-1)s_i\bigr)}{(\theta-1)(1-s_i^2)\bigl(1+2(\theta-1)s_i\bigr)}
	\]
	tends to $\theta/(\theta-1)$ as $s_i\to0$ and rises with $s_i$.
\end{corollary}

The change in allocative efficiency, positive and dampened by the rise in the innovator's markup, is that of \citet{BaqaeeFarhi2020}. The induced price term moves in proportion to it and offsets part of it in the wedge, and more than all of it when $s_i>1/(\theta-1)$, although it consists of transfers rather than of any part of that change. Both vanish relative to the innovator's sales as its share goes to zero, the change in allocative efficiency here because the fringe charges the markup of a small Cournot producer. The ratio rises with the share, as $\T_i/\Pv_i=\mu_i$ does, a gradient that requires the fringe's mass to be fixed. When the fringe enters freely, the ratio no longer depends on the share (Remark~\ref{rem:entry}).

\subsection{Proofs}\label{oa:allocative-proofs}

The proof of Corollary~\ref{cor:uniform} uses only equation~\eqref{eq:ae}. The other proofs use the nested-CES setting of Section~\ref{sec:ces}, whose equilibrium is that of the economy of \citet{BaqaeeFarhi2020} with wedges equal to the markups, and this section also derives Proposition~\ref{prop:ae} and equation~\eqref{eq:onefixed} from their formulas.

\begin{proof}[Proof of Corollary~\ref{cor:uniform}]
	Without intermediate inputs, $\widetilde\Omega$ and $\Omega$ have no entries between goods, so $\tilde\lambda=\lambda=\beta$, and equation~\eqref{eq:ae} becomes
	\[
		\AEff=\sum_k(\mu_k-1)\,dI_k,
		\qquad
		dI_k=\sum_fw_f\,dl_{kf}.
	\]
	With fixed factor supplies, $\sum_kdI_k=\sum_fw_f\,dL_f=0$, so subtracting $(c-1)\sum_kdI_k$ gives $\AEff=\sum_k(\mu_k-c)\,dI_k$ for every constant $c$. With one factor, $dI_k=wl_k\,\dlog l_k$, and with constant returns the cost of producer $k$ is $wl_k=\mc_ky_k=E\lambda_k/\mu_k$. Hence $\AEff/E=\sum_k\lambda_k(1-c/\mu_k)\,\dlog l_k$ for every $c$, and $c=\bar\mu$ gives the second expression, whose weights sum to zero because $\sum_k\lambda_k=1$. If all markups equal a common $\mu$, taking $c=\mu$ gives $\AEff=0$. Neither step uses the demand system or the response of markups.
\end{proof}

\paragraph{Proposition~\ref{prop:ae}.} With one factor, $\widetilde\Lambda_L=1$, so the expression of Theorem~1 of \citet{BaqaeeFarhi2020} in Section~\ref{sec:lemmas} becomes $\AEff/E=-\tilde\lambda^\top\dlog\bm\mu-\dlog\Lambda_L$. Their equations~(10) and~(13), the one-factor forms of their Propositions~2 and~3, multiplied by $\dlog A_k$ and $\dlog\mu_k$ and summed over producers, give $\dlog\Lambda_L$ and hence equation~\eqref{eq:ae-prim}.

\begin{proof}[Proof of Proposition~\ref{prop:markups}]
	Let $n=n(k)$. The nest is competitive with constant returns, so $p_k=p_n\,\partial y_n/\partial x_{nk}$, where $x_{nk}=y_k$, and with a CES technology $\partial y_n/\partial x_{nk}$ is proportional to $(y_n/y_k)^{1/\theta_n}$. Holding fixed the quantities of the nest's other inputs, $y_n$ responds to $y_k$ with elasticity $s_k$, and producer $k$ perceives $p_n$ to respond to $y_n$ with elasticity $-1/\eta_n$. Its perceived inverse demand therefore has elasticity
	\[
		\frac{\partial\log\varphi_k}{\partial\log y_k}=-\frac{s_k}{\eta_n}-\frac{1-s_k}{\theta_n}=-\frac{1}{\varepsilon_k},
	\]
	and the first-order condition of Assumption~\ref{ass:conduct}, $p_k(1-1/\varepsilon_k)=\mc_k$, gives the markup rule.

	With $\eta_n$ constant, differentiating the rule gives
	\[
		\dlog\mu_k=\mu_k\,d(1/\varepsilon_k)=\mu_ks_k\Bigl(\frac{1}{\eta_n}-\frac{1}{\theta_n}\Bigr)\dlog s_k=\zeta_k\,\dlog s_k .
	\]
	CES demand gives $\dlog s_k=(1-\theta_n)\bigl(\dlog p_k-\sum_m\widetilde\Omega_{nm}\,\dlog p_m\bigr)$, the sum running over all inputs of the nest. By the price equation of Section~\ref{sec:ces}, $\dlog p_m=u_m+\dlog w$ with $\bm u=\widetilde\Psi_{\cdot\mathcal K}(\dlog\bm\mu-\dlog\bm A)$, and $u_L=0$ for labor, whose price is $w$. Because the nest's cost shares sum to one, $\dlog w$ cancels, so $\dlog s_k=(\mathcal D\bm u)_k$ and
	\[
		\dlog\bm\mu=Z\mathcal D\widetilde\Psi_{\cdot\mathcal K}\,(\dlog\bm\mu-\dlog\bm A),
	\]
	where producers that are not strategic have $\zeta_k=0$ and constant markups. Solving for $\dlog\bm\mu$ gives the result. Under price competition, differentiating $\mu_k=\varepsilon_k/(\varepsilon_k-1)$ with $\varepsilon_k=\theta_n(1-s_k)+\eta_ns_k$ gives $\dlog\mu_k=(\theta_n-\eta_n)s_k\,\dlog s_k/\bigl(\varepsilon_k(\varepsilon_k-1)\bigr)=\zeta_k\,\dlog s_k$ with $\zeta_k=\mu_ks_k(\theta_n-\eta_n)/\varepsilon_k^2$, and the rest of the proof applies.
\end{proof}

\paragraph{The within-sector case.} When no nest buys labor or the output of other nests directly and the suppliers of each nest have the same input shares, a change in productivities or markups outside a nest moves the prices of all its suppliers in the same proportion. It leaves their within-nest shares and markups unchanged, so $B$ reduces to the within-sector solution of \citet{BursteinCarvalhoGrassi2025} even with input--output linkages.

\begin{proof}[Proof of Corollary~\ref{cor:cd}]
	If $\theta_j=1$ for every $j$, the substitution term of \eqref{eq:ae-prim} vanishes. The factor $1-\theta_{n(k)}$ in the map $\mathcal D$ is zero, so $\mathcal D=0$ and $B=0$, and markups do not change. The direct term of \eqref{eq:ae-prim} then vanishes as well, so $\AEff=0$, with fixed markups or with the markups of Proposition~\ref{prop:markups}. With $\theta_n=1$ the markup rule gives $1-1/\varepsilon_k=s_k(1-1/\eta_n)$, which is positive when $\eta_n>1$, so these markups are finite.
\end{proof}

\paragraph{Equation~\eqref{eq:onefixed}.} In the horizontal economy of Section~\ref{sec:onemarket}, $\widetilde\Psi$ is the identity on producers, $\Psi_{kL}=1/\mu_k$, and $\Lambda_L=1/\bar\mu$. With their elasticity $\theta_0$ and sales shares $\lambda_k$ equal to $\theta$ and $s_k$, the harmonic mean markup of \citet{BaqaeeFarhi2020} is $\bar\mu$, and their equations~(5) and~(6) give the effects on output of a productivity change and of a markup change. Combining them net of the technology effect $Es_i\,\dlog A_i$ gives equation~\eqref{eq:onefixed}.

\begin{proof}[Proof of Corollary~\ref{cor:onecournot}]
	Only $\mu_i$ changes. The innovator's price relative to the price $p_0$ of the final-demand aggregator changes by $(1-s_i)(\dlog\mu_i-\dlog A_i)$, so CES demand gives
	\[
		\dlog s_i=(1-\theta)(1-s_i)(\dlog\mu_i-\dlog A_i).
	\]
	Together with $\dlog\mu_i=\zeta_i\,\dlog s_i$ from the proof of Proposition~\ref{prop:markups}, this gives $\dlog\mu_i=B_{ii}\,\dlog A_i$ with
	\[
		B_{ii}=\frac{\zeta_i(\theta-1)(1-s_i)}{1+\zeta_i(\theta-1)(1-s_i)} .
	\]
	Substituting into equation~\eqref{eq:onefixed} replaces the bracket by $\bigl[(\theta-1)-\theta B_{ii}\bigr]\dlog A_i$, and
	\[
		(\theta-1)-\theta B_{ii}=\frac{(\theta-1)\bigl(1-\zeta_i(1-s_i)\bigr)}{1+\zeta_i(\theta-1)(1-s_i)},
	\]
	which gives the result. With $\eta=1$ the markup rule gives $1/\mu_i=(1-s_i)(1-1/\theta)$, so $\zeta_i=\mu_is_i(1-1/\theta)=s_i/(1-s_i)$, and the factor becomes $(1-s_i)/\bigl(1+s_i(\theta-1)\bigr)$. The markup response is then $B_{ii}=s_i(\theta-1)/\bigl(1+s_i(\theta-1)\bigr)$, that of \citet{BaqaeeFarhi2020}, and the factor is the ratio of the changes in allocative efficiency with and without this response in their example.
\end{proof}

\begin{proof}[Proof of Corollary~\ref{cor:nests}]
	Let the final-demand aggregator be the numeraire, $p_0=1$, and let $p_n$ be the price of nest $n$. Producers in other nests keep their shares and markups. Their productivities do not change, so their prices all move with the wage, and so does the price of every other nest. Within nest $n$, $\dlog\mu_k=\zeta_k\,\dlog s_k$ by the proof of Proposition~\ref{prop:markups}, $\dlog s_k=(1-\theta)(\dlog p_k-\dlog p_n)$, and $\dlog p_k=\dlog\mu_k-\dlog A_k+\dlog w$. Solving for $\dlog s_k$, and using $\sum_{k\in n}s_k\,\dlog s_k=0$ to find the nest's price,
	\[
		\dlog s_k=\frac{\theta-1}{1+(\theta-1)\zeta_k}\bigl(\dlog A_k+\dlog(p_n/w)\bigr),
		\qquad
		\dlog(p_n/w)=-\frac{s_i\,\dlog A_i}{\bigl(1+(\theta-1)\zeta_i\bigr)S_n},
	\]
	with $\dlog A_k=0$ for $k\neq i$. Since $\dlog p_0=\lambda_n\,\dlog p_n+(1-\lambda_n)\,\dlog w=0$, the wage changes by $\dlog w=-\lambda_n\,\dlog(p_n/w)$, and CES final demand gives
	\[
		\dlog\lambda_n=(1-\eta)(1-\lambda_n)\,\dlog(p_n/w),
		\qquad
		\dlog\lambda_{n'}=-(1-\eta)\lambda_n\,\dlog(p_n/w)\quad\text{for }n'\neq n .
	\]

	Nests and the final-demand aggregator are competitive, and producers use only labor, so $\tilde\lambda=\lambda$ and equation~\eqref{eq:ae} gives $\AEff=\sum_k(\mu_k-1)\,w\,dl_k$, the sum running over producers. As in the proof of Corollary~\ref{cor:uniform}, $wl_k=E\lambda_k/\mu_k$ and $\sum_kdl_k=0$ give $\AEff/E=\sum_k\lambda_k(1-\bar\mu/\mu_k)\,\dlog l_k$, with weights that sum to zero. Since $\dlog l_k=\dlog\lambda_k-\dlog\mu_k+\dlog(E/w)$, the common term drops out, and
	\[
		\frac{\AEff}{E}=\sum_k\lambda_k\Bigl(1-\frac{\bar\mu}{\mu_k}\Bigr)(\dlog\lambda_k-\dlog\mu_k),
		\qquad
		\dlog\lambda_k=\dlog\lambda_{n(k)}+\dlog s_k .
	\]
	Within each nest $n'$, the terms in $\dlog\lambda_{n'}$ sum to $\lambda_{n'}(1-\bar\mu/\bar\mu_{n'})\,\dlog\lambda_{n'}$, and the coefficients $\lambda_{n'}(1-\bar\mu/\bar\mu_{n'})$ sum to zero across nests. Every nest other than $n$ has the same $\dlog\lambda_{n'}$, so these terms equal
	\[
		\lambda_n\Bigl(1-\frac{\bar\mu}{\bar\mu_n}\Bigr)\bigl(\dlog\lambda_n-\dlog\lambda_{n'}\bigr)=(1-\eta)\lambda_n\Bigl(1-\frac{\bar\mu}{\bar\mu_n}\Bigr)\dlog(p_n/w).
	\]
	The remaining terms come from nest $n$, where $\dlog s_k-\dlog\mu_k=(1-\zeta_k)\,\dlog s_k$, and equal
	\[
		\lambda_n\sum_{k\in n}s_k\gamma_k\,\dlog s_k=\lambda_n(\theta-1)\Bigl[\frac{s_i\gamma_i\,\dlog A_i}{1+(\theta-1)\zeta_i}+S_n\bar\gamma_n\,\dlog(p_n/w)\Bigr].
	\]
	Substituting $\dlog(p_n/w)$ into both expressions and dividing by $\lambda_iE=\lambda_ns_iE$ gives the result.
\end{proof}

\paragraph{The induced price term in closed form.} The same steps give the induced price term of the innovator in Corollary~\ref{cor:nests}. Producers have constant returns, so for a process innovation Corollary~\ref{cor:os} gives $\IP_i=p_iy_i\bigl[(1-1/\mu_i)\,\dlog(p_iy_i)-(\dlog A_i-\dlog\mu_i)/\mu_i\bigr]$. With $p_iy_i=\lambda_iE$ and $\dlog\mu_i=\zeta_i\,\dlog s_i$,
\begin{equation}
	\frac{\IP_i}{\lambda_iE}=\Bigl(1-\frac{1}{\mu_i}\Bigr)(\dlog\lambda_i+\dlog E)-\frac{\dlog A_i-\zeta_i\,\dlog s_i}{\mu_i},
	\label{eq:os-closed}
\end{equation}
where $\dlog\lambda_i=\dlog\lambda_n+\dlog s_i$, and the proof of Corollary~\ref{cor:nests} gives
\begin{align*}
	\dlog s_i&=\frac{\theta-1}{1+(\theta-1)\zeta_i}\Bigl(1-\frac{s_i}{\bigl(1+(\theta-1)\zeta_i\bigr)S_n}\Bigr)\dlog A_i,\\
	\dlog\lambda_n&=-\frac{(1-\eta)(1-\lambda_n)\,s_i}{\bigl(1+(\theta-1)\zeta_i\bigr)S_n}\,\dlog A_i .
\end{align*}
With final demand as the numeraire, $E$ is real output, so $\dlog E=dW/E$, and Theorem~\ref{thm:main} with $\tilde\lambda_i=\lambda_i$ gives $\dlog E=\lambda_i\,\dlog A_i+\AEff/E$. Together with the closed form of $\AEff$ in Corollary~\ref{cor:nests}, every term of \eqref{eq:os-closed} is therefore a function of the shares $s_k$ and $\lambda_n$, the markups, and the elasticities $\theta$ and $\eta$. In Table~\ref{tab:dominant}, $\dlog A_i=1$ and there are no spillovers, so $dW=\lambda_iE+\AEff$ and $d\pi_i=\lambda_iE/\mu_i+\IP_i$, and the ratio
\[
	\frac{dW}{d\pi_i}=\frac{1+\AEff/(\lambda_iE)}{1/\mu_i+\IP_i/(\lambda_iE)}
\]
is a function of the same objects alone.

\begin{proof}[Proof of Corollary~\ref{cor:chain}]
	With fixed markups, only the first term of \eqref{eq:ae-prim} remains, and only producer $j$ has more than one input, since $i$ uses only labor and the final-demand aggregator buys only $j$'s good. The exposure is zero for labor and $\dlog A_i$ for good $i$. Producer $i$'s revenue share of labor is $\Psi_{iL}=1/\mu_i$, and $\Lambda_L=\Psi_{jL}=(1-\omega+\omega/\mu_i)/\mu_j$, so $\sum_m\widetilde\Omega_{jm}\Psi_{mL}/\Lambda_L=\mu_j$ and
	\[
		\Cov_{\widetilde\Omega(j)}\Bigl(\widetilde\Psi\,\dlog\bm A,\frac{\Psi_{(L)}}{\Lambda_L}\Bigr)=\omega\,\dlog A_i\Bigl(\frac{1}{\mu_i\Lambda_L}-\mu_j\Bigr)=-\mu_j\,\frac{\omega(1-\omega)(1-1/\mu_i)}{1-\omega(1-1/\mu_i)}\,\dlog A_i .
	\]
	Since $\lambda_j=1$, multiplying by $-(\lambda_j/\mu_j)(\theta_j-1)$ gives the result.
\end{proof}

\begin{proof}[Proof of Corollary~\ref{cor:fringe}]
	Let $s=s_i$. The fringe's price moves with the wage, so the innovator's price relative to the price of its industry changes by $(1-s)(\dlog\mu_i-\dlog A_i)$, and, as in the proof of Corollary~\ref{cor:onecournot}, $\dlog s_i=(1-\theta)(1-s)(\dlog\mu_i-\dlog A_i)$. With $\eta=1$, $\zeta_i=s/(1-s)$ and $\dlog\mu_i=\zeta_i\,\dlog s_i$, so
	\[
		\dlog s_i=\frac{(\theta-1)(1-s)}{1+(\theta-1)s}\,\dlog A_i .
	\]
	As in the proof of Corollary~\ref{cor:nests}, $\AEff/E=\sum_k\lambda_k(1-\bar\mu/\mu_k)(\dlog\lambda_k-\dlog\mu_k)$. Industries' shares of final expenditure do not change, and producers in other industries keep their shares and markups, so only the innovator's industry contributes. There $\dlog\lambda_k=\dlog s_k$, the fringe's markup is constant, and the fringe's share changes by $\dlog(1-s)=-s\,\dlog s_i/(1-s)$. The industries are identical, so $\bar\mu$ is also the industry's harmonic mean markup, the industry's weights $s_k(1-\bar\mu/\mu_k)$ sum to zero, and
	\[
		\frac{\AEff}{\lambda_iE}=\Bigl(1-\frac{\bar\mu}{\mu_i}\Bigr)\Bigl(\frac{\dlog s_i}{1-s}-\dlog\mu_i\Bigr)=\Bigl(1-\frac{\bar\mu}{\mu_i}\Bigr)\dlog s_i .
	\]
	Since $1/\bar\mu=s/\mu_i+(1-s)(1-1/\theta)=(1-1/\theta)(1-s)(1+s)$, $1-\bar\mu/\mu_i=s/(1+s)$. Without intermediate inputs, $\T_i=\lambda_iE\,\dlog A_i$, and $dW=\T_i+\AEff$ gives the first expression.

	With constant returns, $\pi_i=(1-1/\mu_i)\lambda_iE$. The industry's share of final expenditure is constant and negligible, so $\dlog\lambda_i=\dlog s_i$ and the change in $E$ is negligible per unit of the innovator's sales. Hence
	\[
		\frac{d\pi_i}{\lambda_iE}=\Bigl(1-\frac{1}{\mu_i}\Bigr)\dlog s_i+\frac{\dlog\mu_i}{\mu_i}=\frac{1+2(\theta-1)s}{\theta}\,\dlog s_i ,
	\]
	which gives the second expression. Subtracting $\Pv_i/(\lambda_iE)=\dlog A_i/\mu_i=(1-s)(1-1/\theta)\,\dlog A_i$ gives $\IP_i/(\lambda_iE)=(\theta-1)^2s(1-s)\,\dlog A_i/\bigl(\theta(1+(\theta-1)s)\bigr)$, which is $(\theta-1)(1+s)/\theta$ times $\AEff/(\lambda_iE)$. Finally,
	\[
		\frac{d}{ds}\log\frac{dW}{d\pi_i}=\frac{2\theta-1}{1+(2\theta-1)s}-\frac{2(\theta-1)}{1+2(\theta-1)s}+\frac{2s}{1-s^2}.
	\]
	Because $x/(1+sx)$ increases with $x$, the first two terms sum to a positive number, so the ratio rises with $s$.
\end{proof}

\subsection{A dominant producer}\label{oa:dominant}

All producers use only labor with unit productivity, $y_k=l_k$, and labor supply is one. Final demand is a Cobb--Douglas aggregate of a continuum of nests, so every nest has a negligible share of final expenditure. Final demand is the numeraire, and welfare is $W=Y$. Within nest $n$, goods are combined with elasticity $\theta=5$ and demand weights $\delta_k$, $Y_n=\bigl(\sum_{k\in n}\delta_k^{1/\theta}y_k^{(\theta-1)/\theta}\bigr)^{\theta/(\theta-1)}$, so that at equal prices spending within a nest is proportional to the weights. The large producer's nest contains it and four small producers with weight $1/4$ each, and every other nest contains five producers with weight $1$ each. The large producer's weight is $1$, $6$, $25$, and $60$ in the four rows of Table~\ref{tab:dominant}, which raises its share from \IWXDomShareLow{} to \IWXDomShareHigh. Every producer competes in quantities within its nest with the markups of Proposition~\ref{prop:markups}, taking final output and its price as given, as in \citet{BaqaeeFarhi2020}. This is exact because every nest is negligible, so $\eta=1$, $1/\mu_k=(1-s_k)(1-1/\theta)$, and $\zeta_k=s_k/(1-s_k)$. The economy's harmonic mean markup $\bar\mu$ is the common markup of the producers in the other nests, each of which has the share $1/5$ of its nest. The innovation raises the large producer's productivity by $\dlog A_i=1$. The values in the table are those of Corollary~\ref{cor:nests} and equation~\eqref{eq:os-closed} in the limit in which the share of the large producer's nest in final expenditure goes to zero. Equilibria re-solved with small shares of that nest converge to them.

\section{Extensions}\label{oa:extensions}

This section extends Theorem~\ref{thm:main} to arbitrary technology parameters, derives the results for price competition, collects remarks on factor supply, technique choice, and free entry, and proves Lemma~\ref{lem:accounting} and Corollaries~\ref{cor:os} and~\ref{cor:bridge}.

\subsection{General technology}

Section~\ref{sec:environment} restricts technology to Hicks-neutral productivity and to the quality of goods sold only to final consumers. Theorem~\ref{thm:main} does not depend on these restrictions. Let the technology parameters $\tau=(\tau_1,\dots,\tau_N)$, with $\tau_m$ those of producer $m$, enter the production functions $y_k=F_k(\bm x_k,\bm l_k;\tau)$ and utility $U(\bm c;\tau)$ in any differentiable way. This covers factor-augmenting technology and the quality of an intermediate good, which enters the production functions of the good's buyers. Each $\tau_m$ may be a vector, in which case $(\partial F_k/\partial\tau_m)\,d\tau_m$ denotes the change in $F_k$ along $d\tau_m$. Assumption~\ref{ass:inputs} holds with $F_k$ in place of $A_kG_k$, so that $\mc_k\,\partial F_k/\partial x_{km}=p_m$ and $\mc_k\,\partial F_k/\partial l_{kf}=w_f$, and in Assumption~\ref{ass:conduct} producer $i$ perceives the inverse demand $\varphi_i(y_i,\tau_i)$. For an innovation $d\tau$, let
\[
	g_k\equiv\frac{1}{y_k}\sum_m\frac{\partial F_k}{\partial\tau_m}\,d\tau_m
\]
be producer $k$'s proportional output gain at fixed inputs.

\begin{proposition}[General technology]
	\label{prop:general}
	For an innovation $d\tau$, define
	\begin{align*}
		\T_m&\equiv E\sum_k\frac{\tilde\lambda_k}{y_k}\frac{\partial F_k}{\partial\tau_m}\,d\tau_m+\frac{\partial U}{\partial\tau_m}\,d\tau_m,\\
		\Pv_i&\equiv\mc_i\,y_i\,g_i+y_i\frac{\partial\varphi_i}{\partial\tau_i}\,d\tau_i,
		\qquad
		(dp_i)^{-i}\equiv dp_i-\varphi_{i,y}\,dy_i-\frac{\partial\varphi_i}{\partial\tau_i}\,d\tau_i,
	\end{align*}
	let $\AEff$ be as in Definition~\ref{def:ae}, and let $\IP_i$ be defined as in Lemma~\ref{lem:private} with this $(dp_i)^{-i}$. Then the following hold.
	\begin{enumerate}
		\item Lemma~\ref{lem:allocation} holds with these $\T_m$, Lemma~\ref{lem:private} holds with these $\Pv_i$ and $\IP_i$, and Lemma~\ref{lem:accounting} holds in the form
		\[
			dW=\sum_m\Bigl(\sum_k\mc_k\frac{\partial F_k}{\partial\tau_m}\,d\tau_m+\frac{\partial U}{\partial\tau_m}\,d\tau_m\Bigr)+\sum_k(p_k-\mc_k)\,dy_k .
		\]
		\item Theorem~\ref{thm:main} holds in the form
		\[
			dW-d\pi_i=\sum_m\T_m-\Pv_i+\AEff-\IP_i .
		\]
		\item Equations~\eqref{eq:ae} and~\eqref{eq:bridge} hold with $g_k$ in place of $\dlog A_k$, and equation~\eqref{eq:harberger} holds in the form
		\[
			dW-d\pi_i=\sum_k(p_k-\mc_k)\,dy_k+\sum_m\Bigl(\sum_k\mc_k\frac{\partial F_k}{\partial\tau_m}\,d\tau_m+\frac{\partial U}{\partial\tau_m}\,d\tau_m\Bigr)-\Pv_i-\IP_i .
		\]
		\item The innovator's own appropriability gap, $\T_i$ less the part of $\Pv_i$ due to $d\tau_i$, which is $\mc_i(\partial F_i/\partial\tau_i)\,d\tau_i+y_i(\partial\varphi_i/\partial\tau_i)\,d\tau_i$, equals
		\[
			\Bigl(\frac{E\tilde\lambda_i}{y_i}-\mc_i\Bigr)\frac{\partial F_i}{\partial\tau_i}\,d\tau_i+\Bigl(\frac{\partial U}{\partial\tau_i}+E\sum_{k\neq i}\frac{\tilde\lambda_k}{y_k}\frac{\partial F_k}{\partial\tau_i}-y_i\frac{\partial\varphi_i}{\partial\tau_i}\Bigr)d\tau_i .
		\]
	\end{enumerate}
\end{proposition}

\begin{proof}
	The three lemmas use the production functions only through the output gain at fixed inputs and the cost-minimization conditions, so their proofs go through with $y_kg_k$ in place of $y_k\,\dlog A_k$ and $\sum_m(\partial U/\partial\tau_m)\,d\tau_m$ in place of $\sum_kU_{b_k}\,db_k$. In the proof of Lemma~\ref{lem:accounting}, cost minimization gives $dI_k=\mc_k(dy_k-y_kg_k)$, and collecting $\sum_k\mc_ky_kg_k$ by the producer $m$ whose parameters change gives the values at fixed quantities in part~1. In the proof of Lemma~\ref{lem:allocation}, $\partial\log F_k/\partial\log x_{km}=\widetilde\Omega_{km}$, so at a fixed allocation $\dlog\bm y=\widetilde\Psi\bm g$ and
	\[
		dW=E\sum_k\tilde\lambda_k\,g_k+\sum_m\frac{\partial U}{\partial\tau_m}\,d\tau_m=\sum_m\T_m .
	\]
	In the proof of Lemma~\ref{lem:private}, Shephard's lemma and the envelope theorem give $dC_i=\mc_i\,dy_i-\mc_iy_ig_i+\sum_mx_{im}\,dp_m+\sum_fl_{if}\,dw_f$, and after substituting $dp_i=\varphi_{i,y}\,dy_i+(\partial\varphi_i/\partial\tau_i)\,d\tau_i+(dp_i)^{-i}$ the terms in $dy_i$ again vanish by the first-order condition. Part~2 follows from the lemmas as in the proof of Theorem~\ref{thm:main}. The proof of Corollary~\ref{cor:bridge} goes through with $g_k$ in place of $\dlog A_k$, and the last display of part~3 is the form of Lemma~\ref{lem:accounting} in part~1 minus Lemma~\ref{lem:private}. Part~4 follows from the definitions.
\end{proof}

The first term of the gap is the production term of Corollary~\ref{cor:gap}, $E\tilde\lambda_i-\mc_iy_i\ge0$, times the proportional gain in $i$'s output at fixed inputs that its own technology change produces, so it is nonnegative for any technology that raises that output. The second is a general Spence term. It is the value of the improvement to all of $i$'s buyers, final and intermediate, propagated downstream at a fixed allocation, less the revenue from the price increase that $i$ captures at its current output. The innovator also captures $\mc_i(\partial F_i/\partial\tau_m)\,d\tau_m$ of the technology changes of other producers $m$ that enter its own production, such as a better input, which is why the theorem is stated with $\sum_m\T_m-\Pv_i$. With Hicks-neutral productivity and quality on final goods, $\partial F_k/\partial\log A_m$ is $y_k$ for $m=k$ and zero otherwise, $\Pv_i$ is due to $d\tau_i$ alone, and the general Spence term is $\nu_i\,db_i$, so the proposition reduces to the results of Sections~\ref{sec:main} and~\ref{sec:harberger}. Under the general technology of Proposition~\ref{prop:general}, the same steps apply when neither $j$'s own technology nor its output at fixed inputs changes.

\subsection{Small innovators in general markets}\label{oa:smallgeneral}

Corollary~\ref{cor:small} takes the innovator's market to be a CES aggregate and rules out spillovers. Its result on the induced price term holds for any smooth constant-returns aggregate and with spillovers that shrink in proportion to the innovator's size. Pool spillovers, $\dlog A_k=\phi_T\,\mathrm{prox}_{ik}\,dG_i/\mathrm{Pool}_k$ with own effect $\dlog A_i=\phi_R\,dG_i/G_i$ \citep{BloomSchankermanVanReenen2013}, are an example when the innovator's knowledge stock $G_i$ is proportional to its sales.

\begin{proposition}[Small innovators in general markets]\label{prop:smallgeneral}
	Let the assumptions of Corollary~\ref{cor:small} hold, except that $n$ has a constant-returns technology whose unit cost $C_n(\bm p;\delta_i)$ is smooth in input prices and in a weight $\delta_i$ on good $i$, with $\partial C_n/\partial p_i=0$ for all $\bm p$ at $\delta_i=0$ and $\partial^2C_n/\partial p_i\partial\delta_i>0$ at the limit prices, and that the innovation also changes the productivity of each producer $k\neq i$ by $\kappa_k\,\dlog A_i$, with $\kappa_k$ smooth in $\delta_i$ and zero at $\delta_i=0$. Let $\varepsilon_i$ and $e_{im}$ be the own- and cross-price elasticities of $n$'s demand for good $i$ at given output of $n$ and given other prices, with limit $\varepsilon_i^0>1$, and let $n$'s demand system for its other inputs be invertible, up to the normalization of prices, at $\delta_i=0$. Then, whatever $i$'s conduct,
	\begin{equation}
		\begin{split}
			\frac{d\pi_i-\Pv_i}{p_iy_i}={}&\Bigl[\varepsilon_i\Bigl(1-\frac{1}{\mu_i}\Bigr)-1\Bigr](\dlog A_i-\dlog\mu_i)\\
			&+\Bigl(1-\frac{1}{\mu_i}\Bigr)\Bigl[\sum_{m\neq i}e_{im}\,\dlog p_m-(\varepsilon_i-1)\sum_m\widetilde\Omega_{im}\,\dlog p_m+\dlog y_n\Bigr],
		\end{split}
		\label{eq:smallgeneral}
	\end{equation}
	which is equation~\eqref{eq:small} when $n$ is CES. If $i$ competes in quantities or in prices, perceiving the demand for $n$'s output to have an elasticity $\eta_n>0$, then as $\delta_i\to0$, $(d\pi_i-\Pv_i)/(p_iy_i)$ is of order $s_i$ and $\mu_i$ tends to $\varepsilon_i^0/(\varepsilon_i^0-1)$.
\end{proposition}

\begin{proof}
	Producer $n$ buys $x_{ni}=y_n\,\partial C_n/\partial p_i$, so $\dlog(p_iy_i)=\dlog y_n+(1-\varepsilon_i)\,\dlog p_i+\sum_{m\neq i}e_{im}\,\dlog p_m$. Substituting this and $\dlog p_i=\dlog\mu_i-\dlog A_i+\sum_m\widetilde\Omega_{im}\,\dlog p_m$ into the constant-returns form of Corollary~\ref{cor:os}, and using $(\varepsilon_i-1)(1-1/\mu_i)-1/\mu_i=\varepsilon_i(1-1/\mu_i)-1$, gives~\eqref{eq:smallgeneral}. Under CES, $\varepsilon_i=\theta_n(1-s_i)$ and $e_{im}=\theta_ns_m$, and $\sum_ms_m\,\dlog p_m=\dlog p_n$ turns it into~\eqref{eq:small}.

	Because $\partial C_n/\partial p_i=0$ for all $\bm p$ at $\delta_i=0$, the ratio of $\partial C_n/\partial p_i$ to $\delta_i$ is smooth and tends to $\partial^2C_n/\partial p_i\partial\delta_i>0$, so $\varepsilon_i$ and $e_{im}$, which are elasticities of $\partial C_n/\partial p_i$ and hence of this ratio, are smooth in $\delta_i$. Under price competition, $i$ perceives $\dlog y_i=-(\varepsilon_i+\eta_ns_i)\,\dlog p_i$, so its perceived elasticity is $\hat\varepsilon_i=\varepsilon_i+\eta_ns_i$. Under quantity competition, let $E$ be the matrix of elasticities of $n$'s conditional demands. Its rows sum to zero, $E\bm 1=0$, and $s_mE_{mi}=s_iE_{im}$ by the symmetry of the Hessian of $C_n$, so $\bm s^\top E=0$. At given other quantities, the prices relative to $p_n$ solve $E\,\dlog(\bm p/p_n)=(\bm e_i-s_i\bm 1)\,\dlog y_i$, and $p_n$ moves by $-(s_i/\eta_n)\,\dlog y_i$. In the rows $m\neq i$, the coefficients $E_{mi}$ and the right side are of order $s_i$, and $n$'s demand system for its other inputs is invertible, so the other relative prices move by terms of order $s_i$, and the row of $i$ gives $\dlog(p_i/p_n)=-\dlog y_i/\varepsilon_i$ up to a term of order $s_i$. Hence $1/\hat\varepsilon_i=1/\varepsilon_i+O(s_i)$. Under either conduct $\mu_i=\hat\varepsilon_i/(\hat\varepsilon_i-1)$, so $\mu_i-\varepsilon_i/(\varepsilon_i-1)$ and the bracket $\varepsilon_i(1-1/\mu_i)-1=\varepsilon_i/\hat\varepsilon_i-1$ are of order $s_i$.

	The smoothness argument of the proof of Corollary~\ref{cor:small} applies with the dual weight $\delta_i$, except that at $\delta_i=0$ the innovator's markup depends on $p_i$ through $\varepsilon_i^0$, so the nonsingularity of the Jacobian also requires $\zeta_i\neq1$\ifjournal, where $\zeta_i$ is the elasticity of $\varepsilon_i^0/(\varepsilon_i^0-1)$ with respect to $p_i$\fi. At $\delta_i=0$, $n$ does not use good $i$ and $\kappa_k=0$, so only $p_i$ and $\mu_i$ depend on $A_i$, and every other price, $y_n$, and welfare change by terms of order $\delta_i$, hence of order $s_i$, since $s_i/\delta_i$ tends to $p_i(\partial^2C_n/\partial p_i\partial\delta_i)/C_n>0$. The second term of~\eqref{eq:smallgeneral} is therefore of order $s_i$, and so is the first, the bracket times $\dlog A_i-\dlog\mu_i$, which is bounded because $\mu_i$ is smooth in $\delta_i$ and $\log A_i$.
\end{proof}

\subsection{Price competition}\label{oa:bertrand}

\paragraph{Price competition.} Suppose that producer $i$ chooses its price, taking input and factor prices as given, and perceives the demand $y_i=D_i(p_i,b_i)$, with $D_{i,p}<0$, in which the prices and qualities of the other producers selling to its buyers are held fixed and its buyers adjust their purchases along it. Its profit is $\pi_i=p_i\,D_i(p_i,b_i)-C_i\bigl(D_i(p_i,b_i);\bm p,\bm w,A_i\bigr)$, and the first-order condition is
\[
	y_i+(p_i-\mc_i)\,D_{i,p}=0 .
\]

\begin{proof}[Derivation of the results in Remark~\ref{rem:bertrand}]
	Write the change in $i$'s output as $dy_i=D_{i,p}\,dp_i+D_{i,b}\,db_i+(dy_i)^{-i}$, where $(dy_i)^{-i}$ is the part that $i$ does not attribute to its own price or quality. Substituting into the profit differential above gives
	\begin{align*}
		d\pi_i={}&\bigl[y_i+(p_i-\mc_i)\,D_{i,p}\bigr]dp_i+\mc_i\,y_i\,\dlog A_i+(p_i-\mc_i)\,D_{i,b}\,db_i\\
		&+(p_i-\mc_i)(dy_i)^{-i}-\sum_mx_{im}\,dp_m-\sum_fl_{if}\,dw_f .
	\end{align*}
	The bracket is zero by the first-order condition, so $d\pi_i=\Pv_i^B+\IP_i^B$ with
	\begin{align*}
		\Pv_i^B&=\mc_i\,y_i\,\dlog A_i+(p_i-\mc_i)\,D_{i,b}\,db_i,\\
		\IP_i^B&=(p_i-\mc_i)(dy_i)^{-i}-\sum_mx_{im}\,dp_m-\sum_fl_{if}\,dw_f .
	\end{align*}
	Lemma~\ref{lem:allocation} uses price taking by households and in input markets, cost minimization, and market clearing, but not the conduct of producers, and Definition~\ref{def:ae} involves welfare alone. Hence $dW=\sum_k\T_k+\AEff$ under price competition as well, and subtracting $d\pi_i$ gives
	\[
		dW-d\pi_i=(\T_i-\Pv_i^B)+\sum_{j\neq i}\T_j+\AEff-\IP_i^B ,
	\]
	which is Theorem~\ref{thm:main} with these terms. For a process innovation $db_i=0$, and therefore $\Pv_i^B=\Pv_i$.
\end{proof}

Equations~\eqref{eq:ae} and~\eqref{eq:bridge} use only welfare accounting at fixed quantities (Lemma~\ref{lem:accounting} and the proof of Corollary~\ref{cor:aeformula}) and the definitions of $\T_k$ and $\AEff$, so they hold under price competition. Equation~\eqref{eq:harberger} holds under price competition with $\IP_i^B$ in place of $\IP_i$ and with $(p_i-\mc_i)\,D_{i,b}$ in place of $y_i\,\varphi_{i,b}$ in $\nu_i$, because $\mc_i\,y_i\,\dlog A_i+U_{b_i}\,db_i-\Pv_i^B=\bigl[U_{b_i}-(p_i-\mc_i)\,D_{i,b}\bigr]db_i$. In that equation $-\IP_i^B$ cancels the part $(p_i-\mc_i)(dy_i)^{-i}$ of the innovator's term $(p_i-\mc_i)\,dy_i$ in the Harberger term, so under price competition markup-weighted output changes are external only net of $\IP_i^B$.

\subsection{Remarks on the assumptions}

\begin{remark}[Endogenous markups]\label{rem:wedges}
	The result that $\AEff$ is external relies on each producer's output being at its optimum given the demand it perceives, which holds when markups arise from producers' optimization. Suppose instead that prices were set by exogenous markups that producers keep as profit, $p_k=\mu_k\mc_k$, with quantities determined by demand. The first-order condition of Assumption~\ref{ass:conduct} requires the margin $p_i-\mc_i$ to equal $-y_i\varphi_{i,y}$, and an exogenous markup would satisfy it only by coincidence. Otherwise the private return in Lemma~\ref{lem:private} would contain the additional term $(p_i-\mc_i+y_i\varphi_{i,y})\,dy_i$. When the exogenous margin exceeds $-y_i\varphi_{i,y}$, this term is positive for an output expansion, and the innovator would capture part of its margin on that expansion. For a producer whose share of a CES aggregate with elasticity $\theta_n$ goes to zero, equation~\eqref{eq:small} gives this part. With the exogenous markup $\mu_i$, the producer captures the fraction $1-\mu_i/\bigl(\theta_n(\mu_i-1)\bigr)$ of its margin on the expansion that a process innovation causes, which is positive exactly when $\mu_i$ exceeds $\theta_n/(\theta_n-1)$, the markup that it would choose. Endogenous markups are thus the step beyond the exogenous wedges of \citet{BaqaeeFarhi2020}, whose decomposition of welfare holds for any wedges but does not separate a producer's private return from allocative efficiency.
\end{remark}

\begin{remark}[Factor supply]
	\label{rem:factorsupply}
	With elastic factor supplies, Lemma~\ref{lem:accounting} is unchanged, because $w_f=V_f'(L_f)$ values each change in factor supply at its marginal disutility. Lemma~\ref{lem:allocation} holds factor supplies fixed, so $\AEff$ also contains the welfare effect of the change in factor supplies. At a fixed allocation, one more unit of factor $f$ raises the value of consumption by $E\widetilde\Lambda_f/L_f$, its social marginal product, as in the distortion-adjusted Solow residual of \citet{BaqaeeFarhi2020}, and its marginal disutility is $w_f=E\Lambda_f/L_f$. The change in factor supplies therefore contributes $E\sum_f(\widetilde\Lambda_f-\Lambda_f)\,dL_f/L_f$ to $\AEff$, which values it at the gap between social marginal products and prices. Lemma~\ref{lem:private} does not involve factor supplies, and Theorem~\ref{thm:main} is unchanged.
\end{remark}

\begin{remark}[Endogenous technology choice]
	\label{rem:technique}
	Suppose producer $k$ also chooses its technology $\tau_k=(A_k,b_k)$ subject to $R_k(\tau_k)\le z_k$, where $z_k$ is its knowledge, taking its perceived inverse demand and input prices as given. A price change can then induce a change in technique at given knowledge, such as the shift between productivity and quality that a change in the wage induces in the linear--quadratic model of \ifjournal \citet{HopenhaynOkumura2026}\fi. Let $\omega_k$ be the multiplier on the constraint, the private flow value of knowledge. At an interior optimum with a binding constraint, the first-order conditions for $\tau_k$ are $\mc_ky_k=\omega_kA_k\,\partial R_k/\partial A_k$ and $y_k\varphi_{k,b}=\omega_k\,\partial R_k/\partial b_k$, so for the total change $d\tau_k$
	\[
		\Pv_k=\omega_k\Bigl(\frac{\partial R_k}{\partial A_k}\,dA_k+\frac{\partial R_k}{\partial b_k}\,db_k\Bigr)=\omega_k\,dz_k .
	\]
	A change in technique at given knowledge therefore has no private value. Theorem~\ref{thm:main} holds with $d\tau$ the total change, $\Pv_i=\omega_i\,dz_i$, and $\Pv_j=0$ for producers whose knowledge does not change. Their induced technology changes are valued privately at zero and socially at $\T_j$, which is then also their appropriability gap $\T_j-\Pv_j$.
\end{remark}

\begin{remark}[Free entry]
	\label{rem:entry}
	Suppose a group $g$ of symmetric producers of measure $N_g$ sells to one CES nest with elasticity $\theta_g$, each producer pays a fixed cost in factors, and producers enter until each earns zero profit. The other producers take $N_g$ as given, as they take the outputs of other producers in Assumption~\ref{ass:conduct}. Entry and exit move the innovator's profit only through prices, so Lemma~\ref{lem:private} is unchanged and their effect on that profit is part of $\IP_i$. At fixed quantities, a marginal entrant adds output worth $\theta_g/(\theta_g-1)$ times its revenue $R_g$ and uses inputs worth $R_g$, its total cost at zero profit, so the welfare change of Lemma~\ref{lem:accounting} gains the term $R_g\,dN_g/(\theta_g-1)$. The fixed allocation of Lemma~\ref{lem:allocation} holds $N_g$ fixed, so this term is part of $\AEff$, and Theorem~\ref{thm:main} holds with $\AEff$ including entry and exit. Because free entry keeps the group's profits at zero, business stealing from its producers leaves their profits and market values unchanged, whereas the exit that it causes enters $\AEff$ with the negative term $R_g\,dN_g/(\theta_g-1)$.

	Free entry also removes the dependence of the ratio $dW/d\pi_i$ on the innovator's share in Corollary~\ref{cor:fringe}. Suppose that in that economy the fringe producers of the innovator's industry pay a fixed cost in labor and enter freely. The zero-profit condition then holds the industry's price index fixed given its spending and the wage, so the innovation leaves the industry's buyers no better off and lowers the industry's total cost by $d\pi_i$. The social return is therefore $d\pi_i$ times the welfare gain from a unit of labor released to the rest of the economy, per unit of its wage, whatever the innovator's share and conduct. When the other industries keep their sets of producers fixed, they absorb the released labor in proportion, and this multiple is the mean markup $\bar\mu$ of Corollary~\ref{cor:uniform}. Labor hired for R\&D costs society the same multiple of its wage, so with R\&D paid in labor the static terms of the wedge cancel at the first-order condition for R\&D, and only technology spillovers and R\&D responses remain. If every industry has free entry, the multiple rises with the common share $s$, because the other industries are then more concentrated, not because the innovator's own share rises.
\end{remark}

\subsection{Other proofs}\label{oa:otherproofs}

\begin{proof}[Proof of Lemma~\ref{lem:accounting}]
	Using $p_m=U_{c_m}$, $w_f\,dL_f=V_f'\,dL_f$ for every factor (for a factor in inelastic supply because $dL_f=0$), and market clearing,
	\[
		dW=\sum_mp_m\Bigl(dy_m-\sum_kdx_{km}\Bigr)+\sum_kU_{b_k}\,db_k-\sum_fw_f\sum_kdl_{kf}.
	\]
	By Assumption~\ref{ass:inputs}, the change in producer $k$'s input use valued at current prices is
	\[
		dI_k\equiv\sum_mp_m\,dx_{km}+\sum_fw_f\,dl_{kf}=\mc_kA_k\,dG_k=\mc_k\bigl(dy_k-y_k\,\dlog A_k\bigr),
	\]
	where the last step uses $dy_k=y_k\,\dlog A_k+A_k\,dG_k$. The input terms of the first display are $-\sum_kdI_k$. Substituting the expression for $dI_k$ and using $\sum_mp_m\,dy_m-\sum_k\mc_k\,dy_k=\sum_k(p_k-\mc_k)\,dy_k$ gives the result.
\end{proof}

The profit differential at the start of Appendix~\ref{app:proofs} gives the following proofs.

\begin{proof}[Proof of Corollary~\ref{cor:os}]
	By the definitions of $\IP_i$ and $(dp_i)^{-i}$ in Lemma~\ref{lem:private},
	\[
		\IP_i=y_i\,dp_i-\sum_mx_{im}\,dp_m-\sum_fl_{if}\,dw_f-y_i\,\varphi_{i,y}\,dy_i-y_i\,\varphi_{i,b}\,db_i
		=\TR_i-y_i\,\varphi_{i,y}\,dy_i-y_i\,\varphi_{i,b}\,db_i .
	\]
	The first-order condition of Assumption~\ref{ass:conduct} gives $y_i\,\varphi_{i,y}=\mc_i-p_i$, which yields the first expression for $\IP_i$. With constant returns, $C_i=\mc_i\,y_i$ with $\mc_i$ independent of $y_i$, so $\pi_i=(1-1/\mu_i)\,p_iy_i$, and differentiating gives, without any optimality condition,
	\[
		d\pi_i=\Bigl(1-\frac{1}{\mu_i}\Bigr)d(p_iy_i)+\frac{p_iy_i}{\mu_i}\,\dlog\mu_i .
	\]
	By Lemma~\ref{lem:private}, $\IP_i=d\pi_i-\Pv_i$. For a process innovation $db_i=0$, so $\Pv_i=\mc_i\,y_i\,\dlog A_i=p_iy_i\,\dlog A_i/\mu_i$. Substituting this and $d(p_iy_i)=p_iy_i\,\dlog(p_iy_i)$ gives the expression for a process innovation.
\end{proof}

\begin{proof}[Proof of Corollary~\ref{cor:bridge}]
	Subtracting $\sum_k\T_k=\sum_k\bigl(E\tilde\lambda_k\,\dlog A_k+U_{b_k}\,db_k\bigr)$ from Lemma~\ref{lem:accounting} gives
	\[
		\AEff=\sum_k(\mc_k\,y_k-E\tilde\lambda_k)\,\dlog A_k+\sum_k(p_k-\mc_k)\,dy_k ,
	\]
	which rearranges to \eqref{eq:bridge}. Finally, subtracting $d\pi_i=\Pv_i+\IP_i$ (Lemma~\ref{lem:private}) from Lemma~\ref{lem:accounting} gives
	\begin{align*}
		dW-d\pi_i={}&\sum_k(p_k-\mc_k)\,dy_k+\sum_{j\neq i}\bigl(\mc_j\,y_j\,\dlog A_j+U_{b_j}\,db_j\bigr)\\
		&+\bigl(\mc_i\,y_i\,\dlog A_i+U_{b_i}\,db_i-\Pv_i\bigr)-\IP_i ,
	\end{align*}
	and $\mc_i\,y_i\,\dlog A_i+U_{b_i}\,db_i-\Pv_i=(U_{b_i}-y_i\,\varphi_{i,b})\,db_i=\nu_i\,db_i$, which gives \eqref{eq:harberger}.
\end{proof}

\section{Example Economies}\label{oa:examples}

The values reported in the text are computed from equilibria of the economies below, re-solved after each change in technology, or from closed forms checked against such equilibria.

\subsection{A three-producer economy}\label{oa:roundabout}

Three producers, $k=1,2,3$, compete in quantities. Producer $k$ produces $y_k=A_k\,l_k^{\alpha_k}M_k^{1-\alpha_k}$ from labor $l_k$ and a quantity $M_k$ of an aggregate good. A competitive aggregator produces the aggregate good,
\[
	Y=\Bigl(\sum_k(b_ky_k)^{(\theta-1)/\theta}\Bigr)^{\theta/(\theta-1)},
	\qquad
	\theta=\IWTheta,
\]
where $b_k$ is the quality of good $k$ in efficiency units. The aggregate good is consumed and is the only intermediate input, so $Y=C+\sum_kM_k$. It is the numeraire, and welfare equals consumption, $W=C$. Labor is supplied inelastically, with $\sum_kl_k=1$. Producers take the wage and the price of the aggregate good as given. With the other producers' outputs fixed, producer $k$'s inverse demand $p_k=b_k^{(\theta-1)/\theta}(Y/y_k)^{1/\theta}$ has elasticity $-(1-s_k)/\theta$, where $s_k=p_ky_k/Y$ is its sales share, so its markup solves $1/\mu_k=1-(1-s_k)/\theta$. The output elasticities of labor are $(\alpha_1,\alpha_2,\alpha_3)=(0.6,0.5,0.7)$, baseline productivities are $(A_1,A_2,A_3)=(1,1.3,0.8)$, and baseline qualities are $(b_1,b_2,b_3)=(1,1.1,0.9)$.

\paragraph{Business stealing and the cost of research labor.} In this economy, let producer~1's productivity alone rise, $\dlog\bm A=(1,0,0)$, so that Proposition~\ref{prop:bs} applies to both rivals. Producer~2's output falls, and its markup-weighted output change, the welfare effect of business stealing from it, is \IWXBSMarkupTwo. Its profit nevertheless rises by \IWXBSProfitTwo, because the transfer $\TR_2$, \IWXBSTransferTwo, more than offsets this loss. In this economy a larger producer faces a more elastic demand and charges a lower markup (Proposition~\ref{prop:markups} with $\eta=\infty$), so producer~2 raises its markup as it loses share, and its buyers pay the higher margin on its remaining output. Where markups instead rise with the share, as in \citet{AtkesonBurstein2008}, a rival that loses business lowers its margin, and the transfer adds to its loss of profit (Section~\ref{sec:app-bs}). Measures of business stealing based on rivals' profits or market values can therefore differ even in sign from its welfare effect. For producer~3, the markup-weighted output change is \IWXBSMarkupThree{} and the transfer $\TR_3$ is \IWXBSTransferThree, so its profit changes by \IWXBSProfitThree. Its profit falls, but by less than the welfare cost of its lost output. In the experiment of Proposition~\ref{prop:factors}, an amount $dL^R_L$ of labor is withdrawn from production with no change in technology. With a single factor, the cost-based share of labor is $\widetilde\Lambda_L=\sum_k\tilde\lambda_k\alpha_k=1$, so the ratio $\widetilde\Lambda_L/\Lambda_L$ is the inverse of labor's income share, $E/w=\IWXRDRatio$.

\section{The 2015 Economy}\label{oa:application}\label{app:application}%

This section describes the economy of Section~\ref{sec:application} and its alternatives, the computation of the returns, the technology spillovers, and the output-based returns, and Table~\ref{tab:oa-robustness} reports further specifications.

\paragraph{Industries.} The input--output data are the BEA summary use table for \IWAppYear{} before redefinitions, with 66 industries after dropping the government, scrap, and noncomparable-import industries, as in \citet{BaqaeeFarhi2020}. Each industry's cost shares follow their calibration, and so do the elasticities of substitution, \IWAppEpsilonVA{} between value added and intermediate inputs, \IWAppThetaInt{} across intermediate inputs, \IWAppEtaKL{} between labor and capital, and \IWAppSigmaFinal{} across industries in final demand. Labor's share comes from compensation, capital's from gross operating surplus less the profits implied by the industry's markup, and the shares of intermediate inputs from the use table. Where the profits implied by the markup exceed gross operating surplus, capital's share is zero and the shares of intermediate inputs are scaled down so that the cost shares sum to one, as in their calibration. This is the case in \IWAppClipUserCost{} and \IWAppClipAccounting{} of the 66 industries with the user-cost and accounting-profit markups, in \IWAppSegClip{} with the markups of the baseline, and in \IWAppClipBase{} industries with \IWAppClipBaseOutput{} percent of gross output in the BEA economy with $\theta=\IWAppThetaBase$. An industry's markup is the sales-weighted harmonic mean of the markups of its producers and fringes. Industries without sample firms get the output-weighted mean markup of the other industries.

\paragraph{Firms.} The firms are the Compustat firms in the \IWAppYear{} sample of \citet{BaqaeeFarhi2020}, which excludes agriculture, construction, general merchandise stores, utilities, finance, and nonclassifiable firms, here without the requirement that a markup be measured. In the baseline, \IWAppSegDropped{} records whose sales another record of the sample contains, dual listings and consolidated subsidiaries, are dropped.

\paragraph{Markets and producers.} A market is a six-digit NAICS code of the 2017 Economic Census, in the BEA industry to which the code belongs. Codes that changed between the 2012 and 2017 classifications are mapped by hand, and the 2017 codes that split a 2012 code form one market. Each business segment of a firm in the Compustat segment files for fiscal \IWAppYear{} is a producer in the market of its NAICS code, and the segments of a firm in one market form one producer. Where a firm's largest segment has no six-digit code, the firm's own code is used. A shorter code merges the markets within it into one where the firm would hold more than a tenth of the largest of them, and otherwise places the firm in that market. Segments whose code belongs to no industry of the network, or to an industry without firms of the sample, are left out, \IWAppSegLeftOut{} segments with \IWAppSegLeftOutSales{} percent of the producers' domestic sales. A firm's innovation and its R\&D are in its largest producer, and \IWAppSegMultiFirms{} firms have more than one producer.

A producer's domestic sales are its segment sales times the firm's U.S. share, the share of the firm's geographic segment sales in the United States. Where the domestic segment covers North America, the Americas, or the United States and Canada, it is scaled by the U.S. share of the region's \IWAppYear{} GDP, \IWAppSegFactorNA{}, \IWAppSegFactorAM{}, and \IWAppSegFactorUSCA{}. U.S. firms without geographic segments sell only domestically. Foreign firms without geographic segments get the median U.S. share of the foreign firms of their industry, and foreign firms whose geographic segments name no U.S. region get zero, so \IWAppSegNoUSFirms{} firms with \IWAppSegNoUSSales{} percent of the sample's sales leave the economy. Domestic sales are \IWAppSegDomesticShare{} percent of the sample's sales. A market's size is its census receipts scaled to \IWAppYear{} by its industry's gross output and to domestic absorption by the industry's commodity output less exports plus imports, relative to commodity output, so that it includes the U.S. sales of foreign producers. Receipts withheld for a six-digit code are the five-digit total less the published codes within it. A producer's share is its domestic sales divided by the larger of its market's size and the domestic sales of the market's producers. The \IWAppSegFringeless{} markets in which the second is larger have no fringe. They hold \IWAppSegFringelessSales{} percent of the producers' domestic sales, and the innovators whose largest producer sells in them account for \IWAppSalesShareNoFringeSeg{} percent of the innovators' sales and \IWAppRDShareNoFringeSeg{} percent of their R\&D.

A firm's producers in different markets of one industry are separate competitors, so the firm does not internalize the effect of one producer's output on the prices of its others. The private return is the change in the profits of all of the firm's producers, and the profit changes of the producers other than the innovator are below a thousandth of the private return in total. The two elasticities $\theta$ and $\sigma$ solve two equations, the aggregate markup equal to the accounting-profit value of \citet{BaqaeeFarhi2020}, \IWAppFactorPremiumAccounting{}, and the sales-weighted own-cost pass-through equal to that of the BEA economy, \IWAppPassOwn{}.

\paragraph{The BEA economy.} The alternative economy assigns each firm to one BEA industry by its NAICS code. The code is read at four digits where BEA industries divide a three-digit code, which separates motor vehicles from other transportation equipment and legal services and computer systems design from other professional services, and at three digits within retail trade. A firm's share is its worldwide sales divided by its industry's gross output. In \IWAppFullCoverage{} industries, which account for \IWAppFullCoverageSales{} percent of the sample's sales and \IWAppFullCoverageRD{} percent of its R\&D, the sales of the Compustat firms exceed the industry's gross output, by a median factor of \IWAppFullCoverageRatioMedian{}. Shares there are relative to the sum of these sales, and the industry has no fringe. Firms perceive a unit elasticity of their industry's demand, as in Online Appendix~G of \citet{BaqaeeFarhi2020}, and $\theta=\IWAppThetaBase$ gives the aggregate markup \IWAppFactorPremium{}. Their value, $\theta=\IWAppThetaAlt$, raises the aggregate markup to \IWAppFactorPremiumAlt{}, and the profits it implies exceed gross operating surplus in \IWAppClipAlt{} industries with \IWAppClipAltOutput{} percent of gross output.

The ratio of Section~\ref{sec:app-ae} uses the firms, industries, and measured markups of the structural sample of \citet{BaqaeeFarhi2020}. Innovators with a measured markup below one, which Assumption~\ref{ass:conduct} rules out, are excluded. They account for \IWAppExcludedUserCost{}\ifjournal{} and \IWAppExcludedAccounting{} percent of sales under the user-cost and accounting-profit markups.\fi

\paragraph{Markups and their responses.} In a market with elasticity $\theta$ whose producers perceive the elasticity $\eta$ of its demand, $\eta=\sigma$ in the baseline and $\eta=1$ in the BEA economy, Proposition~\ref{prop:markups} gives $\mu_j=(1-1/\varepsilon_j)^{-1}$ with $1/\varepsilon_j=(1-s_j)/\theta+s_j/\eta$, and $\zeta_j=\mu_js_j(1/\eta-1/\theta)$ for each producer $j$. The producers and the fringe share their industry's technology, so their prices move relative to each other only through their markups and productivities. With $u_j=\dlog\mu_j-\dlog A_j$, and $u=0$ for the fringe,
\[
	\dlog s_j=(1-\theta)\Bigl(u_j-\sum_ms_mu_m\Bigr),
\]
where the sum runs over the market's producers. With $\dlog\mu_j=\zeta_j\,\dlog s_j$, the markup responses $\dlog\bm\mu=X\,\dlog\bm A$ within the market solve $(I-M)X=-M$, with $M=(1-\theta)\,Z\,(I-\bm 1\bm s^\top)$ and $Z=\operatorname{diag}(\zeta_j)$. With a single strategic producer, $X$ reduces to the response in the proof of Corollary~\ref{cor:onecournot}. When only its own markup responds, a producer passes $1/\bigl(1+(\theta-1)(1-s_j)\zeta_j\bigr)$ of a cost reduction into its price. Markups in other markets do not respond, because a productivity change in one market leaves relative prices unchanged within every other market.

When the producers of a market compete in prices \citep{AtkesonBurstein2008}, a producer with share $s_i$ perceives the elasticity $\varepsilon_i=\theta(1-s_i)+\eta s_i$, charges $\mu_i=\varepsilon_i/(\varepsilon_i-1)$, and its markup responds to its share with the elasticity $\zeta_i=\mu_is_i(\theta-\eta)/\varepsilon_i^2$ in the fixed point of Proposition~\ref{prop:markups}. A producer that sets its price perceives its true demand when its rivals' prices are fixed, so its own price change leaves nothing in its induced price term, which contains only its rivals' price responses and general-equilibrium changes, and is therefore small.

\paragraph{Returns.} For a unit increase in the log productivity of firm $i$'s largest producer, the responses of Domar weights, factor shares, and real output solve the propagation equations of \citet{BaqaeeFarhi2020}. The network represents the innovator's market by its producers and its fringe, each other market of its industry by a single producer, and every other industry by a single producer. Collapsing the other markets is exact at first order, because their within-market shares, and so their markups, do not move. With final demand as the numeraire, $\dlog E=\dlog Y$, so $dW=E\,\dlog Y$. With constant returns, Corollary~\ref{cor:os} gives
\[
	d\pi_j=E\lambda_j\Bigl[\Bigl(1-\frac1{\mu_j}\Bigr)\bigl(\dlog\lambda_j+\dlog Y\bigr)+\frac{\dlog\mu_j}{\mu_j}\Bigr]
\]
for each producer of the innovator's market. A producer in another market keeps its share and markup, so its profit changes by $E\lambda_j(1-1/\mu_j)(\dlog\lambda_{m(j)}+\dlog Y)$, where $\lambda_{m(j)}$ is the Domar weight of its market. A rival's lost output valued at its markup is $(p_j-\mc_j)\,dy_j=E\lambda_j(1-1/\mu_j)\,\dlog y_j$, with $\dlog y_j=\dlog\lambda_j+\dlog Y-\dlog p_j$. The ratios divide the sum of $dW$ over innovators by the sum of $d\pi_i$.

\paragraph{Free entry.} When the fringe of the innovator's market enters freely, the market's price index is fixed given its spending and input prices (Remark~\ref{rem:entry}). Holding these fixed, as for a market that is small relative to the economy, the other producers of the market keep their shares and markups, and the innovator keeps its markup rule and faces a demand of elasticity $\theta$, so that
\[
	d\pi_i=\Pv_i\,\frac{(\theta-1)(\mu_i-1+\zeta_i)}{1+(\theta-1)\zeta_i},
\]
which with $\eta=1$ is $\Pv_i\bigl(1+2(\theta-1)s_i\bigr)/\bigl(1+(\theta-2)s_i\bigr)$. The market's output and price do not change, and its total cost falls by $d\pi_i$, so $dW$ is $d\pi_i$ times the welfare gain per dollar of the industry's input bundle released to the rest of the economy. This gain is the response of the propagation equations to equal changes in the productivity and the markup of the industry, which leave its price and output unchanged and lower its cost. The other markets keep their fringes fixed, and the baseline flows are those of the fixed-fringe economy. Free entry is defined only where the innovator's market has a fringe.

\paragraph{Further specifications.} Table~\ref{tab:oa-robustness} reports variants of the baseline and of the BEA economy, and its notes define their calibration, the unit elasticities, and free entry. Main segment only keeps each firm's largest producer and leaves the sales of its other segments in the fringes of their markets. All sales in the main market places a firm's domestic sales in the market of its largest producer, and firm-level codes place them in the market of the firm's own code. Receipts without trade adjustment omit the absorption adjustment of market size, and unscaled receipts omit the scaling to \IWAppYear{}. The next rows set the regional factors to one, give foreign firms without a U.S. region their industry's median share or foreign firms without geographic segments zero, and raise the threshold for merging markets. In the BEA economy, the perceived elasticity of the row with $\eta$ from measured markups makes the slope of $1/\mu_i$ on $s_i$ within industries, $-(1/\eta-1/\theta)$ by Proposition~\ref{prop:markups}, equal to that of the user-cost markups, \IWAppSlopeUser{}, estimated with industry fixed effects across the \IWAppSlopeFirms{} firms of the sample with a user-cost markup between 0.5 and 5, against \IWAppSlopeCournot{} with $\eta=1$. Weighted by sales, the slope is \IWAppSlopeUserWeighted{}, which implies $\eta=\IWAppEtaFlat$, so that markups barely rise with share.{} The last row drops the firms incorporated outside the United States, \IWAppForeignFirms{} of the sample with \IWAppForeignSales{} percent of its sales, whose sales become part of their industries' fringes.

\paragraph{Measured markups.} Within six-digit markets, the measured markups of \citet{BaqaeeFarhi2020} also rise with market share, but less steeply than those of the baseline. With market fixed effects, across firms with a measured markup between 0.5 and 5, the slope of $1/\mu_i$ on $s_i$ is \IWAppSegMSSlopeUser{} unweighted and \IWAppSegMSSlopeUserWeighted{} weighted by sales for the user-cost markups and \IWAppSegMSSlopeAccounting{} and \IWAppSegMSSlopeAccountingWeighted{} for the accounting-profit markups, which by Proposition~\ref{prop:markups} imply values of $\sigma$ from \IWAppSegMSSigmaAccounting{} to \IWAppSegMSSigmaUserWeighted{} with the baseline $\theta$, against \IWAppSegEta{}. Across firms, the correlation of the measured markups with the markups that conduct implies is only \IWAppSegMSCorrAccounting{} to \IWAppSegMSCorrUserWeighted{}, so conduct explains little of their dispersion. In the economies with $\sigma$ from the user-cost slopes (panel~C of Table~\ref{tab:oa-robustness}), small innovators' change in allocative efficiency remains negative, and the static return falls short of the social cost for nearly all of their R\&D.

\begin{table}[!t]
	\centering
	\scriptsize
	\setlength{\tabcolsep}{2pt}
	\caption{Robustness of the returns to process innovation, 2015}
	\label{tab:oa-robustness}
	\begin{tabular}{@{}lccccccccccccc@{}}
		\toprule
		& & & & Pass- & & & & & \multicolumn{2}{c}{R\&D-weighted} & \multicolumn{3}{c}{Small innovators} \\ \cmidrule(lr){10-11}\cmidrule(lr){12-14}
		& $\theta$ & $\eta$ & Markup & through & $dW/d\pi_i$ & Gap & AE & $-$IP & AE & Below & AE & IP & Below \\
		\midrule
		\multicolumn{14}{@{}l}{\emph{A. Six-digit markets with segments, quantity competition}} \\
		Baseline & 13.8 & 1.7 & 1.20 & 0.62 & 1.14 & 0.27 & 0.04 & \ensuremath{-}0.17 & \ensuremath{-}0.04 & 80 & \ensuremath{-}0.16 & 0.026 & 100 \\
		$\eta$ lower by 0.5 & 16.1 & 1.2 & 1.20 & 0.50 & 1.13 & 0.29 & 0.05 & \ensuremath{-}0.21 & \ensuremath{-}0.10 & 77 & \ensuremath{-}0.18 & 0.047 & 95 \\
		$\eta$ higher by 0.5 & 12.7 & 2.2 & 1.20 & 0.69 & 1.15 & 0.26 & 0.03 & \ensuremath{-}0.14 & \ensuremath{-}0.04 & 82 & \ensuremath{-}0.15 & 0.017 & 100 \\
		Main segment only & 13.5 & 1.8 & 1.20 & 0.62 & 1.16 & 0.28 & 0.05 & \ensuremath{-}0.17 & \ensuremath{-}0.02 & 79 & \ensuremath{-}0.14 & 0.026 & 100 \\
		All sales in the main market & 14.2 & 2.0 & 1.20 & 0.62 & 1.17 & 0.28 & 0.05 & \ensuremath{-}0.16 & \ensuremath{-}0.04 & 75 & \ensuremath{-}0.17 & 0.024 & 100 \\
		Firm-level codes & 13.9 & 1.8 & 1.20 & 0.62 & 1.15 & 0.27 & 0.04 & \ensuremath{-}0.16 & \ensuremath{-}0.06 & 81 & \ensuremath{-}0.18 & 0.026 & 99 \\
		Receipts, no trade adjustment & 13.8 & 1.7 & 1.20 & 0.62 & 1.13 & 0.27 & 0.03 & \ensuremath{-}0.17 & \ensuremath{-}0.06 & 82 & \ensuremath{-}0.18 & 0.028 & 100 \\
		Receipts of 2017, unscaled & 13.6 & 1.7 & 1.20 & 0.62 & 1.14 & 0.27 & 0.04 & \ensuremath{-}0.17 & \ensuremath{-}0.04 & 81 & \ensuremath{-}0.15 & 0.027 & 100 \\
		Region factors of 1 & 13.8 & 1.8 & 1.20 & 0.62 & 1.14 & 0.27 & 0.04 & \ensuremath{-}0.17 & \ensuremath{-}0.04 & 81 & \ensuremath{-}0.16 & 0.027 & 95 \\
		Foreign firms, industry median & 14.1 & 1.7 & 1.20 & 0.62 & 1.13 & 0.27 & 0.04 & \ensuremath{-}0.18 & \ensuremath{-}0.05 & 81 & \ensuremath{-}0.19 & 0.028 & 100 \\
		Foreign without segments, zero & 13.7 & 1.8 & 1.20 & 0.62 & 1.14 & 0.27 & 0.04 & \ensuremath{-}0.17 & \ensuremath{-}0.04 & 81 & \ensuremath{-}0.16 & 0.026 & 100 \\
		Merge threshold 0.11 & 13.7 & 1.8 & 1.20 & 0.62 & 1.14 & 0.27 & 0.04 & \ensuremath{-}0.17 & \ensuremath{-}0.04 & 80 & \ensuremath{-}0.16 & 0.026 & 100 \\
		Unit elasticity, VA and inputs & 13.8 & 1.7 & 1.20 & 0.62 & 1.18 & 0.27 & 0.08 & \ensuremath{-}0.17 & 0.02 & 70 & \ensuremath{-}0.10 & 0.026 & 93 \\
		Unit elasticities & 13.8 & 1.7 & 1.20 & 0.62 & 1.20 & 0.27 & 0.10 & \ensuremath{-}0.17 & 0.05 & 68 & \ensuremath{-}0.09 & 0.027 & 91 \\
		Free entry & 13.8 & 1.7 & 1.20 & -- & 1.07 & 0.22 & 0.13 & \ensuremath{-}0.28 & 0.09 & 100 & \ensuremath{-}0.08 & 0.035 & 100 \\
		\addlinespace \multicolumn{14}{@{}l}{\emph{B. BEA industries, quantity competition}} \\
		$\theta=16$ & 16.0 & 1.0 & 1.20 & 0.62 & 0.99 & 0.23 & \ensuremath{-}0.03 & \ensuremath{-}0.21 & \ensuremath{-}0.09 & 90 & \ensuremath{-}0.31 & 0.054 & 99 \\
		$\theta=8$ & 8.0 & 1.0 & 1.37 & 0.74 & 1.19 & 0.42 & \ensuremath{-}0.11 & \ensuremath{-}0.13 & \ensuremath{-}0.14 & 88 & \ensuremath{-}0.26 & 0.024 & 99 \\
		$\eta$ from measured markups & 16.0 & 1.8 & 1.17 & 0.73 & 1.00 & 0.18 & \ensuremath{-}0.02 & \ensuremath{-}0.15 & \ensuremath{-}0.07 & 90 & \ensuremath{-}0.23 & 0.028 & 99 \\
		$\eta$ from sales-weighted slope & 16.0 & 10.2 & 1.13 & 0.97 & 1.15 & 0.15 & \ensuremath{-}0.04 & 0.04 & \ensuremath{-}0.05 & 47 & \ensuremath{-}0.06 & \ensuremath{-}0.002 & 97 \\
		Unit elasticities & 16.0 & 1.0 & 1.20 & 0.62 & 1.05 & 0.22 & 0.04 & \ensuremath{-}0.21 & 0.00 & 87 & \ensuremath{-}0.22 & 0.054 & 99 \\
		Free entry & 16.0 & 1.0 & 1.20 & -- & 1.08 & 0.19 & 0.15 & \ensuremath{-}0.26 & 0.18 & 100 & \ensuremath{-}0.02 & 0.066 & 100 \\
		U.S.-incorporated firms & 16.0 & 1.0 & 1.21 & 0.58 & 1.05 & 0.24 & 0.02 & \ensuremath{-}0.21 & \ensuremath{-}0.06 & 78 & \ensuremath{-}0.29 & 0.053 & 100 \\
		\addlinespace \multicolumn{14}{@{}l}{\emph{C. Six-digit markets with segments, further checks}} \\
		User-cost markup level & 31.2 & 3.7 & 1.08 & 0.62 & 1.01 & 0.10 & 0.11 & \ensuremath{-}0.20 & 0.01 & 83 & \ensuremath{-}0.14 & 0.029 & 100 \\
		$\sigma$ from markup slope & 13.8 & 2.3 & 1.19 & 0.68 & 1.13 & 0.24 & 0.04 & \ensuremath{-}0.15 & \ensuremath{-}0.04 & 82 & \ensuremath{-}0.15 & 0.019 & 100 \\
		$\sigma$ from weighted slope & 13.8 & 4.1 & 1.17 & 0.81 & 1.15 & 0.20 & 0.03 & \ensuremath{-}0.07 & \ensuremath{-}0.03 & 80 & \ensuremath{-}0.11 & 0.007 & 100 \\
		Markets with a fringe & 13.8 & 1.7 & 1.20 & 0.66 & 1.17 & 0.25 & 0.09 & \ensuremath{-}0.17 & 0.02 & 79 & \ensuremath{-}0.11 & 0.027 & 100 \\
		\addlinespace
		Price competition & 11.2 & 1.7 & 1.20 & 0.87 & 1.26 & 0.26 & \ensuremath{-}0.04 & 0.05 & \ensuremath{-}0.07 & 46 & \ensuremath{-}0.09 & \ensuremath{-}0.001 & 99 \\
		User-cost markup level & 25.1 & 3.7 & 1.08 & 0.87 & 1.15 & 0.10 & 0.01 & 0.03 & \ensuremath{-}0.02 & 46 & \ensuremath{-}0.05 & \ensuremath{-}0.001 & 100 \\
		Markets with a fringe & 11.2 & 1.7 & 1.20 & 0.89 & 1.26 & 0.25 & \ensuremath{-}0.03 & 0.03 & \ensuremath{-}0.04 & 47 & \ensuremath{-}0.08 & \ensuremath{-}0.001 & 99 \\
		\bottomrule
	\end{tabular}

	\vspace{0.3em}
	\parbox{\linewidth}{\footnotesize Notes: Per unit of the private return, without spillovers, with all markups of the innovator's market responding; ratios pool innovations. Markup: the aggregate markup, which is also the factor premium. Pass-through: own-cost pass-through with only the innovator's markup responding. Below: percentage of R\&D whose static social return falls short of the factor premium. Panel A: each data variant recalibrates $\theta$ and $\eta=\sigma$ to the aggregate markup 1.20 and the pass-through 0.62; the rows that move $\eta$ recalibrate $\theta$ only. Unit elasticities: the elasticities of the network outside markets (between value added and intermediates, across intermediates, between labor and capital, and in final demand) set to one, keeping $\sigma=\eta$ across the markets of an industry and the baseline calibration; the row above it sets only the elasticity between value added and intermediates to one. Free entry applies to innovators whose market (panel A) or industry (panel B) has a fringe. Small innovators hold less than 1 percent of their market (panels A and C) or BEA industry (panel B). Panel C: the first four rows with quantity competition, the last three with price competition. User-cost markup level: $\theta$ and $\eta=\sigma$ recalibrated to the aggregate markup of the user-cost markups of \citet{BaqaeeFarhi2020}, 1.08, and the pass-through 0.62 ($\theta$ alone with price competition). $\sigma$ from markup slope: $1/\sigma=1/\theta$ minus the within-market slope of the inverse user-cost markup on market share, unweighted or weighted by sales, with the baseline $\theta$. Markets with a fringe: innovators whose market has a fringe, in the baseline economy. -- marks values not defined.}
\end{table}

\paragraph{Technology spillovers.} Let $\varpi_{jk}$ be the proximity of firms $j$ and $k$ of \citet{Jaffe1986}, the uncentered correlation of the shares of their patents across technology classes, and let $\mathcal R_k$ be firm $k$'s R\&D stock. In the production function of \citet{BloomSchankermanVanReenen2013}, a firm's sales deflated by its industry's price index have elasticity \IWAppPhiR{} with respect to the firm's R\&D stock and \IWAppPhiT{} with respect to its spillover pool $\sum_{k\neq j}\varpi_{jk}\mathcal R_k$, the values they use to compute returns, estimated with controls for the industry's output. With the industry's output fixed, a firm in a market with elasticity $\theta$ raises its deflated sales by $1-1/\theta$ times its output growth \citep{KletteGriliches1996}, so the corresponding elasticities of its productivity are these values divided by $1-1/\theta$. A change in $\mathcal R_i$ that raises $A_i$ by $\dlog A_i=\IWAppPhiR\,\dlog\mathcal R_i/(1-1/\theta)$ therefore raises $A_j$ by $\dlog A_j=\upsilon_{ji}\,\dlog A_i$, with
\[
	\upsilon_{ji}=\frac{\IWAppPhiT}{\IWAppPhiR}\,\frac{\varpi_{ji}\mathcal R_i}{\sum_{k\neq j}\varpi_{jk}\mathcal R_k}.
\]
The factor $1-1/\theta$ is common to all firms and cancels from $\upsilon_{ji}$. The patents are those granted from \IWAppWindowStart{} to \IWAppWindowEnd{}, linked to Compustat firms by \citet{AroraEtAl2024}, with their CPC main groups as technology classes. Firms with at least \IWAppMinRecords{} patent-class records enter, \IWAppPatentingFirms{} firms over \IWAppCPCGroups{} classes. R\&D stocks cumulate Compustat R\&D expenditures with depreciation \IWAppDepreciation{}, starting from the first year's expenditure divided by \IWAppDepreciation{} plus a pre-sample growth rate of \IWAppPreGrowth{}. The pool of each firm sums over all these firms, including those outside the sample. A technology change at a firm is a change at its largest producer.

Because the returns are first order, the returns with spillovers are sums of unit responses. With $W_j$ and $\Pi_{ij}$ the effects on welfare and on firm $i$'s profit of a unit increase in $A_j$ alone, each computed in the network that represents $j$'s market producer by producer,
\[
	dW_i=W_i+\sum_{j\neq i}\upsilon_{ji}W_j,
	\qquad
	d\pi_i=\Pi_{ii}+\sum_{j\neq i}\upsilon_{ji}\Pi_{ij}.
\]
A firm $i$ outside $j$'s market keeps its shares and markups, so
\[
	\Pi_{ij}=\sum_kE\lambda_k\Bigl(1-\frac1{\mu_k}\Bigr)\bigl(\dlog\lambda_{m(k)}+\dlog Y\bigr),
\]
where the sum runs over $i$'s producers $k$.{}

\paragraph{Output-based returns.} Following \citet{BloomSchankermanVanReenen2013}, the output-based social and private returns of a unit increase in $A_i$ are
\[
	\widehat{dW}_i=E\Bigl(\lambda_i+\sum_{j\neq i}\upsilon_{ji}\lambda_j\Bigr),
	\qquad
	\widehat{d\pi}_i=E\lambda_i-\frac12\sum_{j\in\mathcal N_i}E\lambda_j\,\dlog\pi_j=E\lambda_i-\frac12\sum_{j\in\mathcal N_i}\frac{\mu_j}{\mu_j-1}\,d\pi_j,
\]
where $\mathcal N_i$ collects $i$'s rivals in its market, the fringe included, and $d\pi_j$ is the change in $j$'s profit caused by $i$'s productivity change alone. The last term is the output that \citet{BloomSchankermanVanReenen2013} credit to the innovator, half of each rival's percentage change in market value times the rival's sales, with market values proportional to profits. With constant returns, a rival's sales are $\mu_j/(\mu_j-1)$ times its profit. The credit includes the fringe, because the output that the innovator takes from it is part of the output gain that this return measures, and the welfare-based returns count business stealing from every rival. {} By Theorem~\ref{thm:main}, with $\upsilon_{ii}=1$,
\begin{align*}
	\widehat{dW}_i-dW_i&=-E\sum_j\upsilon_{ji}\bigl(\tilde\lambda_j-\lambda_j\bigr)-\AEff,\\
	\widehat{d\pi}_i-d\pi_i&=E\lambda_i\Bigl(1-\frac1{\mu_i}\Bigr)-\IP_i-\frac12\sum_{j\in\mathcal N_i}\frac{\mu_j}{\mu_j-1}\,d\pi_j .
\end{align*}
{} In the estimates of \citet{BloomSchankermanVanReenen2013}, the credit is \IWAppBSVCreditBSV{} times the gain in the innovator's sales deflated by its industry's price index, half their market-value coefficient on rivals' R\&D divided by the elasticity of these sales with respect to own R\&D. Because their estimates control for industry output, these sales rise by $1-1/\theta$ times the output gain at its price, so per unit of the output gain their credit is \IWAppSegBSVCreditBSVPrice{}.\ifjournal{} In the baseline, the credit that the model implies is \IWAppSegBSVCreditModel{} per unit of the output gain, and with it the output-based ratio is \IWAppSegBSVRatio{}. The welfare-based ratio is\fi{} \IWAppSegBSVWelfareOverOutputPrice{} times the output-based ratio with the credit of \citet{BloomSchankermanVanReenen2013}. One plus that credit contributes a factor of \IWAppSegBSVFactorCreditPrice{}, the innovator's output gain at its price relative to its profit change a factor of \IWAppSegBSVFactorOwnOutput{}, and the welfare gain relative to the output-based social return a factor of \IWAppSegBSVFactorSocial{}.{} \citet{BloomSchankermanVanReenen2013} compute their returns for an increase in R\&D that includes the R\&D it induces other firms to undertake. The static returns here omit these responses in both accountings.

\end{document}